\documentclass[acmsmall, screen, balance=false]{acmart}

\usepackage{jmsdelim}
\usepackage{mathpartir}
\usepackage{macros}
\usepackage[inline]{enumitem}
\usepackage{todonotes}
\usepackage{float}
\usepackage{stmaryrd}
\usepackage{tikz-cd}
\usepackage{tikz}
  \usetikzlibrary{decorations.pathmorphing}

\newtheorem{remark}{Remark}

\newcommand{\peq}{=}

\newcommand{\equi}{\simeq}
\newcommand{\Univ}[1][]{\mathsf{U}_{#1}}
\renewcommand{\Prop}[1][]{\mathsf{Prop}_{#1}}
\newcommand{\Set}[1][]{\mathsf{Set}_{#1}}

\newcommand\tabs[2]{\lambda (#1\! :\! #2).}

\newcommand\tapp[2][\tickA]{#2\,[#1] }

\newcommand{\tickA}{\alpha}
\newcommand{\tickB}{\beta}
\newcommand\latbind[2]{{\triangleright}\, (#1 \!: \!#2) .}
\newcommand\toksubst[3][\kappa]{\left[#2/#3\right]}
\newcommand{\later}{\triangleright}
\newcommand{\dfix}{\mathsf{dfix}}
\newcommand{\fix}{\mathsf{fix}}
\newcommand{\clockc}{\kappa_0}

\newcommand{\grD}[1][\kappa]{{L^{#1}}}
\newcommand{\ciD}{{L}}

\newcommand{\inr}{\mathsf{inr}}
\newcommand{\inl}{\mathsf{inl}}

\newcommand{\tirrAx}[1][\kappa]{\mathsf{tirr}^{#1}}
\newcommand{\pfix}[1][\kappa]{\mathsf{pfix}^{#1}}

\newcommand{\clocktype}{\mathsf{clock}}

\newcommand{\capp}[2][\kappa]{#2\,[#1]}

\newcommand\wfcxt[2][]{#2 \vdash_{#1}}
\newcommand\hastype[4][]{
#2 \vdash_{#1} #3: #4
}
\newcommand{\subst}[2]{[#1/#2]}

\newcommand{\pFPC}{\ProbFPC}

\newcommand{\Fin}{\mathsf{Fin}}

\AtBeginDocument{%
  }

\setcopyright{none}
\copyrightyear{2024}
\acmYear{2024}
\acmDOI{XXXXXXX.XXXXXXX}

\begin{document}

\title{Modelling Recursion and Probabilistic Choice in Guarded Type Theory} 

\author{Philipp Jan Andries Stassen}
\email{stassen@cs.au.dk}
\affiliation{%
  \institution{Aarhus University}
  \country{Denmark}
}

\author{Rasmus Ejlers M{\o}gelberg}
\email{mogel@itu.dk}
\affiliation{%
  \institution{IT University of Copenhagen}
  \country{Denmark}
}

\author{Maaike Zwart}
\email{mazw@itu.dk}
\affiliation{%
  \institution{IT University of Copenhagen}
  \country{Denmark}
}

\author{Alejandro Aguirre}
\orcid{0000-0001-6746-2734}
\affiliation{
  \institution{Aarhus University}
  \country{Denmark}
}
\email{alejandro@cs.au.dk}

\author{Lars Birkedal}
\orcid{0000-0003-1320-0098}
\affiliation{
  \institution{Aarhus University}
  \country{Denmark}
}
\email{birkedal@cs.au.dk}

\renewcommand{\shortauthors}{Stassen et al.}

\begin{abstract}
  Constructive type theory combines logic and programming in one language.
  This is useful both for reasoning about programs written in type theory, as well as for reasoning about other programming languages inside type theory.
  It is well-known that it is challenging to extend these applications to languages with recursion and computational effects such as probabilistic choice, because these features are not easily represented in constructive type theory. 

    We show how to define and reason about $\ProbFPC$, a programming language with probabilistic choice and recursive types, in guarded type theory. We use higher inductive types to represent finite distributions and guarded recursion to model recursion.
    We define both operational and denotational semantics of $\ProbFPC$, as well as a relation between the two.
    The relation can be used to prove adequacy, but we also show how to use it to reason about programs up to contextual equivalence.
\end{abstract}



\keywords{Probabilistic Programming Languages, Type Theory, Guarded Recursion, Recursive Types}

\maketitle


\section{Introduction}
\label{sec:introduction}

Traditionally, modelling and reasoning about programming languages is done using either operational or denotational techniques.
Denotational semantics provides mathematical abstractions that 
are used to see beyond the details of the operational implementation, to
describe principles common to many different languages, and to provide modular building blocks that can be extended 
to settings with more language features. 
On the other hand, denotational semantics can easily become very sophisticated mathematically, and 
modelling combinations of recursion and other features such as probabilistic sampling and higher-order store can be difficult.
Operational techniques are often more direct and flexible, and are usually easier to implement in proof assistants, 
but reasoning requires constructing complex operational models and logics.

\emph{Synthetic guarded domain theory} (SGDT) provides an alternative approach
to both denotational and operational semantics. 
The idea is to work in an expressive meta-language with guarded recursion in the sense of Nakano~\cite{nakano:Modality},
and to use recursion on the meta-level to model recursion on the object level.
The specific meta-language used 
in this paper is Clocked Cubical Type Theory (CCTT)~\cite{CubicalCloTT}.
It includes a modal type-operator $\later^\kappa$, indexed by a so-called clock $\kappa$ (see Section \ref{sec:cctt}), to describe data that is available 
one time step from now. It is possible to define elements by guarded recursion by means of a fixed point combinator $\fix^\kappa : (\later^\kappa A \to A) \to A$.
By applying the fixed point combinator to an operator on a universe one can obtain solutions to guarded recursive type equations \cite{DBLP:conf/lics/BirkedalM13}.
For example, one can define a \emph{guarded delay monad} $\grD$, which maps a type $A$ to $\grD A$ satisfying $\grD A \equi A + \later^\kappa (\grD A)$.
This guarded delay monad has been used to model recursive computations, both 
in operational semantics, as well as in denotational semantics, which simply models function types of the object language using the Kleisli functions 
of the metalanguage: $\Interp{A \to B} \defeq \Interp A \to \grD \Interp B$. Guarded recursion is then also used for reasoning about the model, in particular
to establish a relation between syntax and semantics to prove adequacy. These results were originally proved for the simply typed
lambda calculus extended with recursive terms~\cite{DBLP:journals/entcs/PaviottiMB15} 
and recursive types~\cite{DBLP:journals/mscs/MogelbergP19}, but have been recently extended to cover also languages with general store and 
polymorphism~\cite{DBLP:journals/corr/abs-2210-02169}.

SGDT has several of the benefits listed for operational and denotational techniques above: the models can be described at a high level of
abstraction, and once one is familiar with programming in the meta-language,
the mathematical knowledge required for constructing the models is limited. 
The main reason for both of these is that recursion and other tools needed for constructing and reasoning about these models, are build into the meta-language,
so the mathematical difficulties appear on the next meta-level. Moreover, since metalanguages such as CCTT are dependent type theories, these results
can potentially be directly implemented in proof assistants. For example, CCTT has been implemented\footnote{https://github.com/agda/guarded} 
as an experimental extension of Cubical Agda~\cite{DBLP:journals/pacmpl/VezzosiM019}. 
Finally, the denotational semantics of the language can be seen as a shallow embedding
of the object language, so one can potentially also use the meta-language as a language for both programming and reasoning. 

In this paper we show how to extend SGDT to model probabilistic programming languages, i.e., 
languages that include commands that generate random values by sampling from a probability distribution.
It is well known that it is challenging to develop semantic models for reasoning about higher-order probabilistic programming languages
that include recursion, even in a classical meta-theory.
Nonetheless, a plethora of denotational approaches have been investigated in recent years, \cite{DBLP:conf/lics/JonesP89, DBLP:journals/pacmpl/VakarKS19,DBLP:conf/lics/HeunenKSY17,DBLP:journals/jacm/EhrhardPT18, DBLP:journals/pacmpl/DahlqvistK20}.
Other operational-based approaches to reason about probabilistic programs have also been shown to scale to rich languages with a variety of features, 
using techniques such as logical relations~\cite{DBLP:conf/lics/JohannSV10,DBLP:conf/fossacs/BizjakB15,DBLP:conf/esop/CulpepperC17,DBLP:journals/pacmpl/WandCGC18,DBLP:journals/pacmpl/ZhangA22}, or bisimulations~\cite{DBLP:conf/esop/CrubilleL14,DBLP:conf/popl/LagoSA14}.


Precisely, we show how to develop operational and denotational semantics of 
$\ProbFPC$\footnote{The name derives from Plotkin's Fixed Point Calculus (FPC)~\cite{plotkin1985denotational,DBLP:phd/ethos/Fiore94}}, a call-by-value higher-order probabilistic programming language with recursive types,
in CCTT, and prove that the denotational semantics is adequate with respect to the operational semantics. We also show how to use these results for reasoning about $\ProbFPC$ programs
in CCTT. 
One of the challenges for doing so is that most of the previously developed theory for probabilistic languages relies on classical reasoning, which is not available in CCTT. 
For example, many types used in our model do not have decidable equality, and as a consequence, even the finite distribution monad $\Dist{}$ cannot be defined in the standard
classical way. Fortunately, CCTT not only models guarded recursion but also higher-inductive types (HITs), which
we use to define $\Dist{}$: on a set $A$,  $\Dist{A}$
is the free convex algebra \cite{DBLP:conf/ifipTCS/Jacobs10}, that is, a set $\Dist{A}$ together with a
binary operation $\DistChoice{p}{}{}$, indexed by a rational number $p$, satisfying 
the natural equational theory (idempotency, associativity, and commutativity).
Using a HIT to represent distributions
allows us to use the equational properties 
of the meta-language level when
reasoning about the semantics of $\ProbFPC$, and it provides us with a useful induction principle for
proving propositions ranging over $\Dist{A}$.

To model the combination of recursion and probabilistic choice,
we use the 
\emph{guarded convex delay monad} which we denote by $\GLiftDist{}$. On a type $A$,
$\GLiftDist{A} \simeq \Dist(A+\Later{(\GLiftDist{A})})$. 
Intuitively, this means that a computation of type $A$ will be a distribution over values of $A$ (immediately
available) and delayed computations of type $A$. Quantifying over clocks gives the \emph{convex delay monad}
$\LiftDist A \defeq \forall\kappa. \GLiftDist A$ which is a coinductive solution to the equation 
$\LiftDist A \equi \Dist(A + \LiftDist A)$. Both operational and denotational semantics
are defined using $\GLiftDist{}$ because it allows for guarded recursive definitions. These definitions 
can then be clock quantified to give elements of $\LiftDist{}$, which has the benefit that steps
can be eliminated, because $\Later$ does not get in the way. 

As our main result, we define a logical relation, relating
the denotational and operational semantics, and prove that it is sound with
respect to contextual refinement. Traditionally, defining a logical relation
relating denotational and operational semantics for a language with recursive
types is non-trivial, see, e.g., \cite{DBLP:journals/iandc/Pitts96}. Here we use
the guarded type theory to define the logical relation by guarded recursion. As
usual, the logical relation is divided into a relation for values and a relation
for computations. Since computations compute to distributions, the challenge
here lies in defining the relation for computations in terms of the relation for
values. Earlier work on operationally-based logical relations
\cite{DBLP:conf/fossacs/BizjakB15,DBLP:journals/pacmpl/AguirreB23} used
bi-orthogonality to reduce the problem to relating termination probabilities for
computations of ground type. Here, we follow the approach of the recent
article~\cite{gregersenAHTB24}, which uses
couplings~\cite{lindvall_lectures_2002,thorisson/2000,Villani2008OptimalTO,DBLP:conf/popl/BartheGB09}
to lift relations on values to relations on distributions. Note that the
development in~\cite{gregersenAHTB24} relies on classical logic (in particular,
the composition of couplings, the so-called bind lemma, relies on the axiom of
choice) and thus does not apply here. Instead, we define a novel constructive
notion of lifting of relations $\Rel : A \to B \to \Prop$ to relations on convex delay 
algebras $\GRelLift{\Rel} : \GLiftDist{A}\to \LiftDist{B} \to \Prop$ by guarded recursion. 
We establish a series of basic results for this, including a version of the important bind lemma that
allows us to compose these liftings in proofs.

Finally, we use the semantics and the logical relation to prove contextual refinement of examples that combine probabilistic choice and recursion.
The examples illustrate yet another benefit of the denotational semantics: Since fewer steps are needed in the denotational semantics 
than in the operational semantics, using this for reasoning makes the arguments less cluttered by steps than a direct
relation between syntax and syntax would. 

\paragraph{Contributions}

\begin{enumerate}
\item To the best of our knowledge, we present the first constructive type theoretic account of
  operational and denotational semantics of $\ProbFPC$.
\item We develop the theory of the finite distribution monad in cubical type theory, and we introduce the convex delay monad in CCTT,
and develop the basic theory for it. 
\item We develop the basic constructive theory of couplings for convex delay algebras and use it to define a logical relation, relating denotational and operational semantics.
 \item We demonstrate how to use the semantics to reason about examples that combine probabilistic choice and recursion.
\end{enumerate}

\newcommand{\nextop}[1][\kappa]{\mathsf{next}^{#1}}

\section{Clocked Cubical Type Theory}
\label{sec:cctt}

We will work in Clocked Cubical Type Theory (CCTT)~\cite{CubicalCloTT}, an extension of Cubical Type Theory~\cite{CTT}
with guarded recursion. 
At present, CCTT is the only existing type theory containing all constructions needed for this paper.
We will not describe CCTT in detail, but only
describe the properties we need,
with the hope of making the paper more accessible, and
so that the results can be reused in other (future) type theories with the same properties.

\subsection{Basic properties and HITs}

CCTT has an infinite hierarchy of (Tarski style) universes ($\Univ[i]$), as well as identity types satisfying
the standard rules and function extensionality. Precisely, CCTT, being based on Cubical Type Theory,
has a path type as primitive, rather than an identity type in the traditional sense. However, 
the differences between the two are inessential to this work.
We will write $a \peq_A b$, or just $a \peq b$,
for the identity type associated with terms $a,b : A$ of the same type.
Following conventions from homotopy type
theory~\cite{hottbook}, we say that a type $A$ is a (homotopy) proposition if for
any $x,y : A$, the type $x \peq y$ is contractible, and that it is a (homotopy) set if for any
$x,y : A$, the type $x \peq y$ is a proposition. We write $\Prop[i]$, and $\Set[i]$
for the subuniverses of $\Univ[i]$ of propositions and sets respectively, often omitting the
universe level $i$. We write $A \equi B$ for the type of equivalences
from $A$ to $B$. We shall mostly use this in the case where $A$ and $B$ are sets, and in that
case, an equivalence is simply given by the standard notion of isomorphism of sets as phrased
inside type theory using propositional equality. All types used represent to syntax, as well as the denotation
of all types, are sets. Likewise all relations are valued in $\Prop$. These choices simplify reasoning,
as no higher dimensional structure needs to be accounted for. The only types used that are not sets
are the universes $\Set$ and $\Univ$. 

CCTT also has higher inductive types (HITs), and we will use these to construct
propositional truncation, set truncation and the finite distributions monad (see \autoref{sec:fin:dist}).
In particular, this means that one can express ordinary propositional logic with operators $\land, \lor, \exists, \forall$
on $\Prop$, using the encodings of $\exists$ and $\vee$ defined using propositional truncation~\cite{hottbook}.
Recall in particular the elimination principles for $\exists$: When proving a proposition $\psi$
assuming $\exists (x : X). \phi(x)$ we may assume we have an $x$ in hand satisfying $\phi(x)$, but
we cannot do that when mapping from $\exists (x : X). \phi(x)$ to an arbitrary type.
We will also need inductive types to represent the type $\NN$ of natural numbers,
as well as types and terms of the language $\pFPC$ described in \autoref{sec:prob:FPC}.
These are all captured by the schema for higher inductive types in CCTT~\cite{CubicalCloTT}.
%

\subsection{Guarded recursion}

\begin{figure}
\begin{mathpar}
  \inferrule*
  {\wfcxt{\Gamma} \\ \kappa : \clocktype \in \Gamma}
  {\wfcxt{\Gamma, \tickA : \kappa}{}}
  \and
  \inferrule*
  {\hastype{\Gamma}{t}{\latbind{\tickA}{\kappa} A}\\ \wfcxt{\Gamma,\tickB:\kappa,\Gamma'}}
  {\hastype{\Gamma,\tickB: \kappa,\Gamma'}{\tapp[\tickB] t}{A\toksubst{\tickB}{\tickA}}}
  \and
  \inferrule*
  {\hastype{\Gamma,\tickA:\kappa}{t}{A}}
  {\hastype{\Gamma}{\tabs{\tickA}{\kappa} t}{\latbind{\tickA}{\kappa} A}}
  \and
    \inferrule*
  {\hastype{\Gamma,\kappa : \clocktype}{t}{A}}
  {\hastype{\Gamma}{\Lambda\kappa. t}{\forall \kappa . A}}
  \and
  \inferrule*
  {\hastype{\Gamma}{t}{\forall \kappa . A}\\
    \hastype\Gamma{\kappa'}\clocktype}
  {\hastype{\Gamma}{t [\kappa']}{A \subst{\kappa'}{\kappa}}}
  \and
  \inferrule*
  {\hastype{\Gamma}{t}{\later^\kappa A \to A}}
  {\hastype{\Gamma}{\dfix^\kappa\,t}{\later^\kappa A}}
  \and
  \inferrule*
{\hastype[]{\Gamma}{t}{\later^\kappa A \to A} }
{\hastype{\Gamma}{\pfix[\kappa] \, t}{\latbind{\tickA}{\kappa}{{\tapp{(\dfix^{\kappa} t)}}\peq_A{t(\dfix^{\kappa} t)}}}}
\end{mathpar}
\caption{Selected typing rules for Clocked Cubical Type Theory.}
\label{fig:later:typing}
\end{figure}

Guarded recursion uses a modal operator $\later$ (pronounced `later') to describe data that is available one time step from now.
In multiclocked guarded recursion, $\later$ is indexed by a clock $\kappa$. Clocks can be variables of
the pretype $\clocktype$ or they can be the clock constant $\clockc$. Clocks can be universally quantified in the type
$\forall\kappa. A$, with rules similar to $\Pi $-types, including a functional extensionality principle.
The modality $\later$ is a Fitch-style modality \cite{clouston2018fitch,drat}, in the sense that introduction and elimination for
$\later$ is by abstraction and application to \emph{ticks}, i.e., assumptions of the form $\tickA :\kappa$, to be thought
of as evidence that time has ticked on clock $\kappa$. Because ticks can occur in terms, they can also occur in
types, and the type $\latbind{\tickA}{\kappa} A$ binds $\tickA$ in $A$. When $\tickA$ does not occur in $A$ we simply
write $\later^\kappa A$ for $\latbind{\tickA}{\kappa} A$. The rules for tick abstraction and application that we shall
use in this paper are presented in Figure~\ref{fig:later:typing}. In the figure, the notation $\wfcxt\Gamma$ means
that $\Gamma$ is a well-formed context.  Note that the rule for tick application assumes that
$\tickB$ (or anything occurring after that in the context) does not already occur in $t$. This is to avoid terms of the type
$\later^\kappa \later^\kappa A \to \later^\kappa A$ merging two time steps into one.
We will sometimes use the notation
\begin{equation} \label{eq:nextop}
 \nextop \defeq \lambda x. \tabs\tickA\kappa x : A \to \later^\kappa A.
\end{equation}

The rules presented in Figure~\ref{fig:later:typing} are special cases of those of CCTT. The general rules allow
certain `timeless' assumptions in $\Gamma'$ to occur in $t$ in the tick application rule. This allows typing of
an extensionality principle for $\later$ of type
\begin{equation}
\label{eq:later:ext}
\left(x \peq_{\later^\kappa A} y\right)
\equi
\latbind\tickA\kappa{\left({\tapp x}\peq_A{\tapp y}\right)}.
\end{equation}
In this paper we shall simply take this as an axiom. 
Intuitively, (\ref{eq:later:ext}) states that $x$ and $y$ are equal, 
if the elements they deliver in the next time step are equal. 
One consequence of this
is that $\later$ preserves the property of being a set or a proposition.
Other omitted rules for ticks allow typing of
a \emph{tick-irrelevance} axiom 
\begin{equation} \label{eq:tirrAx}
 \tirrAx 
  : \Pi(x : \later^\kappa A) . \latbind{\tickA}{\kappa}  \latbind{\tickB}{\kappa} {(\tapp{x}) \peq_A(\tapp[\tickB] x)}.
\end{equation}
stating that all ticks are propositionally equal (although they should not be considered judgementally equal~\cite{bahr2017clocks}).

The fixed point combinator $\dfix{}$ allows for defining and programming with guarded recursive types.
Define $\fix^\kappa : (\later^\kappa A \to A) \to A$ as $\fix^\kappa f = f(\dfix^\kappa f)$. Then one can prove that
\begin{equation} \label{eq:fix:unfold}
\fix^\kappa t
\peq t(\tabs\tickA\kappa{\fix^\kappa t}).
\end{equation}
Applying the fix point operator to maps on a universe, as in \cite{DBLP:conf/lics/BirkedalM13}, one can define \emph{guarded recursive types}
such as the \emph{guarded delay monad} $\grD$ mapping a type $A$ to $\grD A$ satisfying
\begin{equation}
\label{eq:grD}
\grD A \equi A + \later^\kappa (\grD A).
\end{equation}
In this paper, we will not spell out how such guarded recursive
types are defined as fixed points, but just give the defining guarded recursive equation. Intuitively, the reason
this is well defined is that $\grD A$ only occurs under a $\later$ on the right-hand side of (\ref{eq:grD}), which
allows the recursive equation to be phrased as a map $\later^\kappa \Univ \to \Univ$.
One can also use $\fix$ to program with $\grD$ defining, e.g., the
diverging computation as $\bot = \fix(\lambda x . \inr(x))$ (leaving the type equivalence between $\grD A$
and its unfolding implicit). In this case (\ref{eq:fix:unfold}) specialises to $\bot \peq \inr(\tabs\tickA\kappa\bot)$.

\subsection{Coinductive types}

Coinductive types can be represented by quantifying over clocks in guarded recursive types~\cite{DBLP:conf/icfp/AtkeyM13}.
For example, the type $\ciD A \defeq \forall\kappa. \grD A$ defines a coinductive solution to $\ciD A \equi A + \ciD A$ in
CCTT, provided $A$ is \emph{clock-irrelevant}, meaning that the canonical map $A \to \forall\kappa. A$ is an equivalence.
More generally, one can prove that any functor $F : \Univ \to \Univ$ (in the naive sense of having a functorial action
$(A \to B) \to FA \to FB$) commuting with clock quantification in the sense that $F(\forall\kappa . X) \equi \forall\kappa. F(X)$
via the canonical map, has a final coalgebra defined as $\nu(F) \defeq \forall \kappa. \nu^\kappa(F)$, where
$\nu^\kappa(F)\equi F(\later^\kappa(\nu^\kappa(F))$ is defined using $\fix$. This encoding also works
for indexed coinductive types. The correctness of the encoding of coinductive types can be proved in CCTT, and
relies on the type equivalence
\[
 \forall\kappa. A \equi \forall\kappa. \later^\kappa \!A.
\]

For the encoding to be useful, one needs a large collection of clock-irrelevant types and functors commuting with clock
quantification. We will need that $\NN$ is clock-irrelevant, and that clock quantification commutes with sums
in the sense that $\forall\kappa.(A + B) \equi (\forall\kappa. A) + (\forall\kappa.B)$. Both of these can be proved in
CCTT using the notion of induction under clocks for higher inductive types~\cite{CubicalCloTT}. Note also that all propositions $P$ are clock-irrelevant,
because the clock constant $\clockc$ can be used to define a map $(\forall\kappa . P) \to P$. The following lemma can be used to prove types clock irrelevant.

\begin{lemma} \label{lem:embed:to:cirr}
 Suppose $i : A \to B$ is injective in the sense of the existence of a map
 $\Pi x,y: A. (i(x) \peq i(y)) \to x\peq y$, and suppose $B$ is clock irrelevant. Then
 also $A$ is clock-irrelevant.
\end{lemma}

For the rest of the paper, we work informally in CCTT. 

\section{Finite distributions}
\label{sec:fin:dist}
In the previous section we saw the definition of the guarded delay monad $\grD$ (Equation~\ref{eq:grD}), which models non-terminating computations. In this section, we introduce the distribution monad $\Dist$, which models finite probability distributions. These two monads will be combined to form the guarded convex delay monad in the next section. 

To construct the distribution monad $\Dist$, we first need some infrastructure to reason about probabilities.  
We will assume types representing the open rational interval $\II$ and closed rational interval $\IIc$, 
and we assume that we have the operations product, division of smaller numbers by larger numbers, and inversion ($1-(-)$). 
We also need that $\II$ is a set with decidable equality, and that $\II$ is clock irrelevant.
In practice, there are several ways of specifying $\II$, for example as a type of pairs $(n,d)$, of mutually prime, positive natural numbers satisfying $n < d$, and $\IIc$ can be obtained, e.g., by just adding two points to $\II$. For this encoding, clock irrelevance of $\II$ follows from the embedding into $\NN\times\NN$ and Lemma~\ref{lem:embed:to:cirr}. However, for this paper the specific implementation is irrelevant, so we leave this open. 

Then, we need to find the right representation of probability distributions in type theory. 
In classical presentations of probability theory, a finite distribution on a set $A$ is a finite map into $\IIc$, whose
values sum to $1$. In type theory, this would be represented by a subtype of $A \to \IIc$, which would need some notion of finite support, as well as a definition of the functorial action of $\Dist$. Especially the latter is difficult to do: Given a map $f: A \rightarrow B$, then $\Dist(f): \Dist(A) \rightarrow \Dist(B)$ should map a probability distribution $\mu$ to the distribution which, for each $b : B$, sums the probabilities $\mu(a)$ for each $a$ such that $f(a) = b$. To define this, we would need decidable equality on $B$ to compute for which $a$ the equality $f(a) = b$ holds, which is too restrictive. 
Another approach could be to use lists of key-value pairs, but this requires a quotient to obtain the correct notion of equality of distributions, which would make it very hard to work with in practice. 

Here we instead choose to represent $\Dist$ as the free monad for the theory of convex algebras, which is known to generate the finite distribution monad~\cite{DBLP:conf/ifipTCS/Jacobs10}. The operation $\DistChoiceOp$ in the definition below should be thought of as a convex sum: $\DistChoice{p}{x}{y} = p \cdot x + (1-p) \cdot y$.  

\begin{definition}[Convex Algebra]
  A \emph{convex algebra} is a set $A$ together with an operation $\DistChoiceOp :\II \to A \to A \to A$, such that
  \begin{align*}
	  \tag{idem}  \DistChoice{p}{\mu}{\mu} &\peq \mu \\
	  \tag{comm} \DistChoice{p}{\mu}{\nu} &\peq \DistChoice{1-p}{\nu}{\mu} \\
	  \tag{assoc}  \DistChoice{q}{\left(\DistChoice{p}{\mu_1}{\mu_2}\right)}{\mu_3}
	  &\peq \DistChoice{pq}{\mu_1}{\left(\DistChoice{\frac{q-pq}{1-pq}}{\mu_2}{\mu_3}\right)}
  \end{align*}
  A map $f : A \to B$ between convex algebras $A$ and $B$ is a \emph{homomorphism} if
  $f(\DistChoice{p}{\mu}{\nu}) \peq \DistChoice{p}{f(\mu)}{f(\nu)}$ holds for all $\mu,\nu, p$.
\end{definition}

\begin{definition} \label{def:Dist}
 Let $\Dist(A)$ be the higher inductive type defined using two constructors
\begin{align*}
 \DiracOp & : A \to \Dist(A) & \DistChoiceOp :\II \to \Dist(A) \to \Dist(A) \to \Dist(A),
\end{align*}
and equations (paths) for idempotency, commutativity and associativity as in the definition of convex algebra, plus
an equation for set truncation.
\end{definition}
%

One advantage of HITs is that they come with an easy induction principle. For $\Dist(A)$, it is as follows: 
If $\phi(x)$ is a proposition for all $x : \Dist(A)$ and $\phi(\Dirac a)$ holds for all $a$, and moreover, $\phi(\mu)$ and $\phi(\nu)$ implies $\phi(\DistChoice p\mu \nu)$ for all $\mu,\nu, p$, then $\phi(x)$ holds for all $x$. We will use this proof technique in numerous places throughout this paper. Similarly, the recursion principle for 
$\Dist(A)$ states that to define a map $f$ from $\Dist(A)$ into a set $B$, it suffices to define the cases of $f(\Dirac x)$ and $f(\DistChoice pxy)$,
the latter using $f(x)$ and $f(y)$, in such a way that the convex algebra equations are respected. The recursion principle implies the following.

\begin{proposition}
 $\Dist(A)$ is the free convex algebra on $A$, in the sense that for any convex algebra $B$, and function
 $f : A \to B$, there exists a unique homomorphism of convex algebras $\uext f: \Dist(A) \to B$ satisfying $f = \uext f \circ \DiracOp$.
 As a consequence, $\Dist(-)$ forms a monad on the category of sets.
\end{proposition}

Now that we have a type $\Dist(A)$ of probability distributions over a set $A$, we would like to reason about these distributions. In particular, we would like to know when two distributions are equal. For that, we use the following trick: we first relate probability distributions on finite sets to the classical probability distributions on these sets (that is, to functions into $\IIc$ with finite support). Then, we use functoriality of $\Dist$ to generalise to arbitrary sets $A$.

Define the Bishop finite sets in the standard way by induction on $n$ as $ \Fin(0) = 0$ and $\Fin(n+1) = \Fin(n) + 1$, where the $+$ on 
the right hand side refers to sum of types.
%
For these sets, we can relate $\Dist$ to its classical definition.
First, we use that if any set $A$ has decidable equality, we can associate a probability function $f_\mu : A \to \IIc$ to a distribution $\mu$ by induction, using $f_{\delta(x)}(y) = 1$ if $x\peq y$ and $f_{\delta(x)}(y) = 0$ else. This gives us an equivalence of types:

\begin{lemma} \label{lem:Fin:Dist}
 $\Dist(\Fin(n)) \equi \Sigma( f : \Fin(n) \to \IIc) . \mathsf{sum}(f) = 1$ 
\end{lemma}
Here, $\mathsf{sum}$ is the sum of the values of $f$, defined by induction on $n$. 
Note that the right hand side of this equivalence indeed captures the classical definition of probability distributions over $\Fin(n)$.
We use Lemma~\ref{lem:Fin:Dist} to reason about equality for probability distributions as follows.
\begin{example} \label{example:proving:eq}
 The isomorphism of Lemma~\ref{lem:Fin:Dist} allows us to prove equations of distributions
 by first mapping to the right hand side of the isomorphism and then using functional extensionality. Consider, for example, the equation
\begin{align*}
 \DistChoice p{(\DistChoice p ab)}{(\DistChoice pca)} \peq \DistChoice{2p(1-p)}{(\DistChoice{\frac12}bc)}a,
\end{align*}
 where $a,b,c : \Dist(A)$ for some $A$. In the case where $A = \Fin(3)$ and $a = \Dirac 0,b = \Dirac 1,c = \Dirac 2$,
 we can prove this by noting that both the left and right-hand sides of the equation correspond to the map
 $f(0) = p^2 + (1-p)^2 = 1 + 2p^2- 2p = 1-2p(1-p)$, $f(1) = f(2) = p(1-p)$. 
 Finally, the case of general $A,a,b,c$ follows from applying functoriality to the canonical map $\Fin(3) \to \Dist A$.
\end{example}


Lastly, the guarded convex delay monad in the next section will be defined as a distribution on a sum type. To reason about such distributions, we will frequently use the fact that any distribution $\mu : \Dist(A + B)$ on a sum type can be written uniquely as
a convex sum of two subdistributions: One on $A$ and one on $B$. In the classical
setting where a distribution is a map with finite support this is trivial: The two subdistributions are simply the normalised restrictions
of the distribution map to $A$ and $B$ respectively. In the constructive type theoretic setting the proof requires a bit more work, so we mention this as a separate theorem.

\begin{theorem} \label{thm:dist:sum:equiv}
 For all sets $A$ and $B$, the map
 \[
  \Dist(A) + \Dist(B) + \Dist(A)\times \II \times \Dist(B) \to \Dist(A + B),
 \]
 defined by
 \begin{align*}
 f(\In1(\mu)) & \defeq \Dist(\inl)(\mu) \\
 f(\In2(\mu)) & \defeq \Dist(\inr)(\mu) \\
 f(\In3(\mu, p , \nu)) & \defeq \DistChoice p{(\Dist(\inl)(\mu))}{(\Dist(\inr)(\nu))},
\end{align*}
is an equivalence of types.
\end{theorem}

\section{Convex delay algebras}
\label{sec:convex:delay:alg}


In this section, we define the guarded convex delay monad $\GLiftDist$ and the convex delay monad $\LiftDist$ modelling the combination of probabilistic choice and recursion.
We first recall the notion of delay algebra (sometimes called a lifting or a  $\Later$-algebra~\cite{DBLP:journals/entcs/PaviottiMB15}) and define a notion of convex delay algebra.

%
\begin{definition}[Convex Delay Algebra]
 A ($\kappa$-)\emph{delay algebra} is a set $A$ together with a map ${\GStep: \later^\kappa A \to A}$. A delay algebra
 homomorphism is a map $f : A \to B$ such that $f(\GStep(a)) \peq \GStep(\tabs{\tickA}\kappa{f(\tapp a)})$ for all $a : \later^\kappa A$.
 A \emph{convex delay algebra} is a set $A$ with both a delay algebra structure and a convex algebra structure, and a
 homomorphism of these is a map respecting both structures.
\end{definition}

Defining the functorial action of $\later^\kappa$ as $\later^\kappa(f)(a) \defeq \tabs{\tickA}\kappa{f(\tapp a)}$, the requirement for $f$
being a delay algebra homomorphism can be expressed as the following commutative diagram 
\[  \begin{tikzcd}
    \later^\kappa A \ar{r}{\later^\kappa (f)} \ar{d}{\GStep} & \later^\kappa B \ar{d}{\GStep} \\
    A \ar{r}{f} & B
  \end{tikzcd}
\]
%

Recall the guarded delay monad $\grD$ satisfying $\grD A \equi A + \later^\kappa (\grD A)$. The type $\grD A$ is easily seen to be the free guarded delay algebra on a set $A$. Similarly, define the \emph{guarded convex delay monad}
%
as the guarded recursive type
\[
  \GLiftDist{A} \simeq \Dist(A+\Later{(\GLiftDist{A})}).
\]
Again, this can be constructed formally as a fixed point on a universe (assuming $A$ lives in the same universe), but
we shall not spell this out here. We show that it is a free convex delay algebra. First
define the convex delay algebra structure ($\GStep, \GChoiceOp$) on $\GLiftDist{A}$ and the inclusion $\GDirac : A \to \GLiftDist{A}$
by
\begin{align*}
 \GStep(a) & = \Dirac {\inr(a)} &
 \GChoice p\mu\nu & = \DistChoice p\mu\nu & 
 \GDirac{a} & = \Dirac {\inl(a)} 
\end{align*}
where the convex algebra structure on the right-hand sides of the equations above refers to those of $\Dist$.

\begin{proposition}
 $\GLiftDist{A}$ is the free convex delay algebra structure on $A$, for any set $A$. As a consequence $\GLiftDist{}$
 defines a monad on the category of sets.
\end{proposition}


\begin{proof}
 Suppose $f: A \to B$ and that $B$ is a convex delay algebra. We define the extension $\uext f : \GLiftDist{A} \to B$
 by guarded recursion, so suppose we are given $g : \later^\kappa(\GLiftDist{A} \to B)$ and define
\begin{align*}
 \uext f(\GStep (a)) & = \GStep(\tabs\tickA\kappa{\tapp g(\tapp a)}) \\
 \uext f(\GDirac(a)) & = f(a) \\
 \uext f(\GChoice p\mu\nu) & = \DistChoice p{\uext f(\mu)}{\uext f(\nu)}.
\end{align*}
Note that these cases define $\uext f$ by induction on $\Dist$ and $+$.
Unfolding the guarded recursive definition using (\ref{eq:fix:unfold}) gives
\[
 \uext f(\GStep (a)) = \GStep(\tabs\tickA\kappa{\uext f(\tapp a)}).
\]
so that $\uext f$ is a homomorphism of convex delay algebras. For uniqueness, suppose $g$ is another
homomorphism extending $f$. We show that $g \peq \uext f$ by guarded recursion and function extensionality.
So suppose we are given $p : \Later(\Pi (x : \GLiftDist{A})(g(a) \peq \uext f(a)))$.
Then, for any $a : \later^\kappa(\GLiftDist{A})$, the term $\tabs\tickA\kappa{\tapp p (\tapp a)}$ proves
\[{\latbind\tickA\kappa{(g(\tapp a)\peq \uext f(\tapp a))}},\]
which by (\ref{eq:later:ext}) is equivalent to
 $\tabs\tickA\kappa{g(\tapp a)} \peq \tabs\tickA\kappa{\uext f(\tapp a)}$.
Then,
\begin{align*}
 g(\GStep(a)) & \peq \GStep(\tabs\tickA\kappa{g(\tapp a)}) \\
 & \peq \GStep(\tabs\tickA\kappa{\uext f(\tapp a)}) \\
 & \peq \uext f(\GStep (a)).
\end{align*}
The rest of the proof that $g(a) = \uext f(a)$, for all $a$, then follows by HIT induction on $\Dist$.

In terms of category theory, we have shown that $\GLiftDist$ is left adjoint to the forgetful functor from 
convex delay algebras to the category of sets. It therefore defines a monad on the category of sets. 
\end{proof}

We will often write $t \GBind f$ for $\uext f(t)$ where $\uext f$ is the unique extension of
$f$ to a convex delay algebra homomorphism.

\begin{example}
 The type $\GLiftDist A$ can be thought of as a type of probabilistic processes returning values in $A$. Define, for example,
 a geometric process with probability $p : \II$ as $\Ggeo p\, 0$ where
\begin{align*}
 \Ggeo p & : \NN \to \GLiftDist \NN \\
 \Ggeo p\,n & \defeq \DistChoice p{(\GDirac n)}{\GStep(\tabs\tickA\kappa{\Ggeo p\,(n+1)})}
\end{align*}
Note that this gives
\begin{align*}
 \Ggeo p (0) & \peq \DistChoice p{(\GDirac 0)}{}{\GStep(\tabs\tickA\kappa{\Ggeo p\,(1)})}  \\
 & \peq \DistChoice p{(\GDirac 0)}{\left(\GStep\tabs\tickA\kappa{\left(\DistChoice p{(\GDirac 1)}{(\GStep\tabs\tickA\kappa{\Ggeo p\,(2)})}\right)}\right)}  \\
 & \peq \ldots
\end{align*}
\end{example}

The modal delay $\Later$ in the definition of $\GLiftDist A$ prevents us from accessing values computed
later, thereby also preventing us from, \emph{e.g.}, computing probabilities of `termination in $n$ steps' as elements of $\IIc$. Such operations
should instead be defined on the convex delay monad $\LiftDist$. This monad is defined in terms of $\GLiftDist$ by quantifying over clocks:
\[
  \LiftDist{A} \defeq \forall\kappa . \GLiftDist A.
\]

\begin{proposition} \label{prop:LiftDist}
 If $A$ is clock-irrelevant then $\LiftDist{A}$ is the final coalgebra
 for the functor $F(X) = \Dist(A+X)$.
\end{proposition}

\begin{proof}
 We must show that $F$ commutes with clock quantification, and this reduces easily to showing that
 $\Dist$ commutes with clock quantification. The latter can be either proved directly by using the
 same technique as for the similar result for the finite powerset functor~\cite{CubicalCloTT}, or by referring to
 \cite[Proposition~14]{ZwartExtraTime}, which states that any free model monad for a theory with finite arities
 commutes with clock quantification.
\end{proof}

\begin{remark}
  From now on, whenever we look at a type $\LiftDist{A}$, we assume $A$ to be clock-irrelevant.
\end{remark}

In particular, $\LiftDist{A} \equi \Dist(A + \LiftDist A)$, and so $\LiftDist$ carries a convex algebra structure $(\CEta\,,  \CChoiceOp)$
as well as a map $\CStep : \LiftDist A \to \LiftDist A$ defined by $\CStep(x) = \Dirac{\inr(x)}$. Define also $\CDirac : A \to \LiftDist{A}$
as $\CDirac x \defeq \Dirac{\inl\, x}$.

\begin{example}
 The geometric process can be defined as an element of $\NN \to \LiftDist{\NN}$ as $\geo p\,n \defeq \Lambda\kappa . \Ggeo p\, n$.
 This satisfies the equations
 \begin{align*}
 \geo p\, 0 &  \peq \CChoice p{(\CDirac\, 0)}{\CStep(\geo p\,(1))}  \\
 & \peq \DistChoice p{(\CDirac\, 0)}{\CStep(\DistChoice p{(\CDirac\, 1)}{\CStep(\geo p\,(2))})}  \\
 & \peq \dots
\end{align*}
\end{example}

\begin{lemma}
 $\LiftDist$ carries a monad structure whose unit is $\CDirac$ and where the Kleisli extension
 $\uext f : \LiftDist A \to \LiftDist B$ of a map $f : A \to \LiftDist B$ satisfies
\begin{align*}
 \uext f(\CStep x) & \peq \CStep(\uext f(x)) &
 \uext f(\CChoice p\mu\nu) & \peq \CChoice p{\uext f(\mu)}{\uext f(\nu)}).
\end{align*}
\end{lemma}
\begin{proof}
 This is a consequence of~\cite[Lemma~16]{ZwartExtraTime}.
\end{proof}


\subsection{Probability of termination}
\label{sec:ProbTerm}

Unfolding the type equivalence $\LiftDist{A} \equi \Dist(A + \LiftDist A)$ we can define maps out of $\LiftDist{A}$ by cases. For example, we can define
the probability of immediate termination $\probtermOp{0} : \LiftDist A \to \IIc$, as:
\begin{align*}
 \probterm{0} {\CDirac a} & = 1 \\
 \probterm{0} {\CStep d} & = 0 \\
 \probterm{0} {\CChoice pxy} & =  p\cdot\probterm{0} x + (1-p)\cdot\probterm{0} y.
\end{align*}
Likewise, we define a function $\run : \LiftDist{A} \to \LiftDist{A}$ that runs a computation for one step, eliminating a single level of $\CStep$ operations:
\begin{align*}
 \run (\CDirac a) & = \CDirac a &
 \run (\CStep d) & = d &
 \run (\CChoice pxy) & =  \CChoice p {(\run \,x)}{(\run\, y)}.
\end{align*}
We can hence compute the probability of termination in $n$ steps, by first running a computation for $n$ steps, and then computing the
probability of immediate termination of the result:
\begin{align*}
 \probterm{n}{x} & = \probterm{0}{\run^n x}.
\end{align*}
For example, running the geometric process for one step gives:
\begin{align*}
 \run(\geo p\,0) & \peq \DistChoice p{(\CDirac\, 0)}{(\DistChoice p{(\CDirac\, 1)}{\CStep(\geo p\,(2))})},
\end{align*}
so $\probterm 1{\geo p \, 0}  \peq p + (1-p)p = 2p - p^2$.


We will later prove (Lemma~\ref{leadsto:lem5}) that $\probtermseq \mu : \NN \to \IIc$ is monotone for all $\mu$. Intuitively, the probability of
termination for $\mu$ is the limit of this sequence. The following definition gives a simple way of comparing such limits, without introducing real
numbers to our type theory. 

\begin{definition}
 Let $f,g \colon \mathbb{N} \to [0,1]$ be monotone. Write $f \leqlim g$ for the following proposition:
 for every $n \colon \mathbb{N}$ and
every rational $\varepsilon > 0$ there exists $m \colon \mathbb{N}$ such that
$f(n) \leq g(m) + \epsilon$. We write $f \eqlim g$ for the conjunction $f \leqlim g$ and $g
\leqlim f$.
\end{definition}
%

It is easy to show that the relation $\leqlim$ is reflexive and transitive and thus $\eqlim$ is an equivalence relation.
Furthermore, $\leqlim$ is closed under pointwise convex sums.
\begin{lemma}\label{lemma:leqlim-conv}
Let $f_1,g_1,f_2,g_2 \colon \mathbb{N} \to [0,1]$ be monotone and let $p \in (0,1)$. 
If $f_1 \leqlim g_1$ and $f_2 \leqlim g_2$, then $p\cdot f_1 + (1-p)\cdot f_2 \leqlim p\cdot g_1 + (1-p)\cdot g_2$.
\end{lemma}

\subsection{Step reductions}
\label{sec:leadsto}
Rather than eliminating a full level of $\CStep$ operations with $\run$, it is sometimes useful to allow different branches to run for a different number of steps. We capture this in the \emph{step reduction relation} $\leadsto$. This is the reflexive, transitive, and convex closure of the one-step reduction relation:
\begin{definition}
Define the step reduction relation $\leadsto : \LiftDist A \to \LiftDist A \to \Prop$ inductively by the following rules
\begin{mathparpagebreakable}
  \inferrule{ }{\nu \leadsto \nu}
  \and
  \inferrule{ }{\CStep \nu \leadsto \nu}
  \and
  \inferrule{\nu \leadsto \nu' \and \nu' \leadsto \nu''}{\nu \leadsto \nu''}
  \and
  \inferrule{\nu_1 \leadsto \nu_1' \and \nu_2 \leadsto \nu_2'}{ \CChoice{p}{\nu_1}{\nu_2} \leadsto \CChoice{p}{\nu_1'}{\nu_2'}}
\end{mathparpagebreakable}
\end{definition}

\begin{remark} \label{rem:leadsto:def}
 CCTT lacks higher inductive families as primitive, but the
 special case of the step reduction relation can be represented by defining
 a fuelled version $\leadsto^n$ by induction on $n$, and existentially quantifying $n$.
\end{remark}

The step reduction relation $\leadsto$ is closely related to $\run$.
\begin{lemma}\label{lemma:leadsto-run}
Let $\nu, \nu' : \LiftDist A $. Then:
\begin{enumerate}
  \item For all $n : \NN$,  $\nu \leadsto \run^n \nu$. \label{lemma:leadsto-run1}
  \item If $\nu \leadsto \nu'$ then, for all $n : \NN$, $\run^{n} \nu \leadsto \run^{n} \nu'$. \label{lemma:leadsto-run2}
  \item If $\nu \leadsto \nu'$ then there exists an $n : \NN$ such that $\nu' \leadsto \run^{n} \nu.$ \label{lemma:leadsto-run3}
\end{enumerate}
\end{lemma}
\begin{proof}
  The first two facts are proven by induction on $n$, the last by induction on $\leadsto$.
\end{proof}

Unlike $\run$, the relation $\leadsto$ can execute different branches of a probabilistic computation for a different number of steps, as illustrated in the following example.
\begin{example}
Consider:
$\nu = \CChoice{p}{(\CStep (\CDirac a))}{(\CStep (\CStep (\CDirac b)))}$.
Then $\run$ removes $\CStep$ operations from both branches. 
\begin{align*}
  \run(\nu) & = \CChoice{p}{(\CDirac a)}{(\CStep (\CDirac b))} \\
  \run^2(\nu) & = \CChoice{p}{(\CDirac a)}{(\CDirac b)} \\
  \run^n(\nu) & = \CChoice{p}{(\CDirac a)}{(\CDirac b)} \text{ for } n \geq 2.
\end{align*}
As per Lemma \ref{lemma:leadsto-run}, we have $\nu \leadsto (\run^n \nu)$, for each $n : \NN$, as well as:
\begin{align*}
\nu & \leadsto \CChoice{p}{(\CDirac a)}{(\CStep( \CStep (\CDirac b)))} \\
\nu & \leadsto \CChoice{p}{(\CStep (\CDirac a))}{(\CStep (\CDirac b))} \\
\nu & \leadsto \CChoice{p}{(\CStep (\CDirac a))}{(\CDirac b)}.
\end{align*}
We will use this flexibility in running different branches for a different number of steps in both our proofs and examples, such as in the proof of Lemma~\ref{lem:prereqs:bind-lemma}.
\end{example}

The step reduction relation is confluent, since we can always match two probabilistic processes that stem from the same parent process by running all branches long enough.

\begin{minipage}{0.6\linewidth}
\begin{corollary}[Confluence]\label{cor:confluence}
$\;$ \\ Let $\nu, \nu_1, \nu_2 : \LiftDist A$. \\ If both $\nu \leadsto \nu_1$ and $\nu \leadsto \nu_2$, \\ then there is a $\nu' : \LiftDist$ such that $\nu_1 \leadsto \nu'$ and $\nu_2 \leadsto \nu'$.
\end{corollary}
\end{minipage}
\hfill
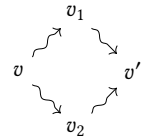
\begin{minipage}{0.39\linewidth}
\scalebox{0.9}{
  \begin{tikzpicture}
    \node (v) at (0,0) {$v$};
    \node (v1) at (0.85,0.85) {$v_1$};
    \node (v2) at (0.85,-0.85) {$v_2$};
    \node (v3) at (1.7,0) {$v'$};
    \draw[->, decorate, decoration={snake, amplitude=.4mm, segment length=3mm, post=lineto, post length=1mm}] 
        (v) -- (v1) node[midway, above right] {};
    \draw[->, decorate, decoration={snake, amplitude=.4mm, segment length=3mm, post=lineto, post length=1mm}] 
        (v) -- (v2) node[midway, below right] {};
    \draw[->, decorate, decoration={snake, amplitude=.4mm, segment length=3mm, post=lineto, post length=1mm}] 
        (v1) -- (v3) node[midway, above right] {};
    \draw[->, decorate, decoration={snake, amplitude=.4mm, segment length=3mm, post=lineto, post length=1mm}] 
        (v2) -- (v3) node[midway, below right] {};
  \end{tikzpicture}
  }
\end{minipage}

\begin{proof}
  We apply Lemma~\ref{lemma:leadsto-run}.\ref{lemma:leadsto-run3}, giving us an $n_1$ and $n_2$ such that $\nu_1 \leadsto \run^{n_1} \nu$ and $\nu_2 \leadsto \run^{n_2} \nu$. Taking $N = \max(n_1, n_2)$ then gives the confluence via Lemma~\ref{lemma:leadsto-run}.\ref{lemma:leadsto-run1}. 
  \end{proof}
%
%
%

We mention two more useful properties of step reductions. Firstly, the bind operation preserves the $\leadsto$ relation, which follows by an easy induction on $\leadsto$.
\begin{lemma}\label{leadsto:lem4}
If $f : A \rightarrow \LiftDist(B)$, and $\nu, \nu' : \LiftDist A$ such that $\nu \leadsto \nu'$. Then also
$\uext f(\nu) \leadsto \uext f(\nu')$.
\end{lemma}

Secondly, the probability of termination is monotone with respect to $n$, and along step reductions:
\begin{lemma}\label{leadsto:lem5}
For $\nu, \nu' : \LiftDist A$, we have:
\begin{enumerate}
  \item For all $n, m : \NN$ such that $n \leq m$: $\probterm{n}{\nu} \leq \probterm{m}{\nu}$. \label{leadsto:lem5-1}
  \item If $\nu \leadsto \nu'$ then for all $n : \NN$: $\probterm{n}{\nu} \leq \probterm{n}{\nu'}$. \label{leadsto:lem5-2}
  \item If $\nu \leadsto \nu'$ then $\probtermseq{\nu} \eqlim \probtermseq{\nu'}$. \label{leadsto:lem5-3}
\end{enumerate}
\end{lemma}
\begin{proof}
  The first statement is by induction on $n$ and case analysis for $\nu$, the second statement by induction on $n$ and $\leadsto$.
  The third statement follows by induction on $\leadsto$, the properties of $\leqlim$ and monotonicity of $\probtermseq{-}$ (i.e. the first statement of this lemma).
\end{proof}

\subsection{Approximate step reductions}
The monads $\GLiftDist$ and $\LiftDist$ model the fact that at any point in time we might not have complete information about a probability distribution, with more information coming in at every time 
step. The relations $\run$ and $\leadsto$ process the incoming information and give new probability distributions that are future variants of the original (delayed) distribution, after a finite number of 
unfoldings. However, it is often useful to talk about what happens in the limit, after infinitely many unfoldings. For this purpose we introduce the relation $\limleadsto$. Intuitively, $v \limleadsto v'$ 
holds if we can get arbitrarily close to $v'$ by applying sufficiently many unfoldings to $v$. 
%

\begin{definition}
Define the \emph{approximate step reduction} relation $\limleadsto : \LiftDist A \to \LiftDist A \to \Prop$ as:
$\nu \limleadsto \nu'$ if for all $p \in \II$ there is a $\nu'' : \LiftDist A$ such that $\nu \leadsto \CChoice{p}{\nu'}{\nu''}$. 
\end{definition}

\begin{example} \label{ex:lem:convergence}
The approximate step reduction relation $\limleadsto$ allows us to talk about limits constructively. For example, say $h_x$ is a hesitant point distribution $h_x$, if 
$h_x \leadsto \CChoice{q}{\CDirac(x)}{h_x}$. For instance, $h_x \peq \Lambda\kappa.(\Fix(\lambda y. \GChoice{q}{\GDirac(x)}{\GStep{y}}))$ satisfies this, but we will see another example in section~\ref{sec:hes:id}. In the limit, $h_x$ should be the same distribution as the normal point distribution $\CDirac(x)$. 
Indeed,
\[h_x \limleadsto \CDirac(x). \]
To prove this, we need to show that for all $p \in \II$ there is a $\nu$ such that $h_x \leadsto \CChoice{p}{\CDirac(x)}{\nu}$.
Let, for any natural number $n$, $q_n = \sum_{i = 0}^{n} q(1-q)^i$, then by applying transitivity of $\leadsto$ $n$ times:
\[
h_x \leadsto \CChoice{q_n}{\CDirac(x)}{h_x}.
\]
If we choose $n$ such that $q_n \geq p$, then we can write 
\[
\CChoice{q_n}{\CDirac(x)}{h_x} \peq \CChoice{p}{\CDirac(x)}{(\CChoice{\frac{q_n - p}{1 - p}}{\CDirac(x)}{h_x})}.
\]
Then choosing $\nu = \CChoice{\frac{q_n - p}{1 - p}}{\CDirac(x)}{h_x}$ gives $h_x \leadsto \CChoice{p}{\CDirac(x)}{\nu}$, which is what we needed to show.
\end{example}

We again prove some useful properties of $\limleadsto$ that we need in future proofs.
We start with the defining properties of the step reductions. These also hold for approximate step reductions:
\begin{lemma}\label{lemma:limleadsto-properties}
 The relation $\limleadsto$ is reflexive and transitive. Furthermore, $\CStep \nu \limleadsto \nu$, and $\limleadsto$ is preserved by sums: if $\nu_1 \limleadsto \nu_1'$ and $\nu_2 \limleadsto \nu_2'$, then for all $p \in \II, \CChoice{p}{\nu_1}{\nu_2} \limleadsto \CChoice{p}{\nu_1'}{\nu_2'}$.
 
\end{lemma}

The bind operation also preserves approximate step reductions, which follows directly from the fact that the bind operation preserves $\leadsto$.
\begin{lemma}\label{lemma:limleadsto-homopreserve}
 If $f : A \rightarrow \LiftDist(B)$, and $\nu, \nu' : \LiftDist A$ such that $\nu \limleadsto \nu'$. Then also
$\uext f(\nu) \limleadsto \uext f(\nu')$.
\end{lemma}

Step reductions are also approximate step reductions, and approximately reducing step reductions just results in further approximate reductions, as does step reducing an approximate step reduction.
\begin{lemma}\label{lemma:limleadsto-leadsto-interact}
Let $\nu, \nu_1, \nu_2 : \LiftDist{A}$. Then:
\begin{enumerate}
  \item If $\nu \leadsto \nu_1$, then $\nu \limleadsto \nu_1$. \label{limleadsto-leadsto-interact1}
  \item If $\nu \leadsto \nu_1$ and $\nu_1 \limleadsto \nu_2$, then $\nu \limleadsto \nu_2$. \label{limleadsto-leadsto-interact2}
  \item If $\nu \limleadsto \nu_1$ and $\nu_1 \leadsto \nu_2$, then $\nu \limleadsto \nu_2$. \label{limleadsto-leadsto-interact3}
\end{enumerate}
\end{lemma}

Finally we show that approximate step reductions preserve limit probabilities of termination:
\begin{lemma}\label{lemma:limleadsto-eqlim}
	Let $\nu_1, \nu_2 : \LiftDist{A}$. If $\nu_1 \limleadsto \nu_2$ then $\probtermseq{\nu_1} \eqlim \probtermseq{\nu_2}$.
\end{lemma}



\section{Probabilistic FPC}
\label{sec:prob:FPC}

In this section we define the language $\pFPC$, its typing rules, and operational semantics in
CCTT. $\pFPC$ is the extension of simply typed lambda calculus with recursive types and probabilistic choice. 
The typing rules of $\pFPC$ are presented in Figure~\ref{fig:pfpc:typing}. In typing judgements, $\Gamma$ is assumed to be
a variable context and all types (which include recursive types of the form $\RecTy\tau$) are closed. 

\begin{figure}[tbp]
  {\small
\begin{mathpar}
  \inferrule*{x:\sigma \in \Gamma \and \vdash \Gamma}{\Gamma \vdash x:\sigma} \and
  \inferrule*{\vdash \Gamma }{\Gamma \vdash \Star : \Unit} \and
  \inferrule*{n : \NN \and \vdash \Gamma}{\Gamma \vdash \Numeral{n} :\Nat} \and
  \inferrule*{\Gamma\vdash L: \Nat\and \Gamma \vdash M : \sigma \and \Gamma \vdash N : \sigma}{\Gamma \vdash\Ifz{L}{M}{N} : \sigma} \and
  \inferrule*{\Gamma \vdash M : \Nat}{\Gamma \vdash \Suc{M}:\Nat} \and
  \inferrule*{\Gamma \vdash M : \sigma}{\Gamma \vdash \Inl{M}:\CoprodTy\sigma\tau} \and
  \inferrule*{\Gamma\vdash L: \CoprodTy\sigma{\sigma'}\and \Gamma,x:\sigma \vdash M : \tau \and \Gamma,y:\sigma' \vdash N : \tau}{\Gamma \vdash\Case{L}{x.M}{y.N} : \tau} \and
  \inferrule*{\Gamma \vdash M : \sigma \and \Gamma \vdash N : \tau}{\Gamma \vdash \Pair{M}{N}:\ProdTy\sigma\tau} \and
  \inferrule*{\Gamma \vdash M : \ProdTy\sigma\tau}{\Gamma \vdash \Fst{M}:\sigma} \and
  \inferrule*{\Gamma,(x : \sigma) \vdash M :\tau}{\Gamma\vdash \Lam{M} :\sigma\to \tau} \and
  \inferrule*{\Gamma \vdash N:\sigma\to \tau \and \Gamma \vdash M:\sigma}{\Gamma\vdash N M : \tau} \and
  \inferrule*{\Gamma \vdash M : \tau[\RecTy\tau/X]}{\Gamma \vdash \Fold{M}:\RecTy\tau} \and
  \inferrule*{\Gamma \vdash M : \RecTy\tau}{\Gamma \vdash \Unfold{M}:\tau[\RecTy\tau/X]} \and
  \inferrule*{\Gamma \vdash M : \sigma \and \Gamma \vdash N : \sigma \and p : \II}{\Gamma \vdash \Choice{p}{M}{N}:\sigma}
\end{mathpar}
}
\caption{Typing rules for $\pFPC$. Rules for $\Pred{}, \Inr{},\Snd{}$ omitted.}
\label{fig:pfpc:typing}
\end{figure}

We assume a given representation of terms and types of $\pFPC$ within CCTT. For example, 
as inductive types 
with the notion of closedness defined as decidable properties 
on these types. We will also assume that terms are annotated with enough types so that one can deduce the type 
of subterms from terms. This allows the evaluation function and the denotational semantics to be defined by
induction on terms rather than typing judgement derivations. It also means the typing judgement 
$\Gamma \vdash M : \sigma$ is decidable. We write $\Ty$ for the type of closed types, and $\Tm[\Gamma]\sigma$
for the type of terms $M$ satisfying $\Gamma \vdash M : \sigma$, constructed as a subtype of the type of all terms.
For empty $\Gamma$ we write simply $\Tm\sigma$.  We write 
$\Val{\sigma}$ for the set of closed values of type $\sigma$, as captured by the grammar
\begin{align*}
 V,W & := \Star\mid \Numeral{n} \mid \Lam{M} \mid \Fold{V} \mid \Pair{V}{W} \mid \Inl{V} \mid \Inr{V}
\end{align*}


The \emph{operational semantics} is given by a function $\GEval$ of type
\begin{equation} \label{eq:GEval}
\GEval : \{\sigma:\Ty\} \to \Tm{\sigma} \to \GLiftDist(\Val{\sigma}),
\end{equation}
associating to a term $M$ with type $\sigma$ 
an element of the free guarded convex delay algebra over $\Val{\sigma}$. 
The corresponding element of the coinductive convex delay algebra
can be defined as 
\[
 \Eval (M) \defeq \Lambda\kappa. \GEval(M) : \LiftDist{}(\Val{\sigma}).
\]
The function $\GEval$ is defined by an outer guarded recursion and an inner
induction on terms. In other words, by first assuming given a term of type  
$\later^\kappa(\{\sigma:\Ty\} \to \Tm{\sigma} \to \GLiftDist(\Val{\sigma}))$, then defining $\GEval(M)$
by induction on $M$.  
The cases are given in Figure~\ref{fig:pfpc:eval}. The figure 
overloads names of some term constructors to functions of values, e.g.
$\SynInl : \Val\sigma \to \Val{\sigma + \tau}$ mapping $V$ to $\Inl V$, and $\SynPred, \SynSuc : \Val\Nat \to \Val\Nat$
defined using $\Val\Nat \equi \NN$. The figure also uses pattern matching within bindings,
e.g. in the case of function application, where we use that all values of function type are of the form $\Lam{M'}$ for some $M'$.

In the cases of $\Case{L}{x.M}{y.N}$ and $M\,N$ the recursive calls
to $\GEval$ are under $\GStep$ and so can be justified by guarded recursion. This is necessary,
because they are not instantiated at a structurally smaller term. The case
for $\Unfold{M}$ also introduces a computation step using $\GStep$. While 
this is not strictly necessary to define the operational
semantics, we introduce it to simulate the steps of the denotational
semantics defined in Section~\ref{sec:den:sem}. This simplifies the proof of Theorem~\ref{thm:eq-pterm-denot-op}.

\begin{figure}[tbp]
\begin{align*}
 \GEval (V) & \defeq \GDirac V \\
 \GEval (\Suc {M}) & \defeq \GLiftDist(\SynSuc)(\GEval M) \\
 \GEval(\Inl{M}) & \defeq \GLiftDist(\SynInl)(\GEval M) \\
 \GEval(\Ifz{L}{M}{N}) & \defeq \GEval{L} \GBind 
    \SemCase{}{\Zero}{\GEval(M)}{\Numeral{n+1}}{\GEval(N)} \\
 \GEval(\Fst{M}) & \defeq (\GEval M) \GBind \lambda\Pair VW. \GDirac V \\ 
 \GEval{\Pair{M}{N}} & \defeq \GEval{M} \GBind \lambda V. 
 \GEval{N}  \GBind \lambda W. \GDirac{\Pair{V}{W}} \\
 \GEval(\Case{L}{x.M}{y.N}) & \defeq \GEval(L) \GBind 
      \SemCase{}{\Inl{V}}{\GStep{}(\tabs\tickA\kappa\GEval(M[V/x]))}
        {\Inr{V}}
        {\GStep{}(\tabs\tickA\kappa\GEval(N[V/y]))} \\
 \GEval (M\,N) & \defeq
      \GEval(M) \GBind \lambda (\Lam{M'}). \ \GEval(N) \\ & \phantom{\defeq\defeq} \GBind 
      \lambda V. \GStep{\parens{\tabs\tickA\kappa\GEval(M'[V/x])}} \\
 \GEval(\Fold{M}) & \defeq \GLiftDist(\Fold) (\GEval(M)) \\
 \GEval(\Unfold{M}) & \defeq  \GEval M \GBind \lambda (\Fold{V}). \GStep(\tabs\tickA\kappa{\GDirac V)} \\
 \GEval{(\Choice{p}{M}{N})} & \defeq \GChoice{p}{(\GEval{M})}{(\GEval N)} 
\end{align*}
\caption{The evaluation function $\GEval{}$. Rules for $\Pred{}, \Inr{},\Snd{}$ omitted.}
\label{fig:pfpc:eval}
\end{figure}

In $\pFPC$ we can define a fixpoint combinator as illustrated below:

\begin{example}\label{ex:Y:comb}
  We define $\EmptyCtx \vdash \YComb{} : \FnTy{(\FnTy{(\FnTy{\sigma}{\tau})}{(\FnTy{\sigma}{\tau})})}{\FnTy\sigma\tau}$ by
  $\YComb \defeq\ \Lam[f]{\Lam[z]{e_f(\Fold{e_f})z}}$, where 
  \begin{align*}
    e_f & : \FnTy{(\RecTy{(\FnTy{X}{\FnTy{\sigma}\tau})})}{\FnTy\sigma\tau} \\
    e_f & \defeq \Lam[y]{\LetIn{y'}{\Unfold{y}}f(\Lam{y' y x})}
  \end{align*}
  Here, $y : \RecTy{(\FnTy{X}{\FnTy{\sigma}\tau})}$,
 $y' : \FnTy{(\RecTy{(\FnTy{X}{\FnTy{\sigma}\tau})})}{\FnTy\sigma\tau}$
 and $\Lam{y' y x} : \FnTy\sigma\tau$.

Then, for any values $\EmptyCtx \vdash \varphi : \FnTy{(\FnTy\sigma\tau)}{\FnTy\sigma\tau}$ and $\EmptyCtx \vdash V : \sigma$,
  \[
	  \GEval{((\YComb{\varphi})(V))} \peq (\ExpStep)^{4}(\GEval{((\varphi(\YComb{\varphi}))(V))})
  \]
  where $\ExpStep \triangleq (\GStep \circ \GNext)$.
\end{example}


\subsection{Contextual Refinement}

Two terms $M,N$ are contextually equivalent if for any closing context $C[-]$ of ground type
the terms $C[M]$ and $C[N]$ have the same observational behaviour. Here, 
the only observable behaviour we consider is the probability of termination for programs of type $\Unit$, which, for a 
closed term $M$ should be the limit of the sequence $\probterm n {\Eval\, M}$. 

\begin{definition}[Contextual Refinement]
\label{def:ctx:ref}
 Let $\Gamma \vdash M,N : \tau$ be terms. We say that $M$ \emph{contextually refines} $N$
 if for any closing context of unit type $C:\CtxTy\cdot\Unit$,
	$\probtermseq{\Eval(C[M])} \leqlim \probtermseq{\Eval(C[N])}$.
 In this case, write $M \CtxRef N$. We say that $M$ and $N$ are contextually equivalent ($M \CtxEq N$) if 
 $M \CtxRef N$ and $N \CtxRef M$. 
\end{definition}



Definition~\ref{def:ctx:ref} should be read as defining a predicate on $M$ and $N$ in CCTT, \emph{i.e.}, a 
function $\Tm[\Gamma]\sigma \to \Tm[\Gamma]\sigma \to \Prop$.
The notion of closing context $C:\CtxTy\cdot\Unit$
is a special case of a typing judgement on contexts $C:\CtxTy\Delta\tau$ defined in a standard way. 
Note that $C[-]$ may capture free variables in $M$. 

\section{Denotational semantics}
\label{sec:den:sem}

\begin{figure}[tbp]
\begin{center}
 \textbf{Interpretation of types}
\end{center}
\begin{align*}
  \GSVal{ \Nat} & \defeq \mathbb{N} & \GSVal{\ProdTy\sigma\tau} & \defeq \GSVal{\sigma} \times \GSVal{\tau} &
  \GSVal{\FnTy\sigma\tau} & \defeq \GSVal{\sigma} \to \GDomain{\tau} \\
  \GSVal{\Unit} & \defeq \SemUnit & \GSVal{\CoprodTy\sigma\tau} & \defeq \GSVal{\sigma} + \GSVal{\tau} &
  \GSVal{\RecTy\tau} &\defeq \Later{\GSVal{\tau[\RecTy\tau/X]}}
\end{align*}
\begin{center}
 \textbf{Interpretation of terms}
\end{center}
\begin{align*}
    \GInterpret{\Star}\rho \defeq \GDirac(\star) \quad
    \GInterpret{\Numeral{n}}\rho &\defeq \GDirac(n) \quad
	\GInterpret{x}\rho \defeq \GDirac(\rho(x)) \\
    \GInterpret{\Suc{M}}\rho &\defeq \GLiftDist{(\SemSuc{})}\left(\GInterpret{\Gamma \vdash M : \Nat}\rho\right)\\
    \GInterpret{\Ifz{t}{M}{N}}\rho &\defeq  \GInterpret{L}{\rho} \GBind \DoIfz{\GInterpret{M}{\rho}}{\GInterpret{N}{\rho}} \\
    \GInterpret{\Lam{M}}\rho &\defeq\GDirac\left(\lambda (v:\GValInterpret{\sigma}{\rho}). \GInterpret{M}{\rho.x\mapsto v}\right) \\
    \GInterpret{M\,N}\rho &\defeq \GInterpret{M}\rho \GBind \lambda f. \GInterpret{N}\rho \GBind \lambda v. fv  \\
    \GInterpret{\Pair{M}{N}}\rho &\defeq \GInterpret{M}\rho \GBind \lambda v. \left(\GInterpret{N}\rho \GBind \lambda w. \GDirac\SemPair{v}{w}\right)\\
    \GInterpret{\Fst{M}}\rho &\defeq \GLiftDist{(\SemFst{})}\left(\GInterpret{M}\rho\right)\\
    \GInterpret{\Inl{M}}\rho &\defeq \GLiftDist{(\SemInl{})}\left(\GInterpret{M}\rho\right)\\
    \GInterpret{\Case{L}{x.M}{y.N}}\rho &\defeq \GInterpret{L}\rho \GBind 
    \begin{cases}
      \SemInl{v}\mapsto {\GInterpret{M}{\rho.x\mapsto v}}\\
      \SemInr{v}\mapsto {\GInterpret{N}{\rho.y\mapsto v}}
    \end{cases}
    \\ 
    \GInterpret{\Fold{M}}{\rho}&\defeq \GLiftDist{(\GNext)}(\GInterpret{M}\rho) \\
    \GInterpret{\Unfold{M}}{\rho}&\defeq \GInterpret{M}\rho \GBind\lambda v. \GStep (\tabs\tickA\kappa \GDirac{(\tapp v)}) \\
    \GInterpret{\Choice{p}{M}{N}}{\rho}&\defeq \GChoice{p}{\GInterpret{M}\rho}{\GInterpret{N}\rho}
\end{align*}
\caption{Denotational semantics for $\pFPC$. Rules for $\Pred{}, \Inr{},\Snd{}$ omitted.} 
\label{fig:den:sem}
\end{figure}

We now define a \emph{denotational semantics} for $\ProbFPC$.
The semantics is based on two ideas: (i) we use guarded recursion to interpret recursive types,
inspired by \cite{DBLP:journals/mscs/MogelbergP19}, and (ii) we use the guarded convex delay monad
to interpret effectful computations, similarly to how monads are used to interpret effectful languages in standard denotational semantics. 
Specifically we define two functions 
\begin{align*}
   \GSVal{-} & : \Ty \to \Set \\
   \GInterpret{-}{-} & : \{\Gamma : \Ctx\} \to \{ \sigma:\Ty\} \to \Tm[\Gamma]\sigma \to \GSVal{\Gamma}\to \GDomain\sigma
\end{align*}
which interpret types and terms, respectively, see Figure~\ref{fig:den:sem}. Note that by using guarded recursion to interpret types,
it suffices to define the denotation of closed types. 

The interpretation of terms is defined by induction on terms.
%
Note that all syntactic values are interpreted
as semantic values in the sense that we can define a map 
\[
\GValInterpret{-}{} : \{\sigma : \Ty\} \to \Val\sigma \to \GSVal\sigma
\]
satisfying $\GInterpret{V}{} = \GDirac (\GValInterpret{V}{})$ for all $V$. Note also that
the denotational semantics only steps when unfolding from a recursive type. As a consequence,
the denotational semantics is much less cluttered by steps than the operational semantics,
which makes it simpler to use for reasoning. Despite the fewer steps, the two semantics
refine each other. In particular we have the theorem below, which we will prove using
the techniques we develop in later sections:
\begin{theorem}
  \label{thm:eq-pterm-denot-op}
  For any well-typed closed expression $\cdot \vdash M : \Unit$ we have that
  \[
	  \probtermseq{\Eval \,M} \eqlim \probtermseq{\Interpret{M}{}}
  \]
  where $\Interpret{M}{} \defeq \Lambda\kappa. \GInterpret{M}{}$
\end{theorem}

The denotational semantics satisfies the following substitution lemma:
\begin{lemma}[Substitution Lemma]
  \label{lem:subst-lemma}
 For any well-typed term $\Gamma.(x:\sigma)\vdash M:\tau$ as well as every well typed value $\Gamma\vdash V:\sigma$ we have
 \[
 \GInterpret{M[V/x]}{\rho} \peq \GInterpret{M}{\rho.x\mapsto\GValInterpret{V}{\rho}}
 \]
\end{lemma}



The following example presents a useful equality to work with denotational semantics of recursive
functions defined via a fixed point combinator:
\begin{example}
 Recall the fixed point operator $\YComb{}$ from Example~\ref{ex:Y:comb}. The denotational semantics
 satisfies the following equations
 \begin{align} 
	 \GInterpret{(\Lam[y]{M})\,x}{x\mapsto v} &\peq \GInterpret{M}{y\mapsto v} \label{eq:den:beta:red} \\
	 \GInterpret{\YComb{f}\,x}{x\mapsto v} &\peq (\GStep{}\circ \GNext)(\GInterpret{f(\YComb{f})\,x}{x\mapsto v}) \label{eq:den:Yf:open}
\end{align}
for any value $f$ and semantic value $v$, and where $x$ does not appear free in $M$. By applying the substitution lemma, we also get
   	\begin{align*}
		\GInterpret{(\Lam[y]{M})\,V}{} &\peq \GInterpret{M[V/x]}{} \\
		\GInterpret{\YComb{f}\,V}{} &\peq (\GStep{}\circ \GNext)(\GInterpret{f(\YComb{f})\,V}{})
	\end{align*}
\end{example}



\section{Couplings and Lifting relations}
\label{sec:couplings}
Our next goal is to define a logical relation between the syntax and semantics and use it to reason about contextual equivalence of $\pFPC$ terms. 
This is done by defining a value relation (of type $\GSVal\sigma \to \Val\sigma \to \Prop$) and 
an expression relation (of type $\GLiftDist{}{\GSVal\sigma} \to \LiftDist{}(\Val\sigma) \to \Prop$) by mutual recursion. 
We start by defining a relational lifting~\cite{katsumata2013preorders} operation mapping a relation
$\Rel : A \to B \to \Prop$ on values and to a relation $\GRelLift{\Rel} : \GLiftDist{A}\to \LiftDist{B} \to \Prop$ on computations.
This will be used to define the expression relation as the relational lifting of the value relation. 
As opposed to standard lifting constructions, that consider the same monad on both sides of the relation, there is an asymmetry in the type of $\GRelLift{\Rel}$.
The reason is that the logical relation will be defined by guarded recursion in the first argument, while using arbitrarily deep unfoldings of the term on the right-hand side.
This section defines the relational lifting construction and proves its basic properties.

The first step is lifting a relation $\Rel$ over $A \times B$ to a relation over $\Dist{A} \times \Dist{B}$.
Recall the notion of $\Rel$-coupling~\cite{lindvall_lectures_2002,DBLP:conf/popl/BartheGB09}, which achieves this precise end:
%
\begin{definition} \label{def:coupling}
 Let $\Rel : A \to B \to \Prop$, and let $\mu : \Dist A, \nu : \Dist B$ be finite distributions. An $\Rel$-coupling between $\mu$ and $\nu$ is a distribution on the total space of $\Rel$, \emph{i.e.}, $\rho : \Dist(\Sigma (a \!:\! A), (b \!: \!B) . a \Rel b)$, whose marginals are $\mu$ and $\nu$:
\begin{align*}
 \DistMarginalFst(\rho) & \peq \mu & \DistMarginalSnd(\rho) \peq \nu,
\end{align*}
where $\mathsf{pr}_1$ and $\mathsf{pr}_2$ are the projection maps $A \xleftarrow{\mathsf{pr}_1} A \times B \xrightarrow{\mathsf{pr}_2} B$.
We write $\CoupType\Rel\mu\nu$ for the type of $\Rel$-couplings between $\mu$ and $\nu$.
\end{definition}

In other words, an $\Rel$-coupling is a joint distribution over $A\times B$ with the property that $\Rel$ always holds for any pair of
values sampled from it.
Our definition of $\mu \GRelLift{\Rel} \nu$ will generalise this idea for pairs of computations, whose final
values might become available at different times.
Recall that by Theorem~\ref{thm:dist:sum:equiv} a distribution
$\mu : \GLiftDist{A} \equi \Dist(A + \Later(\GLiftDist A))$ must be either a distribution of values, one
of delayed computations or a convex combination of the two.
%
The relation $\mu \GRelLift{\Rel} \nu$ should then hold if (1) the values available
in $\mu$ now can be matched by $\nu$ after possibly some computation steps, and (2) the
delayed computation part of $\mu$ can be matched later, in a guarded recursive step. The matching
of values is done via an $\Rel$-coupling, and the computation steps are interpreted using $\limleadsto$. We choose $\limleadsto$ over the more straightforward  $\leadsto$ to allow the matching values for $\mu$ to be delivered in the limit of an infinite sequence of computation steps. By choosing $\limleadsto$, we just require that any fraction of it can be delivered in finite time. 
In the following definition, we leave the inclusions $\Dist(\inl)$ and $\Dist(\inr)$ implicit and simply write, e.g., $\mu : \Dist{A}$ for the first
case mentioned above.

\begin{definition}[Lifting of Relations]\label{def:liftRel}
Given a relation $\Rel : A \to B \to \Prop$ we define its lift $\GRelLift{\Rel} : \GLiftDist{A}\to \LiftDist{B} \to \Prop$ as:
$\mu \GRelLift{\Rel} \nu$ if one of the following three options is true:
\begin{enumerate}
\item\label{item:val}  $\mu : \Dist{A}$, there is a $\nu' : \Dist{B}$ such that $\nu \limleadsto \nu'$, and there is a $\rho: \CoupType\Rel{\mu}{\nu'}$. 
\item\label{item:delay} $\mu : \Dist(\Later{}\GLiftDist{A})$ and $\Later{(\alpha : \kappa)}\  (((\Gzeta (\mu))[\alpha]) \GRelLift{\Rel} \nu)$,\\
where $\Gzeta\, : \Dist(\Later{}\GLiftDist{A}) \rightarrow \Later{}\GLiftDist{A}$ is defined as
\begin{align*}
   \Gzeta (\Dirac{x}) & \defeq x & 
   \Gzeta (\DistChoice{p}{\mu}{\nu}) & \defeq \tabs{\tickA}{\kappa} \GChoice{p}{\tapp[\tickA]{\Gzeta(\mu)}}{\tapp[\tickA]{\Gzeta(\nu)}}
\end{align*}
\item\label{item:both} There exist $\mu_1 : \Dist{A}, \mu_2 : \Dist(\Later{}\GLiftDist{A})$, and $p \in \II$ such that $\mu = \DistChoice{p}{\mu_1}{\mu_2}$, there exist $\nu_1, \nu_2 : \LiftDist{B}$ such that $\nu \limleadsto \CChoice{p}{\nu_1}{\nu_2}$, and $\mu_1 \GRelLift{\Rel} \nu_1$ and $\mu_2 \GRelLift{\Rel} \nu_2$.
\end{enumerate}
We then define $\CRelLift{\Rel} : \LiftDist{A} \to \LiftDist{B} \to \Prop$ as:
\[
\mu \CRelLift{\Rel} \nu \defeq \forall \kappa . (\mu[\kappa] \GRelLift{\Rel} \nu) .
\]
\end{definition}
The following is a consequence of the encoding of coinductive types using guarded recursion.

\begin{lemma} \label{lem:CRelLift:coinductive}
 If $A$ is clock irrelevant, then
 $\CRelLift{\Rel}$ is the coinductive solution to the equation that $\mu \CRelLift{\Rel} \nu$ if and only if
 one of the following:
 \begin{enumerate}
\item\label{item:Cval}  $\mu : \Dist{A}$, there is a $\nu' : \Dist{B}$ such that $\nu \limleadsto \nu'$, and there is a $\rho: \CoupType\Rel{\mu}{\nu'}$.
\item\label{item:Cdelay}  $\mu : \Dist(\LiftDist{A})$ and $\run(\mu) \CRelLift{\Rel} \nu$.
\item\label{item:Cboth} There exist $\mu_1 : \Dist{A}, \mu_2 : \Dist(\LiftDist{A})$, and $p : \II$ such that
$\mu \peq \DistChoice{p}{\mu_1}{\mu_2}$, there exist $\nu_1, \nu_2 : \LiftDist{B}$ such that $\nu \limleadsto \CChoice{p}{\nu_1}{\nu_2}$, and $\mu_1 \CRelLift{\Rel} \nu_1$ and $\mu_2 \CRelLift{\Rel} \nu_2$.
\end{enumerate}
\end{lemma}

The following lemma allows us to use existence of a lifting to prove an inequality between probabilities of termination, which will be useful when reasoning about contextual refinement
\begin{lemma}
\label{lem:eq1-probterm}
Let $\mu, \nu : \LiftDist{1}$, and let $\eqrel : 1 \to 1 \to \Prop$ be the identity relation, relating the unique element to itself. Then,
\[
\mu \CRelLift{\eqrel} \nu \Rightarrow \probtermseq{\mu} \leqlim \probtermseq{\nu}.
\]
\end{lemma}

In the remainder of this section, we will prove some useful reasoning principles for working with $\GRelLift{\Rel}$
at a more abstract level, without the need for unfolding the definition.
These are summarized in Figure~\ref{fig:rules:rel:lift}, where the double bar indicates that the rule
is bidirectional.

\begin{figure}[h]
\begin{mathpar}
  \inferrule{\nu \leadsto \nu' \and \mu \GRelLift{\Rel} \nu}{\mu \GRelLift{\Rel} \nu'}
  \and
  \inferrule{\nu \leadsto \nu' \and \mu \GRelLift{\Rel} \nu'}{\mu \GRelLift{\Rel} \nu}
  \and
  \inferrule{\nu \limleadsto \nu' \and \mu \GRelLift{\Rel} \nu'}{\mu \GRelLift{\Rel} \nu}
  \and
  {\mprset { fraction ={===}}
	\inferrule{\GChoice{p}{(\GStep \mu_1)}{(\GStep \mu_2)} \GRelLift{\Rel} \nu}
	{\GStep(\tabs{\tickA}{\kappa} \GChoice{p}{(\tapp[\tickA]{\mu_1})}{(\tapp[\tickA]{\mu_2})})  \GRelLift{\Rel} \nu}}
  \and
  {\mprset { fraction ={===}}
	\inferrule{\mu \GRelLift{\Rel} \CStep(\CChoice p {\nu_1}{\nu_2})}
	{\mu \GRelLift{\Rel} \CChoice p {(\CStep(\nu_1))}{(\CStep(\nu_2))}}}
  \and
  \inferrule{\mu_1 \GRelLift{\Rel} \nu_1 \and \mu_2 \GRelLift{\Rel} \nu_2}
  {(\GChoice{p}{\mu_1}{\mu_2}) \GRelLift{\Rel} (\CChoice{p}{\nu_1}{\nu_2})}
  \and
  \inferrule{\mu \GRelLift{\Rel} \uext \nu \and \forall a,b. a \Rel b \rightarrow f(a) \GRelLift{\mathcal{S}} g(b)}
  {\uext f(\mu) \GRelLift{\mathcal{S}}  \uext g(\nu)}
  \end{mathpar}

  \caption{Reasoning principles for $\GRelLift{\Rel}$}
\label{fig:rules:rel:lift}
\end{figure}

First we show that $\GRelLift{\Rel}$ is invariant with respect to $\leadsto$ and $\limleadsto$-reductions on its second
argument.

\begin{lemma}\label{lem:prereqs:leadsto-liftrel}
 If $\nu \leadsto \nu'$ then $\mu \GRelLift{\Rel} \nu$ iff $\mu \GRelLift{\Rel} \nu'$.
 In particular $\mu \GRelLift{\Rel} (\CStep{}\nu)$ iff $\mu\GRelLift{\Rel} \nu$.
\end{lemma}
\begin{proof}
Left to right follows from Lemma~\ref{lemma:leadsto-run}, and right to left is by Lemma~\ref{lemma:limleadsto-leadsto-interact}.
\end{proof}

\begin{lemma}\label{lem:prereqs:limleadsto-liftrel}
If $\nu \limleadsto \nu'$ and $\mu \GRelLift{\Rel} \nu'$, then also $\mu \GRelLift{\Rel} \nu$.
\end{lemma}
\begin{proof}
This follows from transitivity of $\limleadsto$.
\end{proof}

The lemma below allows us to commute computation steps with probabilistic choices
on either side of a lifted relation:
\begin{lemma}\label{lem:prereqs:step-choice}
Let $\Rel : A \to B \to \Prop$.
\begin{enumerate}
\item If  $\mu_1,\mu_2 : \Later\GLiftDist{A}$, $\nu : \LiftDist{B}$ then 
\[ \GChoice{p}{(\GStep \mu_1)}{(\GStep \mu_2)} \GRelLift{\Rel} \nu \text{ iff  } \GStep(\tabs{\tickA}{\kappa} \GChoice{p}{(\tapp[\tickA]{\mu_1})}{(\tapp[\tickA]{\mu_2})}) \GRelLift{\Rel} \nu. \]
\item If $\mu : \GLiftDist{A}$ and $\nu_1, \nu_2 : \LiftDist{B}$ then $\mu \GRelLift{\Rel} \CStep(\CChoice p {\nu_1}{\nu_2})$ iff
$\mu \GRelLift{\Rel} \CChoice p {(\CStep(\nu_1))}{(\CStep(\nu_2))}$.
\end{enumerate}
\end{lemma}

\begin{proof}
  The first statement follows from applying Definition~\ref{def:liftRel}(\ref{item:delay}). The second statement follows from Lemma~\ref{lem:prereqs:leadsto-liftrel}.
\end{proof}

Liftings are a useful technique to reason about computations because of the way they
interact with choice, and with the monad structures of $\GLiftDist$ and $\LiftDist$.
First, liftings are closed under choice operators, which allows us to construct them
by parts:

\begin{lemma}\label{lem:prereqs:choice-lemma}
  Let $\Rel : A \to B \to \Prop$, and suppose that $\mu_1 \GRelLift{\Rel} \nu_1$ and $\mu_2 \GRelLift{\Rel} \nu_2$. Then also
  $(\GChoice{p}{\mu_1}{\mu_2}) \GRelLift{\Rel} (\CChoice{p}{\nu_1}{\nu_2})$.
\end{lemma}

\begin{proof}[Proof (sketch)]
There are 9 cases to consider, one for each of the possible combinations of cases for $\mu_1 \GRelLift{\Rel} \nu_1$ and $\mu_2 \GRelLift{\Rel} \nu_2$. The proof for each of these cases follows quite directly from the assumptions that we get from $\mu_1 \GRelLift{\Rel} \nu_1$ and $\mu_2 \GRelLift{\Rel} \nu_2$, after some rewriting of terms using the axioms of convex algebras. For example, if $\mu_1 : \Dist{A}$ and for $\mu_2$ there exist $\mu_{21} : \Dist{A}, \mu_{22} : \Dist(\Later{}\GLiftDist{A})$ such that $\mu_2 \peq \DistChoice{p'}{\mu_{21}}{\mu_{22}}$, then:
         \begin{align*}
           \GChoice{p}{\mu_1}{\mu_2} & \peq \GChoice{p}{\mu_1}{(\GChoice{p'}{\mu_{21}}{\mu_{22}})}  
             \peq \GChoice{p + (1-p)p'}{(\GChoice{\frac{p}{p + (1-p)p'}}{\mu_1}{\mu_{21}})}{\mu_{22}},
         \end{align*}
         where $\GChoice{\frac{p}{p + (1-p)p'}}{\mu_1}{\mu_{21}} : \Dist{A}$ and $\mu_{22} : \Dist(\Later{}\GLiftDist{A})$.
We can then combine information coming from the assumptions $\mu_1 \GRelLift{\Rel} \nu_1$ and $\mu_2 \GRelLift{\Rel} \nu_2$ using the same probabilities as above to reach the desired conclusion. 
\end{proof}

Second, we can prove a bind lemma, that allows us to sequence computations
related by liftings. This lemma will be crucial e.g. in the proof of the
fundamental lemma in the following section.

\begin{lemma}[Bind Lemma]\label{lem:prereqs:eta-lemma}\label{lem:prereqs:bind-lemma}
Let $\Rel : A \to B \to \Prop$ and $\mathcal{S} : A' \to B' \to \Prop$.
 \begin{enumerate}
\item If $f : A \to \GLiftDist{A'}$ and $g : B \to \LiftDist{B'}$ satisfy $f(a) \GRelLift{\mathcal{S}} g(b)$ whenever $a \Rel b$, then for all
$\mu : \Dist{A}$ and $\nu : \Dist{B}$, if there is a $\rho: \CoupType\Rel{\mu}{\nu}$, then 
$\uext f(\mu) \GRelLift{\mathcal{S}}  \uext g(\nu)$.
\item If $f : A \to \GLiftDist{A'}$ and $g : B \to \LiftDist{B'}$ satisfy $f(a) \GRelLift{\mathcal{S}} g(b)$ whenever $a \Rel b$, then for all
$\mu : \GLiftDist{A}$ and $\nu : \LiftDist{B}$ satisfying $\mu \GRelLift{\Rel}  \uext \nu$, also
$\uext f(\mu) \GRelLift{\mathcal{S}}  \uext g(\nu)$.
\end{enumerate}
\end{lemma}

\begin{proof}[ Proof (sketch)]
 The first statement is by induction on the coupling $\rho: \CoupType\Rel{\mu}{\nu}$.
 
  For the second statement, suppose that $\mu\GRelLift{\Rel} \nu$.
 By the definition of $\GRelLift{\Rel}\!\!$, there are three cases to consider.
 We only show the first case here. The second case is by guarded recursion, and the third case follows from the first and second cases.

 For the first case, we assume that $\mu : \Dist{A}$, that there exist $\nu_1 : \Dist{B}$ such that $\nu \limleadsto \nu_1$, and that there exists a $\rho: \CoupType\Rel{\mu}{\nu_1}$. What we want is to show that $\uext f(\mu) \GRelLift{\mathcal{S}} \uext g(\nu)$.
 
 Notice that since the $\limleadsto$ relation is preserved by homomorphisms (Lemma~\ref{lemma:limleadsto-homopreserve}), $\nu \limleadsto \nu_1$ implies $\uext g(\nu) \limleadsto \uext g(\nu_1)$. By Lemma~\ref{lem:prereqs:limleadsto-liftrel} it is therefore enough to show that $\uext f(\mu) \GRelLift{\mathcal{S}} \uext g(\nu_1)$, which follows directly from the first statement of this lemma.
\end{proof}



\section{Relating syntax and semantics}
\label{sec:rel:syntax:semantics}

In this section we prove that our denotational semantics is adequate with
respect to the operational semantics. The overall approach is standard in that
we define type-indexed logical relations between semantic denotations and terms
of the operational semantics, and we define both value relations and expression (aka
term or computation) relations:
\begin{align*}
\GValLogRel\sigma & :\GSVal\sigma \to \Val\sigma \to \Prop \\
\GLogRel\sigma & :\GLiftDist{}{\GSVal\sigma} \to \LiftDist{}(\Val\sigma) \to \Prop
\end{align*}
There are two things to note, however. The first is that traditionally it is
tricky to define such logical relations for a programming language with
recursive types; see, e.g., \cite{DBLP:journals/iandc/Pitts96}. Here we follow
\cite{DBLP:journals/mscs/MogelbergP19} and simply use guarded recursion (and
induction on types), see Figure~\ref{fig:rel:syntax:semantics}. Notice in particular
the case for recursive types, where guarded recursion is used. The second point
to note is that the definition of the term relation is novel: it is defined as the
\emph{lifting} (in the sense of the previous section) of the value relation.
This allows us to use the reasoning principles associated with lifting (e.g.,
Lemmas~\ref{lem:prereqs:leadsto-liftrel},\ref{lem:prereqs:choice-lemma}
and~\ref{lem:prereqs:bind-lemma}) in proofs and examples.


\begin{figure}
\textbf{Value relation}
\begin{mathpar}
  \inferrule{ }
  {n\GValLogRel\Nat \Numeral{n}} 
  \and
  \inferrule{ }{\star \GValLogRel\Unit \Star}
  \and
  \inferrule{v \GValLogRel\sigma V \and w \GValLogRel\tau W }{\SemPair{v}{w} \GValLogRel{\ProdTy\sigma\tau} \Pair{V}{W}}
  \and
  \inferrule{v \GValLogRel\sigma V}{\SemInl{v} \GValLogRel{\CoprodTy\sigma\tau} \Inl{V}}
  \and
  \inferrule{v \GValLogRel\tau V}{\SemInr{v} \GValLogRel{\CoprodTy\sigma\tau} \Inr{V}}
  \and
  \inferrule{\forall w,V. w \GValLogRel\sigma V \to v(w) \GLogRel\tau \Eval{(M[V/x])}}{v \GValLogRel{\FnTy\sigma\tau} \Lam{M}}
  \and
  \inferrule{\ClockedLater{(v[\alpha]\GValLogRel{\tau[\RecTy\tau/X]} V)} }{v \GValLogRel{\RecTy\tau} \Fold{V}}
  \end{mathpar}

\textbf{Expression relation}
\[
 \mu \GLogRel\sigma d \defeq \mu \GRelLift{\GValLogRel\sigma} d
\]
  \caption{Logical Relation.}
\label{fig:rel:syntax:semantics}
\end{figure}

We extend the expression relation to open terms in the standard way by using related environments and closing substitutions:
\begin{definition}
  For $M,N : \Tm[\Gamma]\sigma$ we define
  \[
  M \GOLogRel\sigma N \defeq \left(\forall \rho,\delta. (\rho \GValLogRel\Gamma \delta) \to \GInterpret{M}\rho \GLogRel\sigma \Eval{(N[\delta])}\right)
  \]
  Here $\delta$ is a closing substitution for $\Gamma$ and $\rho : \GSVal\Gamma$
  an environment, and $\rho \GValLogRel\Gamma \delta$ denotes that for every
  variable $x$ of $\Gamma$ the corresponding semantic and syntactic value in
  $\rho$ and $\delta$ are related.
\end{definition}

One can now show that the logical relation is a congruence and that it is reflexive (the fundamental lemma holds). 
Both these are proved by induction on $C$ and $M$ respectively. 
\begin{lemma}[Congruence Lemma]
  \label{lem:guarded-congruence}
  For any terms $M,N : \Tm[\Gamma]\sigma$ and every context $C: \CtxTy\Delta\tau$ we get that
  \[ 
   M \GOLogRel\sigma N \to C[M]\GOLogRel[\Delta]\tau C[N]
  \]
\end{lemma}

\begin{lemma}[Fundamental Lemma]
  \label{lem:guarded-fundamental}
  For all $M : \Tm[\Gamma]\sigma$ we have that
  $
  M \GOLogRel\sigma M.
  $
\end{lemma}

To prove soundness of the logical relation for contextual equivalence, it only remains to relate it to
termination probability at the unit type. To begin, we use our logical relation to prove Theorem~\ref{thm:eq-pterm-denot-op}
presented in Section~\ref{sec:den:sem}:

\begin{proof}[Proof of Theorem~\ref{thm:eq-pterm-denot-op}]
 Recall that our goal is to show that $\probtermseq{\Eval \,M} \eqlim \probtermseq{\Interpret{M}{}}$ for any $M : \Tm\Unit$.
 One direction follows directly from Lemma~\ref{lem:guarded-fundamental}: Since $\forall \kappa. \GInterpret{M}{} \GLogRel\Unit \Eval{\,M}$, by 
 Lemma~\ref{lem:eq1-probterm} also $\probtermseq{\Interpret{M}{}} \leqlim \probtermseq{\Eval \,M}$. For the other direction, we
 define an alternative denotational semantics $\SInterp{-}$, which agrees with $\Interp{-}$ on types and most terms, but which uses the same steps as
 the operational semantics, so for example, 
\begin{align*}
     \SGInterpret{MN}\rho \defeq  \SGInterpret{M}\rho \GBind \lambda f. \SGInterpret{N}\rho \GBind \lambda v.\GStep(\tabs\tickA\kappa fv) 
\end{align*}
 and similarly for the interpretation of elimination for sum types. Since the steps agree, it is easy to show the following soundness result for
 $\SInterp{-}$: For any well typed closed expression $\cdot \vdash M : \sigma$ we have that
   \[
   \GLiftFun{\SGValInterpret{-}{}}(\GEval M) \equiv \SGInterpret{M}{}
   \]
and as a consequence $\probtermseq{\Eval \,M} \peq \probtermseq{\SInterp{M}{}}$. Finally, we
construct a logical relation on the denotation of types to show that  $\probtermseq{\SInterp{M}{}} \leqlim \probtermseq{\Interp{M}{}}$.
\end{proof}

%
%


As a consequence, the logical relation implies relatedness of the termination probabilities.
\begin{corollary}
  \label{cor:soundness_of_logrel}
  For any $M,N : \Tm\Unit$ we have that 
  \[
  (\forall\kappa. \GInterpret{M}{} \GLogRel\Unit \Eval{(N)}) \to  \probtermseq{\Eval{(M)}} \leqlim \probtermseq{\Eval(N)}
  \] 
\end{corollary}
\begin{proof}
By Lemma~\ref{lem:eq1-probterm}, $\probtermseq{\Interpret{M}{}} \leqlim\probtermseq{\Eval(N)}$, and so the
result follows from Theorem~\ref{thm:eq-pterm-denot-op}. 
%
\end{proof}

Combining this with the Congruence Lemma,
we get that the logical
relation is sound with respect to contextual refinement:

\begin{theorem}
  \label{thm:main}
  For any terms $\Gamma \vdash M : \sigma$ and $\Gamma \vdash N: \sigma$ we have 
  \[
    (\forall \kappa. M \GOLogRel\sigma N) \to M \CtxRef N
    \] 
\end{theorem}

As a corollary, we obtain computational adequacy (using the notation $\Interpret M{} \defeq \Lambda\kappa . \GInterpret M{}$):

\begin{theorem}[Adequacy]
 Let $M, N$ be terms of the same type. If $\Interpret M{} = \Interpret N{}$ then $M \CtxEq N$. 
\end{theorem}

\begin{proof}
 By Lemma~\ref{lem:guarded-fundamental} we have $\forall\kappa . N \GOLogRel\sigma N$, since 
 $\kappa$ is free in the statement of the lemma.  So, since 
 $\Interpret M{} = \Interpret N{}$ also $\forall\kappa . M \GOLogRel\sigma N$. By Theorem~\ref{thm:main}
 then also $M \CtxRef N$. The other direction is symmetric. 
\end{proof}



\section{Examples}
\label{sec:examples}

We now give a series of examples to illustrate how Theorem~\ref{thm:main} can be used for reasoning about $\pFPC$. All the examples presented are variants
of similar examples in the literature~\cite{DBLP:conf/fossacs/BizjakB15,DBLP:conf/esop/0001BBBG018}. Our point is to show how these examples can be done in constructive type theory,
and to illustrate the simplicity of these in our abstract viewpoint. We choose to work directly with the logical relation $\GOLogRel\sigma$ between syntax and semantics. 
An alternative could be to use a relation on denotational semantics, as the one used in the proof of Theorem~\ref{thm:eq-pterm-denot-op}. One could have also
defined a relation directly from syntax to syntax, but the additional steps in the operational semantics would have cluttered the proofs. We return to this point in 
Remark~\ref{rem:steps:clutter:examples}.

We start by showing that the reduction relation $\leadsto$ is contained in contextual equivalence. 

\begin{theorem} \label{thm:leadsto:ctx:eq}
 Let $M$ and $N$ be closed terms of the same type $\sigma$. 
 \begin{enumerate}
\item\label{item:leadsto:ctx:eq} If $\Eval \,M \leadsto \Eval\, N$, then $M \CtxEq N$. 
\item\label{item:limleadsto:ctx:eq} If $\Eval \,M  \limleadsto \Eval\, N$, then $N \CtxRef M$. 
\end{enumerate}
\end{theorem}

\begin{proof}
 For (\ref{item:leadsto:ctx:eq}), by Lemma~\ref{lem:prereqs:leadsto-liftrel}, $\mu \GValLogRel\sigma \Eval \,M$ if and only if $\mu \GValLogRel\sigma \Eval \,N$. By 
 Lemma~\ref{lem:guarded-fundamental} $\forall\kappa .\GInterpret{M}{} \GValLogRel\sigma \Eval \,M$, so $\forall\kappa .\GInterpret{M}{} \GValLogRel\sigma \Eval \,N$, and so
 by Theorem~\ref{thm:main} also $M \CtxRef N$. The other way is similar. The second statement is proved similarly using Lemma~\ref{lem:prereqs:limleadsto-liftrel}.
\end{proof}

For example, since $\Eval{(\YComb{\varphi}\,V)} \leadsto \Eval{(\varphi(\YComb{\varphi})V)}$, also 
$\YComb{\varphi}\,V \CtxEq \varphi(\YComb{\varphi})\,V$. Similarly, since 
\begin{align*}
 \Eval((\Lam{M})\, V) & \peq \CStep (\Eval(M[V/x])) \leadsto \Eval(M[V/x]) \\
 \Eval(\Unfold{(\Fold{V})}) & \peq \CStep(V) \leadsto V
\end{align*}
for closed values $V$ and $\Lam{M}$, 
we get the usual call-by-value $\beta$ rules up to contextual equivalence.

\subsection{A hesitant identity function}
\label{sec:hes:id}
For any type $\sigma$ and rational number $p : \II$, define the hesitant identity function $\IdHes : \sigma \to \sigma$
as 
\begin{align*}
 {\IdHes} & \defeq \Lam[z]{\YComb{}\,\IdHesHelp\,z} & \text{where } \\
 \IdHesHelp & \defeq \Lam[f]{\Lam[x]{\Choice{p}{x}{fx}}}
\end{align*}
Given an $x$, $\IdHes$ makes a probabilistic choice between immediately returning $x$ and calling itself recursively.  
We will show that ${\IdHes} \CtxEq {\sf id}$, starting with
${\sf id} \CtxRef {\IdHes}$. 

Since $\IdHes$ is a value, it suffices to show that 
\begin{equation} \label{eq:hid:goal:op}
  \GDirac{(v)} \GLogRel\sigma \Eval{(\IdHes \,V)}
\end{equation}
for all $v, V$ satisfying $v \GValLogRel\sigma V$. Now, $\Eval{(\IdHes\, V)}  \leadsto \Eval{(\YComb{}\,\IdHesHelp\, V)}$, 
and by definition of $\YComb{}$, $\YComb{}\,\IdHesHelp$ reduces in one step to a value, which we shall simply write $\YIdHesVal{}$ for. To prove 
(\ref{eq:hid:goal:op}) it therefore suffices to prove
\begin{equation} \label{eq:hid:goal:op2}
  \GDirac{(v)} \GLogRel\sigma \Eval{(\YIdHesVal{}\, V)}
\end{equation}
Since the reduction $\Eval{(\YComb{}\,\IdHesHelp\, V)} \leadsto \Eval{(\IdHesHelp \,(\YComb{}\,\IdHesHelp) \,V)}$ factors through
$\YIdHesVal{}\,V$
it follows that
\begin{align*}
 \Eval{(\YIdHesVal{}\, V)} & \leadsto \Eval{(\IdHesHelp \,(\YComb{}\,\IdHesHelp) \,V)} \\
 & \leadsto \Eval{(\IdHesHelp \,\YIdHesVal{} \,V)} \\
 & \leadsto \Eval{(\Choice{p}V{\YIdHesVal{}\, V) }} \\
 & \peq \CChoice p{\CDirac V}{\Eval{(\YIdHesVal{}\, V)}}
\end{align*}
By Example~\ref{ex:lem:convergence} then $\Eval{(\YIdHesVal\, V)} \limleadsto \CDirac V$, so that 
(\ref{eq:hid:goal:op2}) follows from Lemma~\ref{lem:prereqs:limleadsto-liftrel}.  


To prove ${\IdHes} \CtxRef {\sf id}$, since $\IdHes$ is a value, it suffices to show that 
\begin{equation} \label{eq:hid:goal}
  \GValInterpret{\IdHes}{}(v) \GLogRel\sigma \CDirac{(V)}
\end{equation}
for all $v, V$ satisfying $v \GValLogRel\sigma V$. 
Using (\ref{eq:den:Yf:open}) we compute
\begin{align*}
 \GValInterpret{\IdHes}{}(v) & \peq \GInterpret{\YComb{}(\IdHesHelp)\, z}{z\mapsto v} \\
 & \peq \ExpStep\GInterpret{\IdHesHelp(\YComb{}(\IdHesHelp))\, z}{z\mapsto v} \\
 & \peq \ExpStep\GInterpret{\Choice{p}{x}{fx}}{x\mapsto v, f \mapsto \GValInterpret{\YIdHesVal}{}} \\
 & \peq \ExpStep(\GChoice p {\GDirac{v}}{\GValInterpret{\IdHes}{}(v)})
\end{align*}
so that (\ref{eq:hid:goal}) is equivalent to 
\begin{equation}\label{eq:hid:goal:unfolded}
\Later(\GChoice p {\GDirac{v}}{\GValInterpret{\IdHes}{}(v)}\GLogRel\sigma \CDirac{(V)} )
\end{equation}
and can therefore be easily proved by guarded recursion using Lemma~\ref{lem:prereqs:step-choice}. 

\begin{remark} \label{rem:steps:clutter:examples}
 Note the benefit of having the denotational semantics on the left hand side of the relation 
 in the example above: Unfolding (\ref{eq:hid:goal}) only leads to a single (outer) $\Later$ in (\ref{eq:hid:goal:unfolded}).
 Had we used operational semantics on the left hand side, the unfolding would have been cluttered by additional steps 
 intertwined with probabilistic choices.
 On the right hand side, the additional steps disappear by the use of $\leadsto$. 
\end{remark}

\subsection{A fair coin from an unfair coin}
\label{sec:fair:unfair}

Define $\Bool \defeq \CoprodTy{\Unit}{\Unit}$ referring to the elements as $\True$ and $\False$. A \emph{coin} is a program 
of the form $\Choice{p}{\True}{\False}:\Bool$ for $p : \II$. A coin is called \emph{fair} if $p\peq \frac{1}{2}$.

Given a coin, one can encode a fair coin as follows: Toss the coin twice. If the results are different, return the first result, otherwise try again. 
This idea can be encoded as the $\pFPC$ term $\GFair$:
%
\begin{align*}
  &\GFair \defeq \YComb(\GFairHelp) : \FnTy\Unit\Bool \\
  &\begin{aligned}
\GFairHelp \defeq \Lam[g]{}\Lam[z]{} &\mathsf{let}\ x = \Choice{p}{\True}{\False}\\
&\mathsf{let}\ y = \Choice{p}{\True}{\False} \\
&\IfThen{{\sf eqbool}(x,y)}{g(z)}{x}
  \end{aligned}
\end{align*}
We now prove that this indeed gives a fair coin.

\begin{theorem}
  \label{thm:unfair-hes}
  $\GFair\Star$ is contextually equivalent to $\Choice{\frac{1}{2}}{\True}{\False}$
\end{theorem}
\begin{proof}
Let ${\YGFairVal}$ be the value that $\YComb(\GFairHelp)$ unfolds to. 
Unfolding definitions shows that 
\begin{align*}
\Eval{(\YGFairVal\Star)} 
& \leadsto \CChoice{p} {(\CChoice{p}{\Eval(\YGFairVal\Star)}{\CDirac\True})\!}{\!(\CChoice{p}{\CDirac\False}{\Eval(\YGFairVal\Star)})} \\
& \peq \CChoice{2p(1-p)}{(\CChoice{\frac12}{\CDirac\True}{\CDirac\False})}{\Eval{(\YGFairVal\Star)}}
\end{align*}
From this, it follows that $\Eval{(\YGFairVal\Star)} \limleadsto\CChoice{\frac12}{\CDirac(\True)}{\CDirac(\False)}$ as in Example~\ref{ex:lem:convergence}, and since 
$\Eval{(\GFair\Star)} \leadsto \Eval{(\YGFairVal\Star)}$, also $\Choice{\frac12}\True\False\CtxRef\GFair\Star$ follows from Theorem~\ref{thm:leadsto:ctx:eq}. 

On the other hand, to show $\GFair\Star \CtxRef \Choice{\frac12}\True\False$, since $\GFair \leadsto \YGFairVal$, it suffices to show that 
$\GInterpret{\YGFairVal\Star}{}\GLogRel\Bool \Eval{(\Choice{\frac{1}{2}}{\True}{\False})}$. Computing 
\begin{align*}
 \GInterpret{\YGFairVal\Star}{} & =  \ExpStep\GInterpret{\GFairHelp\,(\YComb\,\GFairHelp)\Star}{} \\
  & =  \ExpStep\GInterpret{\GFairHelp\,\YGFairVal\,\Star}{} \\
  & =  \ExpStep\left(\GChoice p{\left(\GChoice p{\GInterpret{\YGFairVal\,\Star}{}}{\GDirac(\True)}\right)}{\left(\GChoice p{\GDirac(\False)}{(\GInterpret{\YGFairVal\,\Star}{})}\right)}\right) \\
  & = \ExpStep\left(\GChoice {2p(1-p)}{\left(\GChoice {\frac12}{\GDirac(\True)}{\GDirac(\False)}\right)}{\GInterpret{\YGFairVal\,\Star}{}}\right)
\end{align*}
So that $\GInterpret{\YGFairVal\Star}{}\GLogRel\Bool \Eval{(\Choice{\frac{1}{2}}{\True}{\False})}$
is equivalent to 
\[
 \Later\left(\GChoice {2p(1-p)}{(\GChoice {\frac12}{\GDirac(\True)}{\GDirac(\False)})}{\GInterpret{\YGFairVal\,\Star}{}} \GLogRel\Bool \Eval{(\Choice{\frac{1}{2}}{\True}{\False})}\right)
\]
and can therefore be proved by guarded recursion.
\end{proof}

\subsection{Random Walk}
\label{sec:random:walk}

The next example, inspired by~\cite{DBLP:conf/esop/0001BBBG018}, is an equivalence of two random walks. 
We represent a random walk by its trace, which is a potentially infinite list of integers. 
In a call-by-value language such terms are best implemented as \emph{lazy lists}. We define
an $\ProbFPC$ type of these as 
\[
\LazyL \defeq \RecTy{1 + \Nat \times (\FnTy\Unit{X})}
\]
as well as a familiar interface for programming with these in Figure~\ref{ex:lazyl}. 
\begin{figure}[ht]
  \begin{align*}
    &\begin{aligned}
    &\Nil : \LazyL \\
    &\Nil \defeq \Fold{(\Inl(\Star))} \\ \\
    &\Head : \LazyL \to \Nat + 1 \\
    &\Head \defeq \Lam[l]{\Case{\Unfold{l}}{x.\Inr{\Star}}{(n,f).\Inl{\Numeral{n}}}}
    \end{aligned}
    \qquad \qquad
    \begin{aligned}
    &\Tail : \LazyL \to \LazyL \\
    &\Tail \defeq \Lam[l] {\Case{\Unfold{l}}{x.\Nil}{(n,f).f(\Star)}} \\ \\
    &\Cons{-}{-} : \Nat \to (\FnTy\Unit\LazyL) \to \LazyL \\
    &\Cons{-}{-} \defeq \Lam[n]{\Lam[f]{\Fold{(\Inr(n,f))}}}
    \end{aligned}
  \end{align*}
    \caption{Lazy List Interface}
  \label{ex:lazyl}
\end{figure}

Using these we can define a lazy list processing function $\EverySnd : \LazyL \to \LazyL$, returning a lazy list that contains every second element of the argument. 
Define $\EverySnd \defeq \YComb(G_{\EverySnd})$ where
\begin{align*}
  G_{\EverySnd} &: (\FnTy\LazyL\LazyL) \to \LazyL \to \LazyL \\
  G_{\EverySnd} & \defeq \Lam[g]{\Lam[l]{\Case{\Unfold{(l)}}{x.\Nil}{(n,f).\Cons{n}{\Lam[y]g(\Tail(f(\Star)))}}}}
\end{align*}
 For this example, we consider two random walks: a classical random walk $\RandW : \Nat \to \LazyL$, which in each step takes a step up or down, each with probability $\frac12$, 
 and a sped-up version $\RandWTwo : \Nat \to \LazyL$, which stays put with probability $\frac12$, and otherwise takes two steps either up or down.
 Define these as 
\begin{align*}
 \RandW & \defeq \YComb{(G_{\RandW})} & 
 \RandWTwo & \defeq \YComb{(G_{\RandWTwo})}  
\end{align*}
where 
    \begin{align*}
      G_{\RandW} & \defeq \Lam[g]{\Lam[n]{\Cons{n}{\Lam[y]{\Ifz{n}{\Nil}{\Choice{\frac{1}{2}}{g(n-1)}{g(n+1)}}}}}} \\
      G_{\RandWTwo} & \defeq \Lam[g]{\Lam[n]{\Cons{n}{\Lam[y]{\Ifz{n}{\Nil}{\Choice{\frac{1}{2}}{g(n)}{\Choice{\frac{1}{2}}{g(n-2)}{g(n+2)}}}}}}}
    \end{align*}
Intuitively, $\RandWTwo$ is just $\RandW$ at double speed. The following theorem makes this intuition precise.
    \begin{theorem}
      \label{ex:thm:randw}
      For all $n: \NN$ we have that $\EverySnd(\RandW(\Numeral{2n})) \CtxEq \RandWTwo(\Numeral{2n})$.
    \end{theorem}

\begin{proof}[Proof (Sketch)]
 Let $V_{\RandW}$, $V_{\RandWTwo}$ and $\VEverySnd$ be the values that 
 $\RandW$, $\RandWTwo$ and $\EverySnd$ reduce to. It suffices to show that $\VEverySnd(\VRandW(\Numeral{2n})) \CtxEq \VRandWTwo(\Numeral{2n})$. 
 We just show $\GInterpret{\VEverySnd(\VRandW(\Numeral{2n}))}{} \GLogRel\LazyL \Eval(\VRandWTwo(\Numeral{2n}))$ by guarded recursion, omitting the
 other direction, which can be proved similarly. We focus just on the case of  $n>0$, which is the harder one. Unfolding definitions shows that 
\begin{align*}
 \GInterpret{\VEverySnd(V_{\RandW}(\Numeral{2n}))}{} 
 & \peq (\ExpStep)^3( \GDirac(\nextop(\SemInr{\SemPair{\Numeral{2n}}{\lambda\_ . W_{\RandW}}})))
\end{align*}
where 
\[
W_{\RandW} =  
  \left( \GChoice{\frac12}{\GInterpret{\VEverySnd(\Tail(\VRandW(\Numeral{2n\!-\!1})))}{}}{\GInterpret{\VEverySnd(\Tail(\VRandW(\Numeral{2n\!+\!1})))}{}}  \right) 
\]
Similarly, 
\begin{align*}
 \Eval(\VRandWTwo(\Numeral{2n})) 
 & \leadsto \Fold{(\Inr(\Numeral{2n}, \Lam[y]{\Ifz{\Numeral{2n}}{\Nil}{\Choice{\frac{1}{2}}{\VRandWTwo(\Numeral{2n})}{\Choice{\frac{1}{2}}{\VRandWTwo(\Numeral{2n\!-\!2})}{\VRandWTwo(\Numeral{2n\!+\!2})}}}}))}
\end{align*}
So showing $\GInterpret{\VEverySnd(\VRandW(\Numeral{2n}))}{} \GLogRel\LazyL \Eval(\VRandWTwo(\Numeral{2n}))$ easily reduces to showing 
\begin{equation} \label{eq:randomwalk:1}
 \Later(W_{\RandW} \GLogRel\LazyL \Eval(\Choice{\frac{1}{2}}{\VRandWTwo(\Numeral{2n})}{\Choice{\frac{1}{2}}{\VRandWTwo(\Numeral{2n\!-\!2})}{\VRandWTwo(\Numeral{2n\!+\!2})}} ))
\end{equation}
Since $2n-1>0$, unfolding definitions gives 
\begin{align*}
 \GInterpret{\VEverySnd(\Tail(\VRandW(\Numeral{2n\!-\!1})))}{} 
 & \peq (\ExpStep)^2\left(\GChoice{\frac12}{\GInterpret{\VEverySnd(\VRandW(\Numeral{2n\!-\!2}))}{}}{\GInterpret{\VEverySnd(\VRandW(\Numeral{2n}))}{}}\right)   \\
 \GInterpret{\VEverySnd(\Tail(\VRandW(\Numeral{2n\!+\!1})))}{} 
 & \peq (\ExpStep)^2\left(\GChoice{\frac12}{\GInterpret{\VEverySnd(\VRandW(\Numeral{2n}))}{}}{\GInterpret{\VEverySnd(\VRandW(\Numeral{2n\!-\!2}))}{}}  \right)
\end{align*}
so by Lemma~\ref{lem:prereqs:step-choice} (\ref{eq:randomwalk:1}) is equivalent to 
\[
  (\Later)^3
  \left( 
  \begin{aligned}
   &\GChoice{\frac12}{\GInterpret{\VEverySnd(\VRandW(\Numeral{2n}))}{}}{\left(\GChoice{\frac12}{\GInterpret{\VEverySnd(\VRandW(\Numeral{2n\!-\!2}))}{}}{\GInterpret{\VEverySnd(\VRandW(\Numeral{2n\!+\!2}))}{}}\right)} \\
  & \GLogRel\LazyL \CChoice{\frac{1}{2}}{\Eval(\VRandWTwo(\Numeral{2n}))}{\left(\CChoice{\frac{1}{2}}{\Eval(\VRandWTwo(\Numeral{2n\!-\!2}))}{\Eval(\VRandWTwo(\Numeral{2n\!+\!2}))}\right)} 
  \end{aligned}
  \right)
\]
which follows from the guarded induction hypothesis and Lemma~\ref{lem:prereqs:choice-lemma}. 
\end{proof}

\begin{remark}
	Note again how working with denotational semantics simplifies this proof. This enables us to take two steps in the walk defined by $V_{\RandW}$,
	collect all outcomes, and couple them with the outcomes of a single step of $V_{\RandWTwo}$. Other coupling-based logics like~\cite{DBLP:conf/esop/0001BBBG018}
	require specialized rules to deal with these two-to-one couplings but in our setting it is just a simple consequence of our definitions.
\end{remark}



\section{Related Work}
\label{sec:related-work}

We have discussed some related work in the Introduction and throughout the paper; here we discuss
additional related work.

In this paper, we have shown how to use synthetic guarded domain theory (SGDT) to model $\ProbFPC$.
SGDT has been used in earlier work to model PCF \cite{DBLP:journals/entcs/PaviottiMB15}, a call-by-name variant of FPC \cite{DBLP:journals/mscs/MogelbergP19}, FPC with general references \cite{DBLP:journals/corr/abs-2210-02169,DBLP:journals/corr/abs-2307-16608}, untyped lambda calculus with nondeterminism \cite{DBLP:journals/corr/abs-2112-14056}, and guarded interaction trees \cite{gitrees}.
Thus the key new challenge addressed in this paper is the modelling of probabilistic choice, in combination with recursion, which led us to introduce (guarded) convex delay algebras.

In previous works on SDGT models of PCF and FPC, the logical relation between syntax and semantics required a 1-to-1 correspondence between the steps on either side. 
The previous work on nondeterminism \cite{DBLP:journals/corr/abs-2112-14056} is in some ways closer to our work: To model the combination of non-determinism and recursion in the
case of may-simulation, it uses a monad defined similarly to $\GLiftDist$, but for the finite powerset monad, rather than distributions. It also uses a logical relation between syntax and
semantics, which similarly to ours is a refinement relation. The theory developed for $\GLiftDist$ here, such as the $\leadsto$ relation, and the lifting of relations using couplings is
new. The applications are also different:
%
%
Here we use the relation to reason about contextual refinement whereas in \emph{op. cit.} it is used for proving congruence of an applicative simulation relation defined on operational semantics.

The use of step-indexed logical relations to model the combination of
probabilistic choice, recursive types, polymorphism and other expressive
language features has also been explored in~\cite{DBLP:conf/fossacs/BizjakB15,
DBLP:journals/pacmpl/AguirreB23, gregersenAHTB24}. These papers use an explicit
account of step-indexing, and rely on non-constructive mathematics to define
their operational semantics and to prove their soundness theorems. Our approach
is constructive, and our use of SGDT eliminates the need for explicit
manipulation of step indices. Moreover, their logical relations are defined purely
in operational terms.

Applicative bisimulation for a probabilistic, call-by-value version of PCF is studied in~\cite{DBLP:conf/esop/CrubilleL14}. This work defines an approximation-based operational semantics, where a term evaluates to a family of finite distributions over values, each element of the
family corresponding to the values observed after a given finite number of steps. This is reminiscent of our semantics, although in \emph{op. cit.} the approximation semantics is then used to define a limit semantics using suprema.

Probabilistic couplings and liftings have been a popular technique in recent years to reason about probabilistic programs~\cite{DBLP:conf/popl/BartheGB09, DBLP:conf/lpar/BartheEGHSS15}, since they provide a compositional way to lift relations from
base types to distributions over those types. To the best of our knowledge, our presentation is the first that uses constructive mathematics in the definition of
couplings and the proofs of their properties, e.g. the bind lemma. Our definition of couplings is also asymmetric to account for the fact that the distribution on the left uses guarded recursion. This is similar in spirit to the left-partial coupling definition of \cite{gregersenAHTB24}, where asymmetry is used to account for step indexing.

In this paper, we model the probabilistic effects of $\ProbFPC$ using guarded / co-inductive types. In spirit, this is similar to how effects are being modeled in both coinductive and guarded interaction trees \cite{DBLP:journals/pacmpl/XiaZHHMPZ20,gitrees}. Coinductive interaction trees are useful for giving denotational semantics for first-order programs and are defined using ordinary type theory, whereas guarded interaction trees can also be applied to give denotational semantics for higher-order programs, but are defined using a fragment of guarded type theory.
Interaction trees have recently been extended to account for nondeterminism \cite{DBLP:journals/pacmpl/ChappeHHZZ23} but, to the best of our knowledge, interaction trees have so far not been extended to account for probabilistic effects.





\section{Conclusion and Future Work}
\label{sec:conclusion}

We have developed a notion of (guarded) convex delay algebras and shown how to use
it to define and relate operational and denotational semantics for $\ProbFPC$ in
guarded type theory. To the best of our knowledge, this is the first
constructive type theoretic account of the semantics of $\ProbFPC$. The denotational
semantics can be viewed as a shallow embedding of $\ProbFPC$ in constructive type
theory, which can therefore be used directly as a probabilistic programming language. 
Our examples show how to use the relation between syntax and semantics for proving
contextual equivalence of $\ProbFPC$ programs. The use of denotational semantics
for these simplifies proofs by using much fewer steps than the operational semantics.

Future work includes combining and extending the present work with the account of nondeterminism
in \cite{DBLP:journals/corr/abs-2112-14056} and to compare the resulting model
with the recent classically defined operationally-based logical relation in \cite{DBLP:journals/pacmpl/AguirreB23}.
%


\begin{acks}
  This work was supported in part by the Independent Research Fund Denmark grant number
  2032-00134B, in part by a Villum Investigator grant
  (no. 25804), Center for Basic Research in Program Verification (CPV), from the
  VILLUM Foundation, and in part by the European Union (ERC, CHORDS, 101096090).
  Views and opinions expressed are however those of the author(s) only and do
  not necessarily reflect those of the European Union or the European Research
  Council. Neither the European Union nor the granting authority can be held
  responsible for them.
\end{acks}

\bibliographystyle{ACM-Reference-Format}
\bibliography{bib.bib}


\begin{thebibliography}{44}


\ifx \showCODEN    \undefined \def \showCODEN     #1{\unskip}     \fi
\ifx \showDOI      \undefined \def \showDOI       #1{#1}\fi
\ifx \showISBNx    \undefined \def \showISBNx     #1{\unskip}     \fi
\ifx \showISBNxiii \undefined \def \showISBNxiii  #1{\unskip}     \fi
\ifx \showISSN     \undefined \def \showISSN      #1{\unskip}     \fi
\ifx \showLCCN     \undefined \def \showLCCN      #1{\unskip}     \fi
\ifx \shownote     \undefined \def \shownote      #1{#1}          \fi
\ifx \showarticletitle \undefined \def \showarticletitle #1{#1}   \fi
\ifx \showURL      \undefined \def \showURL       {\relax}        \fi
\providecommand\bibfield[2]{#2}
\providecommand\bibinfo[2]{#2}
\providecommand\natexlab[1]{#1}
\providecommand\showeprint[2][]{arXiv:#2}

\bibitem[Aguirre et~al\mbox{.}(2018)]%
        {DBLP:conf/esop/0001BBBG018}
\bibfield{author}{\bibinfo{person}{Alejandro Aguirre}, \bibinfo{person}{Gilles
  Barthe}, \bibinfo{person}{Lars Birkedal}, \bibinfo{person}{Ales Bizjak},
  \bibinfo{person}{Marco Gaboardi}, {and} \bibinfo{person}{Deepak Garg}.}
  \bibinfo{year}{2018}\natexlab{}.
\newblock \showarticletitle{Relational Reasoning for Markov Chains in a
  Probabilistic Guarded Lambda Calculus}. In
  \bibinfo{booktitle}{\emph{Programming Languages and Systems - 27th European
  Symposium on Programming, {ESOP} 2018, Held as Part of the European Joint
  Conferences on Theory and Practice of Software, {ETAPS} 2018, Thessaloniki,
  Greece, April 14-20, 2018, Proceedings}} \emph{(\bibinfo{series}{Lecture
  Notes in Computer Science}, Vol.~\bibinfo{volume}{10801})},
  \bibfield{editor}{\bibinfo{person}{Amal Ahmed}} (Ed.).
  \bibinfo{publisher}{Springer}, \bibinfo{pages}{214--241}.
\newblock
\urldef\tempurl%
\url{https://doi.org/10.1007/978-3-319-89884-1\_8}
\showDOI{\tempurl}


\bibitem[Aguirre and Birkedal(2023)]%
        {DBLP:journals/pacmpl/AguirreB23}
\bibfield{author}{\bibinfo{person}{Alejandro Aguirre} {and}
  \bibinfo{person}{Lars Birkedal}.} \bibinfo{year}{2023}\natexlab{}.
\newblock \showarticletitle{Step-Indexed Logical Relations for Countable
  Nondeterminism and Probabilistic Choice}.
\newblock \bibinfo{journal}{\emph{Proc. {ACM} Program. Lang.}}
  \bibinfo{volume}{7}, \bibinfo{number}{{POPL}} (\bibinfo{year}{2023}),
  \bibinfo{pages}{33--60}.
\newblock
\urldef\tempurl%
\url{https://doi.org/10.1145/3571195}
\showDOI{\tempurl}


\bibitem[Atkey and McBride(2013)]%
        {DBLP:conf/icfp/AtkeyM13}
\bibfield{author}{\bibinfo{person}{Robert Atkey} {and} \bibinfo{person}{Conor
  McBride}.} \bibinfo{year}{2013}\natexlab{}.
\newblock \showarticletitle{Productive coprogramming with guarded recursion}.
  In \bibinfo{booktitle}{\emph{{ACM} {SIGPLAN} International Conference on
  Functional Programming, ICFP'13, Boston, MA, {USA} - September 25 - 27,
  2013}}, \bibfield{editor}{\bibinfo{person}{Greg Morrisett} {and}
  \bibinfo{person}{Tarmo Uustalu}} (Eds.). \bibinfo{publisher}{{ACM}},
  \bibinfo{pages}{197--208}.
\newblock
\urldef\tempurl%
\url{https://doi.org/10.1145/2500365.2500597}
\showDOI{\tempurl}


\bibitem[Bahr. et~al\mbox{.}(2017)]%
        {bahr2017clocks}
\bibfield{author}{\bibinfo{person}{Patrick Bahr.}, \bibinfo{person}{Hans~Bugge
  Grathwohl}, {and} \bibinfo{person}{Rasmus~Ejlers M{\o}gelberg}.}
  \bibinfo{year}{2017}\natexlab{}.
\newblock \showarticletitle{The clocks are ticking: No more delays!}. In
  \bibinfo{booktitle}{\emph{2017 32nd Annual ACM/IEEE Symposium on Logic in
  Computer Science (LICS)}}. IEEE, \bibinfo{pages}{1--12}.
\newblock


\bibitem[Barthe et~al\mbox{.}(2015)]%
        {DBLP:conf/lpar/BartheEGHSS15}
\bibfield{author}{\bibinfo{person}{Gilles Barthe}, \bibinfo{person}{Thomas
  Espitau}, \bibinfo{person}{Benjamin Gr{\'{e}}goire}, \bibinfo{person}{Justin
  Hsu}, \bibinfo{person}{L{\'{e}}o Stefanesco}, {and}
  \bibinfo{person}{Pierre{-}Yves Strub}.} \bibinfo{year}{2015}\natexlab{}.
\newblock \showarticletitle{Relational Reasoning via Probabilistic Coupling}.
  In \bibinfo{booktitle}{\emph{Logic for Programming, Artificial Intelligence,
  and Reasoning - 20th International Conference, {LPAR-20} 2015, Suva, Fiji,
  November 24-28, 2015, Proceedings}} \emph{(\bibinfo{series}{Lecture Notes in
  Computer Science}, Vol.~\bibinfo{volume}{9450})},
  \bibfield{editor}{\bibinfo{person}{Martin Davis}, \bibinfo{person}{Ansgar
  Fehnker}, \bibinfo{person}{Annabelle McIver}, {and} \bibinfo{person}{Andrei
  Voronkov}} (Eds.). \bibinfo{publisher}{Springer}, \bibinfo{pages}{387--401}.
\newblock
\urldef\tempurl%
\url{https://doi.org/10.1007/978-3-662-48899-7\_27}
\showDOI{\tempurl}


\bibitem[Barthe et~al\mbox{.}(2009)]%
        {DBLP:conf/popl/BartheGB09}
\bibfield{author}{\bibinfo{person}{Gilles Barthe}, \bibinfo{person}{Benjamin
  Gr{\'{e}}goire}, {and} \bibinfo{person}{Santiago~Zanella B{\'{e}}guelin}.}
  \bibinfo{year}{2009}\natexlab{}.
\newblock \showarticletitle{Formal certification of code-based cryptographic
  proofs}. In \bibinfo{booktitle}{\emph{Proceedings of the 36th {ACM}
  {SIGPLAN-SIGACT} Symposium on Principles of Programming Languages, {POPL}
  2009, Savannah, GA, USA, January 21-23, 2009}},
  \bibfield{editor}{\bibinfo{person}{Zhong Shao} {and}
  \bibinfo{person}{Benjamin~C. Pierce}} (Eds.). \bibinfo{publisher}{{ACM}},
  \bibinfo{pages}{90--101}.
\newblock
\urldef\tempurl%
\url{https://doi.org/10.1145/1480881.1480894}
\showDOI{\tempurl}


\bibitem[Birkedal and M{\o}gelberg(2013)]%
        {DBLP:conf/lics/BirkedalM13}
\bibfield{author}{\bibinfo{person}{Lars Birkedal} {and}
  \bibinfo{person}{Rasmus~Ejlers M{\o}gelberg}.}
  \bibinfo{year}{2013}\natexlab{}.
\newblock \showarticletitle{Intensional Type Theory with Guarded Recursive
  Types qua Fixed Points on Universes}. In \bibinfo{booktitle}{\emph{28th
  Annual {ACM/IEEE} Symposium on Logic in Computer Science, {LICS} 2013, New
  Orleans, LA, USA, June 25-28, 2013}}. \bibinfo{publisher}{{IEEE} Computer
  Society}, \bibinfo{pages}{213--222}.
\newblock
\urldef\tempurl%
\url{https://doi.org/10.1109/LICS.2013.27}
\showDOI{\tempurl}


\bibitem[Bizjak and Birkedal(2015)]%
        {DBLP:conf/fossacs/BizjakB15}
\bibfield{author}{\bibinfo{person}{Ales Bizjak} {and} \bibinfo{person}{Lars
  Birkedal}.} \bibinfo{year}{2015}\natexlab{}.
\newblock \showarticletitle{Step-Indexed Logical Relations for Probability}. In
  \bibinfo{booktitle}{\emph{Foundations of Software Science and Computation
  Structures - 18th International Conference, FoSSaCS 2015, Held as Part of the
  European Joint Conferences on Theory and Practice of Software, {ETAPS} 2015,
  London, UK, April 11-18, 2015. Proceedings}} \emph{(\bibinfo{series}{Lecture
  Notes in Computer Science}, Vol.~\bibinfo{volume}{9034})},
  \bibfield{editor}{\bibinfo{person}{Andrew~M. Pitts}} (Ed.).
  \bibinfo{publisher}{Springer}, \bibinfo{pages}{279--294}.
\newblock
\urldef\tempurl%
\url{https://doi.org/10.1007/978-3-662-46678-0\_18}
\showDOI{\tempurl}


\bibitem[Chappe et~al\mbox{.}(2023)]%
        {DBLP:journals/pacmpl/ChappeHHZZ23}
\bibfield{author}{\bibinfo{person}{Nicolas Chappe}, \bibinfo{person}{Paul He},
  \bibinfo{person}{Ludovic Henrio}, \bibinfo{person}{Yannick Zakowski}, {and}
  \bibinfo{person}{Steve Zdancewic}.} \bibinfo{year}{2023}\natexlab{}.
\newblock \showarticletitle{Choice Trees: Representing Nondeterministic,
  Recursive, and Impure Programs in Coq}.
\newblock \bibinfo{journal}{\emph{Proc. {ACM} Program. Lang.}}
  \bibinfo{volume}{7}, \bibinfo{number}{{POPL}} (\bibinfo{year}{2023}),
  \bibinfo{pages}{1770--1800}.
\newblock
\urldef\tempurl%
\url{https://doi.org/10.1145/3571254}
\showDOI{\tempurl}


\bibitem[Clouston(2018)]%
        {clouston2018fitch}
\bibfield{author}{\bibinfo{person}{Ranald Clouston}.}
  \bibinfo{year}{2018}\natexlab{}.
\newblock \showarticletitle{Fitch-style modal lambda calculi}. In
  \bibinfo{booktitle}{\emph{International Conference on Foundations of Software
  Science and Computation Structures}}. Springer, \bibinfo{pages}{258--275}.
\newblock


\bibitem[Clouston et~al\mbox{.}(2020)]%
        {drat}
\bibfield{author}{\bibinfo{person}{Ranald Clouston}, \bibinfo{person}{Bassel
  Mannaa}, \bibinfo{person}{Rasmus~Ejlers M{\o}gelberg},
  \bibinfo{person}{Andrew~M. Pitts}, {and} \bibinfo{person}{Bas Spitters}.}
  \bibinfo{year}{2020}\natexlab{}.
\newblock \showarticletitle{Modal Dependent Type Theory and Dependent Right
  Adjoints}.
\newblock \bibinfo{journal}{\emph{Mathematical Structures in Computer Science}}
  \bibinfo{volume}{30}, \bibinfo{number}{2} (\bibinfo{year}{2020}),
  \bibinfo{pages}{118--138}.
\newblock


\bibitem[Cohen et~al\mbox{.}(2017)]%
        {CTT}
\bibfield{author}{\bibinfo{person}{Cyril Cohen}, \bibinfo{person}{Thierry
  Coquand}, \bibinfo{person}{Simon Huber}, {and} \bibinfo{person}{Anders
  M{\"{o}}rtberg}.} \bibinfo{year}{2017}\natexlab{}.
\newblock \showarticletitle{Cubical Type Theory: {A} Constructive
  Interpretation of the Univalence Axiom}.
\newblock \bibinfo{journal}{\emph{{FLAP}}} \bibinfo{volume}{4},
  \bibinfo{number}{10} (\bibinfo{year}{2017}), \bibinfo{pages}{3127--3170}.
\newblock
\urldef\tempurl%
\url{http://collegepublications.co.uk/ifcolog/?00019}
\showURL{%
\tempurl}


\bibitem[Crubill{\'{e}} and Lago(2014)]%
        {DBLP:conf/esop/CrubilleL14}
\bibfield{author}{\bibinfo{person}{Rapha{\"{e}}lle Crubill{\'{e}}} {and}
  \bibinfo{person}{Ugo~Dal Lago}.} \bibinfo{year}{2014}\natexlab{}.
\newblock \showarticletitle{On Probabilistic Applicative Bisimulation and
  Call-by-Value {\(\lambda\)}-Calculi}. In
  \bibinfo{booktitle}{\emph{Programming Languages and Systems - 23rd European
  Symposium on Programming, {ESOP} 2014, Held as Part of the European Joint
  Conferences on Theory and Practice of Software, {ETAPS} 2014, Grenoble,
  France, April 5-13, 2014, Proceedings}} \emph{(\bibinfo{series}{Lecture Notes
  in Computer Science}, Vol.~\bibinfo{volume}{8410})},
  \bibfield{editor}{\bibinfo{person}{Zhong Shao}} (Ed.).
  \bibinfo{publisher}{Springer}, \bibinfo{pages}{209--228}.
\newblock
\urldef\tempurl%
\url{https://doi.org/10.1007/978-3-642-54833-8\_12}
\showDOI{\tempurl}


\bibitem[Culpepper and Cobb(2017)]%
        {DBLP:conf/esop/CulpepperC17}
\bibfield{author}{\bibinfo{person}{Ryan Culpepper} {and}
  \bibinfo{person}{Andrew Cobb}.} \bibinfo{year}{2017}\natexlab{}.
\newblock \showarticletitle{Contextual Equivalence for Probabilistic Programs
  with Continuous Random Variables and Scoring}. In
  \bibinfo{booktitle}{\emph{Programming Languages and Systems - 26th European
  Symposium on Programming, {ESOP} 2017, Held as Part of the European Joint
  Conferences on Theory and Practice of Software, {ETAPS} 2017, Uppsala,
  Sweden, April 22-29, 2017, Proceedings}} \emph{(\bibinfo{series}{Lecture
  Notes in Computer Science}, Vol.~\bibinfo{volume}{10201})},
  \bibfield{editor}{\bibinfo{person}{Hongseok Yang}} (Ed.).
  \bibinfo{publisher}{Springer}, \bibinfo{pages}{368--392}.
\newblock
\urldef\tempurl%
\url{https://doi.org/10.1007/978-3-662-54434-1\_14}
\showDOI{\tempurl}


\bibitem[Dahlqvist and Kozen(2020)]%
        {DBLP:journals/pacmpl/DahlqvistK20}
\bibfield{author}{\bibinfo{person}{Fredrik Dahlqvist} {and}
  \bibinfo{person}{Dexter Kozen}.} \bibinfo{year}{2020}\natexlab{}.
\newblock \showarticletitle{Semantics of higher-order probabilistic programs
  with conditioning}.
\newblock \bibinfo{journal}{\emph{Proc. {ACM} Program. Lang.}}
  \bibinfo{volume}{4}, \bibinfo{number}{{POPL}} (\bibinfo{year}{2020}),
  \bibinfo{pages}{57:1--57:29}.
\newblock
\urldef\tempurl%
\url{https://doi.org/10.1145/3371125}
\showDOI{\tempurl}


\bibitem[Ehrhard et~al\mbox{.}(2018)]%
        {DBLP:journals/jacm/EhrhardPT18}
\bibfield{author}{\bibinfo{person}{Thomas Ehrhard}, \bibinfo{person}{Michele
  Pagani}, {and} \bibinfo{person}{Christine Tasson}.}
  \bibinfo{year}{2018}\natexlab{}.
\newblock \showarticletitle{Full Abstraction for Probabilistic {PCF}}.
\newblock \bibinfo{journal}{\emph{J. {ACM}}} \bibinfo{volume}{65},
  \bibinfo{number}{4} (\bibinfo{year}{2018}), \bibinfo{pages}{23:1--23:44}.
\newblock
\urldef\tempurl%
\url{https://doi.org/10.1145/3164540}
\showDOI{\tempurl}


\bibitem[Fiore(1994)]%
        {DBLP:phd/ethos/Fiore94}
\bibfield{author}{\bibinfo{person}{Marcelo~P. Fiore}.}
  \bibinfo{year}{1994}\natexlab{}.
\newblock \emph{\bibinfo{title}{Axiomatic domain theory in categories of
  partial maps}}.
\newblock \bibinfo{thesistype}{Ph.\,D. Dissertation}.
  \bibinfo{school}{University of Edinburgh, {UK}}.
\newblock
\urldef\tempurl%
\url{https://hdl.handle.net/1842/406}
\showURL{%
\tempurl}


\bibitem[Frumin et~al\mbox{.}(2024)]%
        {gitrees}
\bibfield{author}{\bibinfo{person}{Dan Frumin}, \bibinfo{person}{Amin Timany},
  {and} \bibinfo{person}{Lars Birkedal}.} \bibinfo{year}{2024}\natexlab{}.
\newblock \showarticletitle{Modular Denotational Semantics for Effects with
  Guarded Interaction Trees}.
\newblock \bibinfo{journal}{\emph{Proc. {ACM} Program. Lang.}}
  \bibinfo{number}{{POPL}} (\bibinfo{year}{2024}).
\newblock


\bibitem[Gregersen et~al\mbox{.}(2024)]%
        {gregersenAHTB24}
\bibfield{author}{\bibinfo{person}{Simon~Oddershede Gregersen},
  \bibinfo{person}{Alejandro Aguirre}, \bibinfo{person}{Philipp~G.
  Haselwarter}, \bibinfo{person}{Joseph Tassarotti}, {and}
  \bibinfo{person}{Lars Birkedal}.} \bibinfo{year}{2024}\natexlab{}.
\newblock \showarticletitle{Asynchronous Probabilistic Couplings in
  Higher-Order Separation Logic}.
\newblock \bibinfo{journal}{\emph{Proc. ACM Program. Lang.}}
  \bibinfo{volume}{8}, \bibinfo{number}{POPL}, Article \bibinfo{articleno}{26}
  (\bibinfo{date}{jan} \bibinfo{year}{2024}), \bibinfo{numpages}{32}~pages.
\newblock
\urldef\tempurl%
\url{https://doi.org/10.1145/3632868}
\showDOI{\tempurl}


\bibitem[Heunen et~al\mbox{.}(2017)]%
        {DBLP:conf/lics/HeunenKSY17}
\bibfield{author}{\bibinfo{person}{Chris Heunen}, \bibinfo{person}{Ohad
  Kammar}, \bibinfo{person}{Sam Staton}, {and} \bibinfo{person}{Hongseok
  Yang}.} \bibinfo{year}{2017}\natexlab{}.
\newblock \showarticletitle{A convenient category for higher-order probability
  theory}. In \bibinfo{booktitle}{\emph{32nd Annual {ACM/IEEE} Symposium on
  Logic in Computer Science, {LICS} 2017, Reykjavik, Iceland, June 20-23,
  2017}}. \bibinfo{publisher}{{IEEE} Computer Society}, \bibinfo{pages}{1--12}.
\newblock
\urldef\tempurl%
\url{https://doi.org/10.1109/LICS.2017.8005137}
\showDOI{\tempurl}


\bibitem[Jacobs(2010)]%
        {DBLP:conf/ifipTCS/Jacobs10}
\bibfield{author}{\bibinfo{person}{Bart Jacobs}.}
  \bibinfo{year}{2010}\natexlab{}.
\newblock \showarticletitle{Convexity, Duality and Effects}. In
  \bibinfo{booktitle}{\emph{Theoretical Computer Science - 6th {IFIP} {TC} 1/WG
  2.2 International Conference, {TCS} 2010, Held as Part of {WCC} 2010,
  Brisbane, Australia, September 20-23, 2010. Proceedings}}
  \emph{(\bibinfo{series}{{IFIP} Advances in Information and Communication
  Technology}, Vol.~\bibinfo{volume}{323})},
  \bibfield{editor}{\bibinfo{person}{Cristian~S. Calude} {and}
  \bibinfo{person}{Vladimiro Sassone}} (Eds.). \bibinfo{publisher}{Springer},
  \bibinfo{pages}{1--19}.
\newblock
\urldef\tempurl%
\url{https://doi.org/10.1007/978-3-642-15240-5\_1}
\showDOI{\tempurl}


\bibitem[Johann et~al\mbox{.}(2010)]%
        {DBLP:conf/lics/JohannSV10}
\bibfield{author}{\bibinfo{person}{Patricia Johann}, \bibinfo{person}{Alex
  Simpson}, {and} \bibinfo{person}{Janis Voigtl{\"{a}}nder}.}
  \bibinfo{year}{2010}\natexlab{}.
\newblock \showarticletitle{A Generic Operational Metatheory for Algebraic
  Effects}. In \bibinfo{booktitle}{\emph{Proceedings of the 25th Annual {IEEE}
  Symposium on Logic in Computer Science, {LICS} 2010, 11-14 July 2010,
  Edinburgh, United Kingdom}}. \bibinfo{publisher}{{IEEE} Computer Society},
  \bibinfo{pages}{209--218}.
\newblock
\urldef\tempurl%
\url{https://doi.org/10.1109/LICS.2010.29}
\showDOI{\tempurl}


\bibitem[Jones and Plotkin(1989)]%
        {DBLP:conf/lics/JonesP89}
\bibfield{author}{\bibinfo{person}{C. Jones} {and} \bibinfo{person}{Gordon~D.
  Plotkin}.} \bibinfo{year}{1989}\natexlab{}.
\newblock \showarticletitle{A Probabilistic Powerdomain of Evaluations}. In
  \bibinfo{booktitle}{\emph{Proceedings of the Fourth Annual Symposium on Logic
  in Computer Science {(LICS} '89), Pacific Grove, California, USA, June 5-8,
  1989}}. \bibinfo{publisher}{{IEEE} Computer Society},
  \bibinfo{pages}{186--195}.
\newblock
\urldef\tempurl%
\url{https://doi.org/10.1109/LICS.1989.39173}
\showDOI{\tempurl}


\bibitem[Katsumata and Sato(2013)]%
        {katsumata2013preorders}
\bibfield{author}{\bibinfo{person}{Shin-ya Katsumata} {and}
  \bibinfo{person}{Tetsuya Sato}.} \bibinfo{year}{2013}\natexlab{}.
\newblock \showarticletitle{Preorders on {Monads} and {Coalgebraic}
  {Simulations}}.
\newblock In \bibinfo{booktitle}{\emph{Foundations of {Software} {Science} and
  {Computation} {Structures}}}, \bibfield{editor}{\bibinfo{person}{David
  Hutchison}, \bibinfo{person}{Takeo Kanade}, \bibinfo{person}{Josef Kittler},
  \bibinfo{person}{Jon~M. Kleinberg}, \bibinfo{person}{Friedemann Mattern},
  \bibinfo{person}{John~C. Mitchell}, \bibinfo{person}{Moni Naor},
  \bibinfo{person}{Oscar Nierstrasz}, \bibinfo{person}{C.~Pandu~Rangan},
  \bibinfo{person}{Bernhard Steffen}, \bibinfo{person}{Madhu Sudan},
  \bibinfo{person}{Demetri Terzopoulos}, \bibinfo{person}{Doug Tygar},
  \bibinfo{person}{Moshe~Y. Vardi}, \bibinfo{person}{Gerhard Weikum}, {and}
  \bibinfo{person}{Frank Pfenning}} (Eds.). Vol.~\bibinfo{volume}{7794}.
  \bibinfo{publisher}{Springer Berlin Heidelberg}, \bibinfo{address}{Berlin,
  Heidelberg}, \bibinfo{pages}{145--160}.
\newblock
\showISBNx{978-3-642-37074-8 978-3-642-37075-5}
\urldef\tempurl%
\url{https://doi.org/10.1007/978-3-642-37075-5_10}
\showDOI{\tempurl}


\bibitem[Kristensen et~al\mbox{.}(2022)]%
        {CubicalCloTT}
\bibfield{author}{\bibinfo{person}{Magnus~Baunsgaard Kristensen},
  \bibinfo{person}{Rasmus~Ejlers M{\o}gelberg}, {and} \bibinfo{person}{Andrea
  Vezzosi}.} \bibinfo{year}{2022}\natexlab{}.
\newblock \showarticletitle{Greatest HITs: Higher inductive types in
  coinductive definitions via induction under clocks}. In
  \bibinfo{booktitle}{\emph{{LICS} '22: 37th Annual {ACM/IEEE} Symposium on
  Logic in Computer Science, Haifa, Israel, August 2 - 5, 2022}},
  \bibfield{editor}{\bibinfo{person}{Christel Baier} {and}
  \bibinfo{person}{Dana Fisman}} (Eds.). \bibinfo{publisher}{{ACM}},
  \bibinfo{pages}{42:1--42:13}.
\newblock
\urldef\tempurl%
\url{https://doi.org/10.1145/3531130.3533359}
\showDOI{\tempurl}


\bibitem[Lago et~al\mbox{.}(2014)]%
        {DBLP:conf/popl/LagoSA14}
\bibfield{author}{\bibinfo{person}{Ugo~Dal Lago}, \bibinfo{person}{Davide
  Sangiorgi}, {and} \bibinfo{person}{Michele Alberti}.}
  \bibinfo{year}{2014}\natexlab{}.
\newblock \showarticletitle{On coinductive equivalences for higher-order
  probabilistic functional programs}. In \bibinfo{booktitle}{\emph{The 41st
  Annual {ACM} {SIGPLAN-SIGACT} Symposium on Principles of Programming
  Languages, {POPL} '14, San Diego, CA, USA, January 20-21, 2014}},
  \bibfield{editor}{\bibinfo{person}{Suresh Jagannathan} {and}
  \bibinfo{person}{Peter Sewell}} (Eds.). \bibinfo{publisher}{{ACM}},
  \bibinfo{pages}{297--308}.
\newblock
\urldef\tempurl%
\url{https://doi.org/10.1145/2535838.2535872}
\showDOI{\tempurl}


\bibitem[Lindvall(2002)]%
        {lindvall_lectures_2002}
\bibfield{author}{\bibinfo{person}{T. Lindvall}.}
  \bibinfo{year}{2002}\natexlab{}.
\newblock \bibinfo{booktitle}{\emph{Lectures on the {Coupling} {Method}}}.
\newblock \bibinfo{publisher}{Dover Publications, Incorporated}.
\newblock
\showISBNx{978-0-486-42145-2}
\showLCCN{92012811}


\bibitem[M{\o}gelberg and Paviotti(2019)]%
        {DBLP:journals/mscs/MogelbergP19}
\bibfield{author}{\bibinfo{person}{Rasmus~Ejlers M{\o}gelberg} {and}
  \bibinfo{person}{Marco Paviotti}.} \bibinfo{year}{2019}\natexlab{}.
\newblock \showarticletitle{Denotational semantics of recursive types in
  synthetic guarded domain theory}.
\newblock \bibinfo{journal}{\emph{Math. Struct. Comput. Sci.}}
  \bibinfo{volume}{29}, \bibinfo{number}{3} (\bibinfo{year}{2019}),
  \bibinfo{pages}{465--510}.
\newblock
\urldef\tempurl%
\url{https://doi.org/10.1017/S0960129518000087}
\showDOI{\tempurl}


\bibitem[M{\o}gelberg and Vezzosi(2021)]%
        {DBLP:journals/corr/abs-2112-14056}
\bibfield{author}{\bibinfo{person}{Rasmus~Ejlers M{\o}gelberg} {and}
  \bibinfo{person}{Andrea Vezzosi}.} \bibinfo{year}{2021}\natexlab{}.
\newblock \showarticletitle{Two Guarded Recursive Powerdomains for Applicative
  Simulation}. In \bibinfo{booktitle}{\emph{Proceedings 37th Conference on
  Mathematical Foundations of Programming Semantics, {MFPS} 2021, Hybrid:
  Salzburg, Austria and Online, 30th August - 2nd September, 2021}}
  \emph{(\bibinfo{series}{{EPTCS}}, Vol.~\bibinfo{volume}{351})},
  \bibfield{editor}{\bibinfo{person}{Ana Sokolova}} (Ed.).
  \bibinfo{pages}{200--217}.
\newblock
\urldef\tempurl%
\url{https://doi.org/10.4204/EPTCS.351.13}
\showDOI{\tempurl}


\bibitem[M{\o}gelberg and Zwart(2023)]%
        {ZwartExtraTime}
\bibfield{author}{\bibinfo{person}{Rasmus~Ejlers M{\o}gelberg} {and}
  \bibinfo{person}{Maaike Zwart}.} \bibinfo{year}{2023}\natexlab{}.
\newblock \showarticletitle{What monads can and cannot do with a bit of extra
  time}.
\newblock \bibinfo{journal}{\emph{CoRR}}  \bibinfo{volume}{abs/2311.15919}
  (\bibinfo{year}{2023}).
\newblock
\urldef\tempurl%
\url{https://doi.org/10.48550/ARXIV.2311.15919}
\showDOI{\tempurl}
\showeprint[arXiv]{2311.15919}


\bibitem[Nakano(2000)]%
        {nakano:Modality}
\bibfield{author}{\bibinfo{person}{Hiroshi Nakano}.}
  \bibinfo{year}{2000}\natexlab{}.
\newblock \showarticletitle{A modality for recursion}. In
  \bibinfo{booktitle}{\emph{Proceedings Fifteenth Annual IEEE Symposium on
  Logic in Computer Science}}. IEEE, \bibinfo{pages}{255--266}.
\newblock


\bibitem[Paviotti et~al\mbox{.}(2015)]%
        {DBLP:journals/entcs/PaviottiMB15}
\bibfield{author}{\bibinfo{person}{Marco Paviotti},
  \bibinfo{person}{Rasmus~Ejlers M{\o}gelberg}, {and} \bibinfo{person}{Lars
  Birkedal}.} \bibinfo{year}{2015}\natexlab{}.
\newblock \showarticletitle{A Model of {PCF} in Guarded Type Theory}. In
  \bibinfo{booktitle}{\emph{The 31st Conference on the Mathematical Foundations
  of Programming Semantics, {MFPS} 2015, Nijmegen, The Netherlands, June 22-25,
  2015}} \emph{(\bibinfo{series}{Electronic Notes in Theoretical Computer
  Science}, Vol.~\bibinfo{volume}{319})},
  \bibfield{editor}{\bibinfo{person}{Dan~R. Ghica}} (Ed.).
  \bibinfo{publisher}{Elsevier}, \bibinfo{pages}{333--349}.
\newblock
\urldef\tempurl%
\url{https://doi.org/10.1016/J.ENTCS.2015.12.020}
\showDOI{\tempurl}


\bibitem[Pitts(1996)]%
        {DBLP:journals/iandc/Pitts96}
\bibfield{author}{\bibinfo{person}{Andrew~M. Pitts}.}
  \bibinfo{year}{1996}\natexlab{}.
\newblock \showarticletitle{Relational Properties of Domains}.
\newblock \bibinfo{journal}{\emph{Inf. Comput.}} \bibinfo{volume}{127},
  \bibinfo{number}{2} (\bibinfo{year}{1996}), \bibinfo{pages}{66--90}.
\newblock
\urldef\tempurl%
\url{https://doi.org/10.1006/INCO.1996.0052}
\showDOI{\tempurl}


\bibitem[Plotkin(1985)]%
        {plotkin1985denotational}
\bibfield{author}{\bibinfo{person}{Gordon~D Plotkin}.}
  \bibinfo{year}{1985}\natexlab{}.
\newblock \showarticletitle{Denotational semantics with partial functions}.
\newblock \bibinfo{journal}{\emph{Lecture at CSLI Summer School}}
  (\bibinfo{year}{1985}).
\newblock


\bibitem[Sterling et~al\mbox{.}(2022)]%
        {DBLP:journals/corr/abs-2210-02169}
\bibfield{author}{\bibinfo{person}{Jonathan Sterling}, \bibinfo{person}{Daniel
  Gratzer}, {and} \bibinfo{person}{Lars Birkedal}.}
  \bibinfo{year}{2022}\natexlab{}.
\newblock \showarticletitle{Denotational semantics of general store and
  polymorphism}.
\newblock \bibinfo{journal}{\emph{CoRR}}  \bibinfo{volume}{abs/2210.02169}
  (\bibinfo{year}{2022}).
\newblock
\urldef\tempurl%
\url{https://doi.org/10.48550/ARXIV.2210.02169}
\showDOI{\tempurl}
\showeprint[arXiv]{2210.02169}


\bibitem[Sterling et~al\mbox{.}(2023)]%
        {DBLP:journals/corr/abs-2307-16608}
\bibfield{author}{\bibinfo{person}{Jonathan Sterling}, \bibinfo{person}{Daniel
  Gratzer}, {and} \bibinfo{person}{Lars Birkedal}.}
  \bibinfo{year}{2023}\natexlab{}.
\newblock \showarticletitle{Free theorems from univalent reference types}.
\newblock \bibinfo{journal}{\emph{CoRR}}  \bibinfo{volume}{abs/2307.16608}
  (\bibinfo{year}{2023}).
\newblock
\urldef\tempurl%
\url{https://doi.org/10.48550/ARXIV.2307.16608}
\showDOI{\tempurl}
\showeprint[arXiv]{2307.16608}


\bibitem[Thorisson(2000)]%
        {thorisson/2000}
\bibfield{author}{\bibinfo{person}{Hermann Thorisson}.}
  \bibinfo{year}{2000}\natexlab{}.
\newblock \bibinfo{booktitle}{\emph{Coupling, stationarity, and regeneration}}.
\newblock \bibinfo{publisher}{Springer-Verlag}, \bibinfo{address}{New York}.
  xiv+517 pages.
\newblock
\showISBNx{0-387-98779-7}


\bibitem[{Univalent Foundations Program}(2013)]%
        {hottbook}
\bibfield{author}{\bibinfo{person}{The {Univalent Foundations Program}}.}
  \bibinfo{year}{2013}\natexlab{}.
\newblock \bibinfo{booktitle}{\emph{Homotopy Type Theory: Univalent Foundations
  of Mathematics}}.
\newblock \bibinfo{publisher}{\url{https://homotopytypetheory.org/book}},
  \bibinfo{address}{Institute for Advanced Study}.
\newblock


\bibitem[V{\'{a}}k{\'{a}}r et~al\mbox{.}(2019)]%
        {DBLP:journals/pacmpl/VakarKS19}
\bibfield{author}{\bibinfo{person}{Matthijs V{\'{a}}k{\'{a}}r},
  \bibinfo{person}{Ohad Kammar}, {and} \bibinfo{person}{Sam Staton}.}
  \bibinfo{year}{2019}\natexlab{}.
\newblock \showarticletitle{A domain theory for statistical probabilistic
  programming}.
\newblock \bibinfo{journal}{\emph{Proc. {ACM} Program. Lang.}}
  \bibinfo{volume}{3}, \bibinfo{number}{{POPL}} (\bibinfo{year}{2019}),
  \bibinfo{pages}{36:1--36:29}.
\newblock
\urldef\tempurl%
\url{https://doi.org/10.1145/3290349}
\showDOI{\tempurl}


\bibitem[Vezzosi et~al\mbox{.}(2019)]%
        {DBLP:journals/pacmpl/VezzosiM019}
\bibfield{author}{\bibinfo{person}{Andrea Vezzosi}, \bibinfo{person}{Anders
  M{\"{o}}rtberg}, {and} \bibinfo{person}{Andreas Abel}.}
  \bibinfo{year}{2019}\natexlab{}.
\newblock \showarticletitle{Cubical agda: a dependently typed programming
  language with univalence and higher inductive types}.
\newblock \bibinfo{journal}{\emph{Proc. {ACM} Program. Lang.}}
  \bibinfo{volume}{3}, \bibinfo{number}{{ICFP}} (\bibinfo{year}{2019}),
  \bibinfo{pages}{87:1--87:29}.
\newblock
\urldef\tempurl%
\url{https://doi.org/10.1145/3341691}
\showDOI{\tempurl}


\bibitem[Villani(2008)]%
        {Villani2008OptimalTO}
\bibfield{author}{\bibinfo{person}{C. Villani}.}
  \bibinfo{year}{2008}\natexlab{}.
\newblock \bibinfo{booktitle}{\emph{Optimal Transport: Old and New}}.
\newblock \bibinfo{publisher}{Springer Berlin Heidelberg}.
\newblock
\showISBNx{9783540710509}
\showLCCN{2008932183}


\bibitem[Wand et~al\mbox{.}(2018)]%
        {DBLP:journals/pacmpl/WandCGC18}
\bibfield{author}{\bibinfo{person}{Mitchell Wand}, \bibinfo{person}{Ryan
  Culpepper}, \bibinfo{person}{Theophilos Giannakopoulos}, {and}
  \bibinfo{person}{Andrew Cobb}.} \bibinfo{year}{2018}\natexlab{}.
\newblock \showarticletitle{Contextual equivalence for a probabilistic language
  with continuous random variables and recursion}.
\newblock \bibinfo{journal}{\emph{Proc. {ACM} Program. Lang.}}
  \bibinfo{volume}{2}, \bibinfo{number}{{ICFP}} (\bibinfo{year}{2018}),
  \bibinfo{pages}{87:1--87:30}.
\newblock
\urldef\tempurl%
\url{https://doi.org/10.1145/3236782}
\showDOI{\tempurl}


\bibitem[Xia et~al\mbox{.}(2020)]%
        {DBLP:journals/pacmpl/XiaZHHMPZ20}
\bibfield{author}{\bibinfo{person}{Li{-}yao Xia}, \bibinfo{person}{Yannick
  Zakowski}, \bibinfo{person}{Paul He}, \bibinfo{person}{Chung{-}Kil Hur},
  \bibinfo{person}{Gregory Malecha}, \bibinfo{person}{Benjamin~C. Pierce},
  {and} \bibinfo{person}{Steve Zdancewic}.} \bibinfo{year}{2020}\natexlab{}.
\newblock \showarticletitle{Interaction trees: representing recursive and
  impure programs in Coq}.
\newblock \bibinfo{journal}{\emph{Proc. {ACM} Program. Lang.}}
  \bibinfo{volume}{4}, \bibinfo{number}{{POPL}} (\bibinfo{year}{2020}),
  \bibinfo{pages}{51:1--51:32}.
\newblock
\urldef\tempurl%
\url{https://doi.org/10.1145/3371119}
\showDOI{\tempurl}


\bibitem[Zhang and Amin(2022)]%
        {DBLP:journals/pacmpl/ZhangA22}
\bibfield{author}{\bibinfo{person}{Yizhou Zhang} {and} \bibinfo{person}{Nada
  Amin}.} \bibinfo{year}{2022}\natexlab{}.
\newblock \showarticletitle{Reasoning about "reasoning about reasoning":
  semantics and contextual equivalence for probabilistic programs with nested
  queries and recursion}.
\newblock \bibinfo{journal}{\emph{Proc. {ACM} Program. Lang.}}
  \bibinfo{volume}{6}, \bibinfo{number}{{POPL}} (\bibinfo{year}{2022}),
  \bibinfo{pages}{1--28}.
\newblock
\urldef\tempurl%
\url{https://doi.org/10.1145/3498677}
\showDOI{\tempurl}


\end{thebibliography}

\appendix

\section{Omitted proofs}

This appendix contains proofs omitted from the main text to be used in the review process. 

\subsection*{Section~\ref{sec:cctt}}

\begin{proof}[Proof of Lemma~\ref{lem:embed:to:cirr}]
 The canonical map $A \to \forall\kappa. A$ maps $a$ to $\lambda\kappa.a$. The map defined by application
 to the clock constant $\clockc$ defines a left inverse to this map. It suffices to show that this is also
 a right inverse, i.e., that $\capp a \peq \capp[\clockc]a$ for all $a : \forall\kappa . A$ and all $\kappa$. Since
 the canonical map $B \to \forall\kappa . B$ is assumed to be an equivalence, and application to $\clockc$ is
 a left inverse, it must also be a right inverse, so the corresponding property
 $\capp b \peq \capp[\clockc]b$ for all $b : \forall\kappa . B$. Applying this to $b = \lambda\kappa. i(\capp a)$
 we get $i(\capp a) \peq i(\capp[\clockc] a)$, which by assumption implies $\capp a \peq \capp[\clockc]a$.
\end{proof}

\subsection*{Section~\ref{sec:fin:dist}}

\begin{proof}[Proof of Lemma~\ref{lem:Fin:Dist}]
 Let $\Dist'(\Fin(n)) \defeq \Sigma( f : \Fin(n) \to \IIc) . \mathsf{sum}(f) = 1$, and note
 that this carries a convex algebra structure 
 as well as contains $\Fin(n)$ via a dirac map $\DiracOp$ both defined in the standard way. We show that this is free by induction
 on $n$. In the case of $n=0$, clearly $\Dist'(\Fin(0))$ is empty, and so the case follows. For the inductive
 case, let $n$ be the element of $\Fin(n+1)$ not in the inclusion from $\Fin(n)$. Given $f : \Dist'(\Fin(n+1))$ we must define $\uext g(f)$.
 By decidability of equality for $\IIc$, we can branch on the values of $f(n)$. If $f(n) = 0$, then $f$ is in the image of the obvious
 inclusion from $\Dist'(\Fin(n))$, and so $\uext g(f)$ is defined by induction. If $f(n) = 1$ then $f= \DiracOp(n)$ and so
 $\uext g(f)$ must be $g(n)$. Finally, if $f(n) = p \in \II$, then $f = \DistChoice p{\Dirac n}{f'}$ for some $f' : \Dist'(\Fin(n))$.
 By induction we get $\uext g(f')$ and so can define $\uext g(f) = \DistChoice p{g(n)}{\uext g(f')}$.

 To prove that $\uext g$ is a
 homomorphism, suppose we are given $f : \Dist'(\Fin(n+1))$ and $h : \Dist'(\Fin(n+1))$. There are nine possible cases of $f(n)$ and $h(n)$
 and we just show the case where $f(n)= q$ and $h(n) = r$. In that case $f \peq \DistChoice q{\Dirac n}{f'}$ and
 $h \peq \DistChoice r{\Dirac n}{h'}$ for some $f', h' : \Dist'(\Fin(n+1))$. From the axioms of convex algebras, one can compute, given
 $p,q,r : \II$, probabilities $p',q',r'$ such that
 $\DistChoice p{(\DistChoice q ab)}{(\DistChoice r cd)} \peq \DistChoice{p'}{(\DistChoice{q'} ac)}{(\DistChoice{r'} bd)}$
 holds for all $a,b,c,d$ in any convex algebra. We use this to get
\begin{align*}
 \uext g(\DistChoice pfh) 
 & \peq \uext g(\DistChoice{p'}{\Dirac n}{(\DistChoice{r'}{f'}{h'})}) \\
 & \peq \DistChoice{p'}{g(n)}{\uext g(\DistChoice{r'}{f'}{h'})} \\
 & \peq \DistChoice{p'}{g(n)}{(\DistChoice{r'}{\uext g(f')}{\uext g(h')})} \\
 & \peq \DistChoice {p'}{(\DistChoice {q'} {g(n)}{g(n)})}{(\DistChoice{r'} {\uext g(f')}{\uext g(h')})} \\
 & \peq \DistChoice {p}{(\DistChoice {q} {g(n)}{\uext g(f')})}{(\DistChoice{r} {g(n)}{\uext g(h')})} \\
 & \peq \DistChoice p{\uext g(f)}{\uext g(h)}.
\end{align*}
Using the induction step to conclude  $\uext g(\DistChoice{r'}{f'}{h'}) \peq \DistChoice{r'}{\uext g(f')}{\uext g(h')}$ in the
third equality.
\end{proof}

In constructive mathematics, a set is said to be Kuratowski-finite if it is the codomain of a bijection from a set of the form
$\Fin(n)$. The distributions of $\Dist(A)$ satisfy a similar property.
%

\begin{lemma} \label{lem:dist:Kfinite}
 For any $\mu : \Dist (A)$, there exists an $n : \NN$, a map $f: \Fin(n) \to A$, and a distribution $\nu : \Dist(\Fin (n))$
 such that $\mu \peq \Dist(f)(\nu)$
\end{lemma}

\begin{proof}
 This is proved by induction on $\mu$. In the case of $\Dirac a$ take $n=1$ and $f$ the constant map to $a$. In case
 of $\DistChoice p\mu{\mu'}$ we get by induction $n, n', f : \Fin(n) \to A, f' : \Fin(n') \to A, \nu : \Dist(\Fin(n)), \nu': \Dist(\Fin(n'))$
 such that $\mu \peq \Dist(f)(\nu)$ and $\mu' \peq \Dist(f)(\nu')$. Set $m = n+ n'$ such that
 $\Fin(m) \equi \Fin(n) + \Fin(n')$, define $\rho \defeq \Dist(\inl)(\nu)$, $\rho' \defeq \Dist(\inr)(\nu')$
 and $g : \Fin(m) \to A$ to be the copairing of $f$ and $f'$. Then $\Dist(g)(\DistChoice p\rho{\rho'}) = \DistChoice p\mu{\mu'}$.
\end{proof}



To prove Theorem~\ref{thm:dist:sum:equiv}, we first prove a lemma:

\begin{lemma} \label{lem:subdist}
 The types $\Dist(A+1)$ and $1+ \IIdc\times \Dist A$ are equivalent.
\end{lemma}

\begin{proof}[Proof of Lemma~\ref{lem:subdist}]
 Using the classical definition of $\Dist$ in terms of finite maps, the equivalence
 $\Dist(A+1) \to 1+ \IIdc\times \Dist A$ would map $f$ to $\star$ if $f(\star) = 1$ and otherwise
 to the pair $(f(\star),\frac{1}{1-f(\star)}\cdot g)$ where $g$ is the restriction of $f$ to $A$.
 In the special case of $A = \Fin(n)$, by Lemma~\ref{lem:Fin:Dist}, we can indeed define the equivalence like that and
 prove it an equivalence, also in CCTT. If we do that and transport the convex algebra structure from
 $\Dist(\Fin(n + 1))$ to $1+ \IIdc\times \Dist (\Fin(n))$ we get a structure satisfying the following equalities.
%
%
\begin{align*}
 \DistChoice p\star\star & \peq \star \\
 \DistChoice p{(q,\mu)}\star & \peq (pq+(1-p),\mu) \\
 \DistChoice p\star{(r,\nu)} & \peq (p + (1-p)r,\nu) \\
 \DistChoice p{(q,\mu)}{(r,\nu)} & \peq (pq + (1-p)r, \DistChoice{\left(\frac{p(1-q)}{p(1-q) + (1-p)(1-r)}\right)}\mu\nu)
\end{align*}
To prove the general statement, we will define a convex algebra on $1+ \IIdc\times \Dist A$
using the above equations, and prove that it defines the free convex algebra on $A + 1$.

First, to see that the convex algebra operation satisfies the axioms of convex algebras,
we use the fact that it does so in the case of $A = \Fin(n)$. To prove transitivity, for example, there
are $8$ cases to consider. One of them is
\begin{align} \label{eq:subdist:assoc}
 \DistChoice{q}{\left(\DistChoice{p}{(r,\mu)}{(s,\nu)}\right)}{(t,\rho)}
 & \peq \DistChoice{pq}{(r,\mu)}{\left(\DistChoice{\frac{q-pq}{1-pq}}{(s,\nu)}{(t,\rho)}\right)}
\end{align}
for given $p,q,r,s,t : \IIdc$ and $\mu, \nu, \rho : \Dist(A)$. By Lemma~\ref{lem:dist:Kfinite} there exists an $n$, and
$\mu',\nu', \rho' : \Dist(\Fin(n))$ as well as $f: \Fin(n) \to A$ such that $\mu \peq \Dist(f)(\mu')$,
$\nu \peq \Dist(f)(\nu')$ and $\rho \peq \Dist(f)(\rho')$. (The lemma gives three different $n$,
and different functions $f$ but we can find a common domain as $n+n+n$ as in the proof of
Lemma~\ref{lem:dist:Kfinite}). Now, since (\ref{eq:subdist:assoc}) holds for $\mu', \nu', \rho'$
we can apply $\IIdc \times \Dist(f)$ to both sides and get (\ref{eq:subdist:assoc}).

To see that this defines the free convex algebra structure on $A+1$,
suppose $B$ is a convex algebra and $f: A+1 \to B$. Let $g : \Dist A \to B$ be the unique extension of
the restriction of $f$ to $A$, and define the extension $\uext f$ of $f$ by the clauses
\begin{align*}
 \uext f(\star) & \defeq f(\star) &
 \uext f(p,\mu) & \defeq \DistChoice p{f(\star)}{g(\mu)}
\end{align*}
if $p>0$ and $g(\mu)$ if $p=0$.
To show that this is a homomorphism, we must consider four cases. One is
\begin{align*}
 \uext f(\DistChoice p{(q,\mu)}{(r,\nu)}) & \peq \DistChoice p{\uext f(q,\mu)}{\uext f(r,\nu)})
\end{align*}
where (assuming $q,r$ not zero) the left hand side unfolds to
\begin{align*}
 & \DistChoice{pq + (1-p)r}{f(\star)}{g(\DistChoice{\left(\frac{p(1-q)}{p(1-q)+(1-p)(1-r)}\right)}\mu\nu)}  \\
  & \peq \DistChoice{pq + (1-p)r}{f(\star)} {\left(\DistChoice{\left(\frac{p(1-q)}{p(1-q)+(1-p)(1-r)}\right)}{g(\mu)}{g(\nu)}\right)}
\end{align*}
and the right hand side to
\[
\DistChoice p{(\DistChoice q{f(\star)}{g(\mu)})}{(\DistChoice r{f(\star)}{g(\nu)})}
\]
and now the case is easily verified using the technique of \autoref{example:proving:eq}. 

Finally, to show uniqueness, suppose $h : 1+ \IIdc\times \Dist A \to B$ is another homomorphism
extending $f$. Then clearly $h(\star) = f(\star) = \uext f(\star)$, and since $\lambda \mu . (0,\mu)$
is a homomorphism from $\Dist A$ to $1+ \IIdc\times \Dist A$, we also get
\begin{align*}
 h(0,\mu) & = g(\mu) = \uext f (0,\mu)
\end{align*}
by induction on $\mu$. So finally since $(p,\mu) = \DistChoice p{\star}{(0,\mu)}$ also $h(p,\mu) = \uext f(p,\mu)$
for any $(p,\mu)$.
\end{proof}

We can now prove the more general theorem.


\begin{proof}[Proof (sketch) of Theorem~\ref{thm:dist:sum:equiv}]
  Since the domain and codomain are both sets, it suffices to prove that $f$ is
  both surjective and injective.

  To prove surjectivity, suppose $\mu : \Dist(A+B)$. We show that $\mu$ is
  in the image of $f$ by induction on $\mu$. The case of $\mu$ being a Dirac distribution is easy.
In the case of a sum $\DistChoice p{\mu}{\nu}$, one must consider the $9$ cases of the
induction hypotheses for $\mu$ and $\nu$. Here we just do the case
where $\mu \peq \DistChoice {q}{\mu_A}{\mu_B}$ and $\nu \peq \DistChoice {r}{\nu_A}{\nu_B}$,
omitting the $\Dist(\inl)$ and $\Dist(\inr)$.
Let $p',q',r'$ be the unique elements in $\II$ such that
\[
 \DistChoice p{(\DistChoice qab)}{(\DistChoice rcd)} \peq \DistChoice {p'}{(\DistChoice {q'}ac)}{(\DistChoice {r'}bd)}
\]
holds for all $a,b,c,d$. Then
\begin{align*}
 \DistChoice{p'}{(\DistChoice {q'}{\mu_A}{\nu_A})}{(\DistChoice {r'}{\mu_B}{\nu_B})}
 & \peq \DistChoice{p}{(\DistChoice {q}{\mu_A}{\mu_B})}{(\DistChoice {r}{\nu_A}{\nu_B})} \\
 & \peq \DistChoice p{\mu}{\nu}.
\end{align*}

 To prove injectivity, there are $8$ cases to cover.
 Suppose for example that $\DistChoice p{\mu_A}{\mu_B} \peq \DistChoice q{\nu_A}{\nu_B}$ for
 $\mu_A,\nu_A : \Dist(A)$ and $\mu_B,\nu_B : \Dist(B)$. Then also
\begin{align*}
\DistChoice p{\mu_A}{\Dirac{\star}}
& \peq \Dist(A + !)(\DistChoice p{\mu_A}{\mu_B}) \\
& \peq \Dist(A + !)(\DistChoice q{\nu_A}{\nu_B}) \\
& \peq \DistChoice q{\nu_A}{\Dirac{\star}},
\end{align*}
 which via the equivalence of Lemma~\ref{lem:subdist} allows us to conclude that $p\peq q$ and $\mu_A \peq \nu_A$. Similarly, we
 can prove that $\mu_B \peq \nu_B$. The other seven cases are similar.
\end{proof}

Alternatively, one can prove Theorem~\ref{thm:dist:sum:equiv} directly by constructing a
convex algebra structure on the $3$-fold sum. Proving Lemma~\ref{lem:subdist} first
reduces the number of cases for associativity from $27$ to $8$.

\subsection*{Section~\ref{sec:convex:delay:alg}}


\begin{proof}[Proof of Lemma~\ref{lemma:leadsto-run}] $\;$ \\
\begin{enumerate}
  \item By induction on $n$. For $n = 0$, note that $\run^{n} \nu = \nu$. So the statement simplifies to $\nu \leadsto \nu$, which holds by definition.
For $n = n' + 1$, we consider the different cases for $\nu$.
\begin{itemize}
\item If $\nu = \CEta{a}$, then $\run^{n} \nu = \nu$ and the statement holds by definition.
\item If $\nu = \CStep \nu'$, then we know $\nu \leadsto \nu'$ by definition of $\leadsto$. In addition, $\run^{n} \nu = \run^{n'} \nu'$, and we know from our IH that $\nu' \leadsto \run^{n'} \nu'$. Hence by transitivity of $\leadsto$: $\nu \leadsto \run^{n} \nu$.
\item If $\nu = \CChoice{p}{\nu_1}{\nu_2}$, then $\run^{n} \nu = \CChoice{p}{\run^{n} \nu_1}{\run^{n} \nu_2}$, then we may assume that $\nu_1 \leadsto \run^{n} \nu_1$ and $\nu_2 \leadsto \run^{n} \nu_2$. Then by the choice axiom of $\leadsto$ also:  $\nu \leadsto \CChoice{p}{\run^{n}\nu_1}{\run^{n}\nu_2} = \run^{n} (\CChoice{p}{\nu_1}{\nu_2}) = \run^{n} \nu$.
\end{itemize}

\item By induction on $\leadsto$.
\begin{itemize}
\item For the base case $\nu = \nu'$, the statement is trivial.
\item For $\nu = \CStep \nu'$, it follows trivially that $\run^0 \nu \leadsto \run^0 \nu'$. For $n \geq 1$, we have $\run^{n} \nu = \run^{n-1} \nu'$. By the first part of this Lemma, $\run^{n-1} \nu' \leadsto \run(\run^{n-1} \nu') = \run^n \nu'$.
\item If there is a $\nu''$ such that $\nu \leadsto \nu''$ and $\nu'' \leadsto \nu'$, then by induction we may assume that $\run^{n} \nu \leadsto \run^{n} \nu''$ and $\run^{n} \nu'' \leadsto \run^{n} \nu'$. Then by transitivity of $\leadsto$ also $\run^{n} \nu \leadsto \run^{n} \nu'$.
\item Lastly, if $\nu = \CChoice{p}{\nu_1}{\nu_2}$ and $\nu' = \CChoice{p}{\nu_1'}{\nu_2'}$ such that $\nu_1 \leadsto \nu_1'$ and $\nu_2 \leadsto \nu_2'$, then by induction we may assume that also $\run^{n} \nu_1 \leadsto \run^{n} \nu_1'$, and similarly for $\nu_2$ and $\nu_2'$. Then also:
\begin{align*}
  \run^{n} \nu  & = \run^{n} \CChoice{p}{\nu_1}{\nu_2} \\
   & = \CChoice{p}{\run^{n} \nu_1}{\run^{n}\nu_2} \\
   & \leadsto \CChoice{p}{\run^{n} \nu_1'}{\run^{n}\nu_2'} \\
   & = \CChoice{p}{\run^{n} \nu_1'}{\run^{n}\nu_2'} \\
   & = \run^{n} \nu'
\end{align*}
\end{itemize}

\item By induction on $\leadsto$.
\begin{itemize}
\item The base case for $\nu' = \nu$ is trivial.
\item For $\nu = \CStep \nu'$, we have that $\run^{1} \nu = \nu'$, and hence $\nu' \leadsto \run^{1} \nu$.
\item If there is a $\nu''$ such that $\nu \leadsto \nu''$ and $\nu'' \leadsto \nu'$, then by induction there is an $m_1$ such that $\nu'' \leadsto \run^{m_1} \nu$ and there is an $m_2$ such that $\nu' \leadsto \run^{m_2} \nu''$. Then the second part of this lemma, also
$\run^{m_2} \nu'' \leadsto \run^{m_2}\run^{m_1} \nu = \run^{m_1 + m_2} \nu$, and hence by transitivity: $\nu' \leadsto \run^{m_1 + m_2} \nu$.
\item Lastly, if $\nu = \CChoice{p}{\nu_1}{\nu_2}$ and $\nu' = \CChoice{p}{\nu_1'}{\nu_2'}$ such that $\nu_1 \leadsto \nu_1'$ and $\nu_2 \leadsto \nu_2'$, then by induction there is an $m_1$ such that $\nu_1' \leadsto \run^{m_1} \nu_1$ and an $m_2$ such that $\nu_2' \leadsto \run^{m_2} \nu_2$. Let $n = max(m_1, m_2)$. Then by the first part of this lemma: $\nu_1' \leadsto \run^{n} \nu_1$ and $\nu_2' \leadsto \run^{n} \nu_2$. So then also $\nu' \leadsto \run^{n} \nu$.
\end{itemize}
\end{enumerate}
\end{proof}

\begin{proof}[Proof of Lemma~\ref{leadsto:lem4}]

By induction on $\leadsto$:
\begin{itemize}
\item The first base case is obvious by reflexivity.
\item For the second base case, note that:
\[
\uext f(\CStep \nu) = \CStep (\uext f(\nu)) \leadsto \uext f(\nu)
\]
\item If $\nu \leadsto \nu'$ came from the transitivity axiom, then there is a $\nu''$ such that $\nu \leadsto \nu''$ and $\nu'' \leadsto \nu'$. By the induction hypothesis we know that then also $\uext f(\nu) \leadsto \uext f(\nu'')$ and $\uext f(\nu'') \leadsto \uext f(\nu')$. Then by transitivity of $\leadsto$: $\uext f(\nu) \leadsto \uext f(\nu')$.
\item If $\nu = \CChoice{p}{\nu_1}{\nu_2}$ and $\nu' = \CChoice{p}{\nu_1'}{\nu_2'}$, where $\nu_1 \leadsto \nu_1'$ and $\nu_2 \leadsto \nu_2'$, then by the induction hypothesis: $\uext f(\nu_1) \leadsto \uext f(\nu_1')$ and $\uext f(\nu_2) \leadsto \uext f(\nu_2')$. Then:
    \begin{align*}
    \uext f(\nu)    & = \uext f(\CChoice{p}{\nu_1}{\nu_2})\\
                    & = \CChoice{p}{(\uext f(\nu_1))}{(\uext f(\nu_2))} \\
                    & \leadsto \CChoice{p}{(\uext f(\nu_1'))}{(\uext f(\nu_2'))} \\
                    & = \uext f(\CChoice{p}{\nu_1'}{\nu_2'}) \\
                    & = \uext f(\nu')
    \end{align*}
\end{itemize}
\end{proof}


\begin{proof}[Proof of Lemma~\ref{leadsto:lem5}] $\;$ \\
\begin{enumerate}
\item For the first statement, we prove that $\probterm{n}{\nu} \leq \probterm{n+1}{\nu}$ by induction on $n$. For $n = 0$, we do a case analysis of $\nu$:
\begin{itemize}
  \item If $\nu = \CDirac a$, then by definition $\probterm{0}{\CDirac a} = 1$, and $\probterm{1}{\CDirac a} = \probterm{0}{\run(\CDirac a)} = \probterm{0}{\CDirac a} = 1$. So in this case we have equality.
  \item If $\nu = \CStep \nu'$, then by definition $\probterm{0}{\CStep \nu'} = 0$ and $\probterm{1}{\CStep \nu'} = \probterm{0}{\run(\CStep \nu')} = \probterm{0}{\nu'} \geq 0$.
  \item If $\nu = \CChoice{p}{\nu_1}{\nu_2}$, then we may assume that $\probterm{0}{\nu_1} \leq \probterm{1}{\nu_1}$ and $\probterm{0}{\nu_2} \leq \probterm{1}{\nu_2}$. Then also:
      \begin{align*}
      \probterm{0}{\CChoice{p}{\nu_1}{\nu_2}} & = \CChoice{p}{\probterm{0}{\nu_1}}{\probterm{0}{\nu_2}} \\
      & \leq \CChoice{p}{\probterm{1}{\nu_1}}{\probterm{1}{\nu_2}} \\
      & = \CChoice{p}{\probterm{0}{\run(\nu_1)}}{\probterm{0}{\run(\nu_2)}} \\
      & = \probterm{0}{\run(\CChoice{p}{\nu_1}{\nu_2})} \\
      & = \probterm{1}{\CChoice{p}{\nu_1}{\nu_2}}
      \end{align*}
\end{itemize}

For $n = n' + 1$, notice that $\probterm{n + 1}{\nu} = \probterm{n}{\run(\nu)}$. We again go by case analysis of $\nu$:
\begin{itemize}
\item If $\nu = \CDirac a$ for some $a : A$, then $\probterm{n}{\CDirac a} = \probterm{n + 1}{\CDirac a} = 1$.
\item If $\nu = \CStep \nu'$, then $\probterm{n}{\CStep \nu'} = \probterm{n'}{\run(\CStep \nu')} = \probterm{n'}{\nu'}$, and $\probterm{n + 1}{\CStep \nu'} = \probterm{n}{\run(\CStep \nu')} = \probterm{n}{\nu'}$. By the induction hypothesis for $n$, we know that $\probterm{n'}{\nu'} \leq \probterm{n}{\nu'}$, and hence $\probterm{n}{\CStep \nu'} \leq \probterm{n + 1}{\CStep \nu'}$.
\item If $\nu = \CChoice{p}{\nu_1}{\nu_2}$, then $\probterm{n}$ distributes through the sum and via a similar reasoning as above we get to the conclusion.
\end{itemize}

\item The second statement of this Lemma we also prove by induction on $n$. For $n = 0$ we go by induction on $\leadsto$.
\begin{itemize}
\item If $\nu = \nu'$, then the statement is trivially true.
\item If $\nu = \CStep \nu'$, by definition $\probterm{0}{\nu} = 0$, and hence: $\probterm{0}{\nu} \leq \probterm{0}{\nu'}$.
\item If there is a $\nu''$ such that $\nu \leadsto \nu''$ and $\nu'' \leadsto \nu'$, then by the induction hypothesis we have $\probterm{0}{\nu} \leq \probterm{0}{\nu''}$ and $\probterm{0}{\nu''} \leq \probterm{0}{\nu'}$. Then by transitivity of $\leq$, we also have $\probterm{0}{\nu} \leq \probterm{0}{\nu'}$.
\item Lastly, if $\nu = \CChoice{p}{\nu_1}{\nu_2}$ and $\nu' = \CChoice{p}{\nu_1'}{\nu_2'}$ such that $\nu_1 \leadsto \nu_1'$ and $\nu_2 \leadsto \nu_2'$, then by the induction hypothesis we have $\probterm{0}{\nu_1} \leq \probterm{0}{\nu_1'}$ and $\probterm{0}{\nu_2} \leq \probterm{0}{\nu_2'}$. Then:
\begin{align*}
  \probterm{0}{\nu} & = p\cdot\probterm{0}{\nu_1} + (1-p)\cdot\probterm{0}{\nu_2} \\
                    & \leq p\cdot\probterm{0}{\nu_1'} + (1-p)\cdot\probterm{0}{\nu_2'} \\
                    & = \probterm{0}{\nu'}.
\end{align*}
\end{itemize}
For $n > 0$, we use that $\probterm{n}{\nu} = \probterm{0}{\run^n \nu}$. By Lemma~\ref{lemma:leadsto-run}, we know that $\nu \leadsto \nu'$ implies $\run^n \nu \leadsto \run^n \nu'$, and hence by the above $\probterm{0}{\run^n \nu} \leq \probterm{0}{\run^n \nu'}$, which means that $\probterm{n}{\nu} \leq \probterm{n}{\nu'}$.
\item The third statement is proven by induction on $\leadsto$.
\begin{itemize}
	\item If $\nu = \nu'$ then we use reflexivity of $\eqlim$.
	\item If $\nu = \CStep \nu'$. Let $n \colon \mathbb{N}$, $\varepsilon > 0$. Then, by definition we have
		$\probterm{n}{\nu'} = \probterm{n + 1}{\nu} \leq  \probterm{n + 1}{\nu} + \varepsilon$, so
		$\probtermseq{\nu'} \leqlim \probtermseq{\nu}$. On the other direction, by monotonicity we
		have $\probterm{n}{\nu} \leq \probterm{n + 1}{\nu} = \probterm{n}{\nu'} \leq \probterm{n}{\nu'} + \varepsilon$,
		so $\probtermseq{\nu} \leqlim \probtermseq{\nu'}$.
	\item If there is a $\nu''$ such that $\nu \leadsto \nu''$ and $\nu'' \leadsto \nu'$, then we apply the induction hypothesis and transitivity of $\eqlim$.
	\item If $\nu = \CChoice{p}{\nu_1}{\nu_2}$, $\nu' = \CChoice{p}{\nu_1'}{\nu_2'}$, $\nu_1 \leadsto \nu'_1$, and $\nu_2 \leadsto \nu_2'$, then we apply the induction hypothesis and Lemma~\ref{lemma:leqlim-conv}.
\end{itemize}

\end{enumerate}
\end{proof}

\begin{proof}[Proof of Lemma~\ref{lemma:limleadsto-properties}] $\;$ \\
\begin{enumerate}
\item For reflexivity, note that $\leadsto$ is reflexive and that for all $\nu : \LiftDist{A}$ and all $p \in \II$, $\nu \peq \CChoice{p}{\nu}{\nu}$. We hence have: $\nu \leadsto \CChoice{p}{\nu}{\nu}$, proving that $\nu \limleadsto \nu$.
\item For transitivity, suppose that $\nu \limleadsto \nu_1$ and $\nu_1 \limleadsto \nu_2$. We need to show that $\nu \limleadsto \nu_2$. By definition, we have that for all $p, q \in \II$ there exist $\nu_1', \nu_2'$ such that:
\begin{align*}
\nu & \leadsto \CChoice{p}{\nu_1}{\nu_1'} \\
\nu_1 & \leadsto \CChoice{q}{\nu_2}{\nu_2'}
\end{align*}
and so:
\begin{align*}
\nu & \leadsto \CChoice{p}{\nu_1}{\nu_1'} \\
    & \leadsto \CChoice{p}{(\CChoice{q}{\nu_2}{\nu_2'})}{\nu_1'} \\
    & \leadsto \CChoice{pq}{\nu_2}{(\CChoice{\frac{p-pq}{1-pq}}{\nu_2'}{\nu_1'})}
\end{align*}
 
Given any $r \in \II$, we can pick $p = \frac{1+r}{2}$ and $q = \frac{r}{p} = \frac{2r}{1+r}$. We then have:
\[
\nu \leadsto \CChoice{r}{\nu_2}{(\CChoice{\frac{p-r}{1-r}}{\nu_2'}{\nu_1'})},
\]
proving that indeed, $\nu \limleadsto \nu_2$.
\item Since $\CStep \nu \leadsto \nu$, we again use $\nu \peq \CChoice{p}{\nu}{\nu}$ to show that for all $p \in \II$: $\CStep \nu \leadsto \CChoice{p}{\nu}{\nu}$ and hence $\CStep \nu \limleadsto \nu$. 
\item Regarding probabilistic choice, suppose that $\nu_1 \limleadsto \nu_1'$ and $\nu_2 \limleadsto \nu_2'$. Then, for all $p \in \II$, we know that there exist $\nu_1''$ and $\nu_2''$ such that $\nu_1 \leadsto \CChoice{p}{\nu_1'}{\nu_1''}$ and $\nu_1 \leadsto \CChoice{p}{\nu_2'}{\nu_2''}$. Then, for any $q \in \II$:
    \begin{align*}
      \CChoice{q}{\nu_1}{\nu_2} & \leadsto \CChoice{q}{(\CChoice{p}{\nu_1'}{\nu_1''})}{(\CChoice{p}{\nu_2'}{\nu_2''})} \\
       & \peq \CChoice{p}{(\CChoice{q}{\nu_1'}{\nu_2'})}{(\CChoice{q}{\nu_1''}{\nu_2''})}
    \end{align*}
    Hence for all $p \in \II$ we have a $\nu'$ (namely $\CChoice{q}{\nu_1''}{\nu_2''}$) such that $\CChoice{q}{\nu_1}{\nu_2} \leadsto \CChoice{p}{(\CChoice{q}{\nu_1'}{\nu_2'})}{\nu'}$, so indeed: $\CChoice{q}{\nu_1}{\nu_2} \limleadsto \CChoice{q}{\nu_1'}{\nu_2'}$.    
\end{enumerate}
\end{proof}

\begin{proof}[Proof of Lemma~\ref{lemma:limleadsto-homopreserve}]
  We are given that $\nu \limleadsto \nu'$. That is, for all $p \in \II$, there is a $\nu'' : \LiftDist A$ such that $\nu \leadsto \CChoice{p}{\nu'}{\nu''}$. By Lemma~\ref{leadsto:lem4}, the bind operation preserves the $\leadsto$ relation, and hence:
  \begin{align*}
    \uext f(\nu) & \leadsto \uext f(\CChoice{p}{\nu'}{\nu''}) \\
     & \peq \CChoice{p}{\uext f(\nu')}{\uext f(\nu'')},
  \end{align*} 
  and hence $\uext f(\nu) \limleadsto \uext f(\nu')$.
\end{proof}

\begin{proof}[Proof of Lemma~\ref{lemma:limleadsto-leadsto-interact}]
For the first statement, if $\nu \leadsto \nu_1$ then we use that for all $p \in \II$: $\nu_1 \peq \CChoice{p}{\nu_1}{\nu_1}$ and hence for all $p \in \II$: $\nu \leadsto \CChoice{p}{\nu_1}{\nu_1}$, proving that $\nu \limleadsto \nu_1$.

For the second statement, suppose $\nu \leadsto \nu_1$ and $\nu_1 \limleadsto \nu_2$. Then, for all $p \in \II$, there exists a $\nu_2'$ such that $\nu_1 \leadsto \CChoice{p}{\nu_2}{\nu_2'}$. So by transitivity of $\leadsto$:
\begin{align*}
 \nu & \leadsto \nu_1 \\
     & \leadsto \CChoice{p}{\nu_2}{\nu_2'},
\end{align*}
which shows that $\nu \limleadsto \nu_2$.
(Alternatively: use the first part of this lemma to conclude that $\nu \limleadsto \nu_1$, and apply transitivity of $\limleadsto$.)

For the third part, suppose $\nu \limleadsto \nu_1$ and $\nu_1 \leadsto \nu_2$. Then, for all $p \in \II$, there exists a $\nu_1'$ such that $\nu \leadsto \CChoice{p}{\nu_1}{\nu_1'}$. So by the transitivity and sum rules of $\leadsto$:
\begin{align*}
 \nu & \leadsto \CChoice{p}{\nu_1}{\nu_1'} \\
     & \leadsto \CChoice{p}{\nu_2}{\nu_1'},
\end{align*}
which shows that $\nu \limleadsto \nu_2$.
(Again, this also follows immediately from the first part of this lemma and transitivity of $\limleadsto$.)
\end{proof}

\begin{proof}[Proof of Lemma~\ref{lemma:limleadsto-eqlim}]
         Suppose that $\nu_1 \limleadsto \nu_2$. Let us first prove that $\probtermseq{\nu_1} \leqlim \probtermseq{\nu_2}$.
	Let $n \colon \mathbb{N}$ and $\varepsilon > 0$. We can assume wlog. that $\varepsilon < 1$, otherwise the proof
	is trivial. By definition of $\limleadsto$,
	there exists $\nu_3 : \LiftDist A$ such that $\nu_1 \leadsto \CChoice{(1-\varepsilon)}{\nu_2}{\nu_3}$. Then, by Lemma~\ref{leadsto:lem5}.
	we have:
	\[
		\probterm{n}{\nu_1} \leq (1-\varepsilon)\cdot\probterm{n}{\nu_2} + \varepsilon \cdot \probterm{n}{\nu_3}
		\leq \probterm{n}{\nu_2} + \varepsilon
	\]

	Now we will prove $\probtermseq{\nu_2} \leqlim \probtermseq{\nu_1}$. Again, let $n \colon \mathbb{N}$ and $0 < \varepsilon < 1$.
	Then there exists $\nu_3 : \LiftDist A$ such that $\nu_1 \leadsto \CChoice{1-\frac{\varepsilon}{2}}{\nu_2}{\nu_3}$. By Lemma~\ref{leadsto:lem5} 
	$\probtermseq{\CChoice{1-\frac{\varepsilon}{2}}{\nu_2}{\nu_3}} \leqlim \probtermseq{\nu_1}$, so there exists $m$ such that
	\[
		(1-\frac{\varepsilon}{2}) \cdot \probterm{n}{\nu_2} + \frac{\varepsilon}{2} \cdot \probterm{n}{\nu_3} \leq \probterm{m}{\nu_1} + \frac{\varepsilon}{2}
	\]
	And therefore,
	\[
		\probterm{n}{\nu_2} \leq \probterm{m}{\nu_1} + \frac{\varepsilon}{2} + \frac{\varepsilon}{2} \cdot (\probterm{n}{\nu_2} - \probterm{n}{\nu_3})
		\leq \probterm{m}{\nu_1} + \frac{\varepsilon}{2} + \frac{\varepsilon}{2}
	\]
	This concludes the proof.

\end{proof}

\subsection*{Section~\ref{sec:prob:FPC}}

\begin{proof}[Proof of Example~\ref{ex:Y:comb}]

  Let
  \[
  E_V \triangleq \lambda (\Lam{M}).\GStep{\parens{\lambda(\alpha:\kappa).\GEval(M[V/x])}}
  \]
  Note that
  \[
	  \GEval{(\YComb{\varphi})} \peq \ExpStep(\GEval{(\Lam{(e_{\varphi} (\Fold{e_{\varphi}})) x})})
  \]
  and
  \[
	  \GEval{(e_{\varphi} (\Fold{e_{\varphi}}))} \peq (\ExpStep)^{3}(\GEval{(\varphi (\Lam{(e_\varphi (\Fold{e_\varphi}) x)})})
  \]
  Therefore,
   \begin{align*}
   \GEval{((\YComb{\varphi})(V))} 
	   \peq\ & (\ExpStep)^{2}(\GEval{(e_{\varphi}(\Fold{e_{\varphi}})) V)}) \\
    \peq\ & \begin{aligned}
	    (\ExpStep)^{5}(\GEval{(\varphi (\Lam{(e_\varphi (\Fold{e_\varphi}) x)})} \GBind E_V)
    \end{aligned} \\
    \peq\ & \begin{aligned}
	    (\ExpStep)^{5}(\GDirac{(\varphi)} \GBind \lambda (\Lam[Z]{P}). \GEval{(\Lam{(e_\varphi (\Fold{e_\varphi}) x)})} \\
	    \GBind \lambda U. \ExpStep(\GEval{(P[U/z])} \GBind E_V))
  \end{aligned}\\
	   \peq\ & \begin{aligned}(\ExpStep)^{4}(\GDirac{(\varphi)} \GBind \lambda (\Lam[Z]{P}). \ExpStep.\GEval{(\Lam{(e_\varphi (\Fold{e_\varphi}) x)})} \\
		   \GBind \lambda U. \ExpStep(\GEval{(P[U/z])} \GBind E_V))
    \end{aligned} \\
	   \peq\ & \begin{aligned}(\ExpStep)^{4}(\GDirac{(\varphi)} \GBind \lambda (\Lam[Z]{P}). \GEval{(\YComb{\varphi})} \\
		   \GBind \lambda U. \ExpStep(\GEval{(P[U/z])} \GBind E_V))
    \end{aligned} \\
	   \peq\ & (\ExpStep)^{4}(\GEval((\varphi(\YComb{\varphi}))) \GBind E_V) \\
	   \peq\ & (\ExpStep)^{4}(\GEval{((\varphi(\YComb{\varphi}))(V))})
  \end{align*}

  This completes the proof.

\end{proof}

\subsection*{Section~\ref{sec:den:sem}}

\begin{proof}[Proof of Lemma~\ref{lem:subst-lemma}]
    We proceed by induction on term derivations.

  \begin{case}{$M = x$}\\
    Then
    \begin{align*}
      &\GInterpret{x[V/x]}{\rho}
      \peq \GInterpret{V}{\rho}
      \peq \GInterpret{x}{\rho.x\mapsto\GValInterpret{V}{\rho}}
    \end{align*}
  \end{case}

  \begin{case}{$\Gamma \vdash \Lam{M'} :\FnTy{\tau_1}{\tau_2}$}\\
    Assume without loss of generality that $y \not= x$.
    Otherwise, use an $\alpha$-equivalent term where this condition is true.
    Consequently, we have that
    \begin{align*}
      \GInterpret{\Lam{M'}[V/y]}\rho 
      \peq\ &\GDirac\parens{\GValInterpret{\Lam{(M'[V/y])}}{\rho}} \\
      \peq\ &\GDirac\parens{\lambda (v:\GSVal{\tau_1}).\GInterpret{M'[V/y]}{\rho,x\mapsto v}} \\
      \peq\ &\GDirac\left(\lambda (v:\GSVal{\tau_1}).\GInterpret{M'}{\rho,x\mapsto v,y\mapsto\GValInterpret{V}{\rho}}\right) \\
      \peq\ &\GDirac\left(\GValInterpret{\Lam{M'}}{\rho,y\mapsto\GValInterpret{V}{\rho}}\right) \\
      \peq\ &\GInterpret{\Lam{M'}}{\rho.y\mapsto\GValInterpret{V}{\rho}}
    \end{align*}
  \end{case}

  \begin{case}{$M = MN$}\\
    We have that
    \begin{align*}
      \GInterpret{(M\,N)[V/x]}{\rho} 
      \peq\ &\GInterpret{M[V/x]}\rho \GBind \lambda f. \GInterpret{N[V/x]}\rho \GBind \lambda v. fv \\ 
      \peq\ &\GInterpret{M}{\rho,x\mapsto\GValInterpret{V}\rho} \GBind \lambda f. \GInterpret{N}{\rho,x\mapsto\GValInterpret{V}\rho} \GBind \lambda v. fv \\ 
      \peq\ &\GInterpret{M\,N}{\rho,x\mapsto\GValInterpret{V}\rho}
    \end{align*}
  \end{case}

  \begin{case}{$M = \Fold{M'}$}\\
    We have that
    \begin{align*}
      \GInterpret{\Fold{M'}[V/x]}\rho 
      \peq \ &\Dist{(\GNext)}\left(\GInterpret{M'[V/x]}\rho\right) \\
      \peq \ &\Dist{(\GNext)}\left(\GInterpret{M'}{\rho,x\mapsto\GValInterpret{V}\rho}\right) \\
      \peq \ &\GInterpret{\Fold{M'}}{\rho,x\mapsto\GValInterpret{V}\rho}
    \end{align*}
  \end{case}

  \begin{case}{$M = \Unfold{M'}$}\\
    We have
    \begin{align*}
      \GInterpret{\Unfold{M'}[V/x]}\rho 
      \peq \ & \GInterpret{M'[V/x]}\rho\GBind \lambda v. \GStep (\lambda \alpha. \GDirac{(v[\alpha])}) \\
      \peq \ & \GInterpret{M'}{\rho,x\mapsto \GValInterpret{V}\rho}\GBind \lambda v. \GStep (\lambda \alpha. \GDirac{(v[\alpha])}) \\
      \peq \ &\GInterpret{\Unfold{M'}}{\rho,x\mapsto\GValInterpret{V}\rho}
    \end{align*}
  \end{case}

  \begin{case}{$M = \Choice{q}{N_1}{N_2}$}
    We have
    \begin{align*}
      \GInterpret{\Choice{q}{N_1}{N_2}[V/x]}\rho 
      \peq\ &\GChoice{q}{\GInterpret{N_1[V/x]}\rho}{\GInterpret{N_2[V/x]}\rho} \\
      \peq\ &\GChoice{q}{\GInterpret{N_1}{\rho,x\mapsto\GInterpret{V}\rho}}{\GInterpret{N_2}{\rho,x\mapsto\GInterpret{V}\rho}} \\
      \peq\ &\GInterpret{\Choice{q}{N_1}{N_2}}{\rho,x\mapsto\GInterpret{V}\rho}
    \end{align*}
  \end{case}
  The remaining cases are similar.
\end{proof}

\begin{proof}[Proof of equation (\ref{eq:den:Yf:open})]
  Note that $f$ is a value by assumption and $\YComb{}$ as well as $e_f$ are values by definition. Let $\rho \defeq x \mapsto v$,
  and write $\ExpStep$ for $\GStep{}\circ \GNext$. Let also $\GInterpret{M}\rho \bullet \GInterpret{N}\rho$ be shorthand for  
  $\GInterpret{M}\rho \GBind \lambda f. \GInterpret{N}\rho \GBind \lambda v. fv$. Then we have:
  \begin{align*}
   &\GInterpret{\YComb\, f\,x}\rho \\
   \peq\ & \GInterpret{\YComb \, f}{\rho} \bullet \GInterpret{x}\rho \\
   \peq\ & \GInterpret{e_{f}(\Fold{(e_{f})})}\rho \bullet \GInterpret{x}\rho\\
   \peq\ & \GInterpret{\LetIn{y'}{\Unfold{(\Fold{e_f})}}f(\Lam{(y' (\Fold{e_f})) x})}\rho \bullet \GInterpret{x}\rho\\
   \peq\ &\left(\GInterpret{\Lam[y']{f(\Lam{(y' (\Fold{e_f})) x})}}\rho \bullet \GInterpret{\Unfold{(\Fold{e_f})}}\rho\right)
    \bullet\ \GInterpret{x}\rho \\
   \peq\ & \left((\ExpStep)\GInterpret{(f(\Lam{(y' (\Fold{e_f})) x}))}{\rho,y'\mapsto\GValInterpret{e_f}\rho}\right) \bullet \GInterpret{x}\rho\\
   \peq\ & \left((\ExpStep)\left(\GInterpret{f}\rho \bullet \GInterpret{\Lam{(e_f (\Fold{e_f})) x}}{\rho}\right)\right) \bullet \GInterpret{x}\rho\\
   \peq\ & \left(\ExpStep(\GInterpret{f}\rho \bullet \GInterpret{\YComb(f)}{\rho})\right) \bullet \GInterpret{x}\rho\\
   \peq\ & (\ExpStep)\left(\left(\GInterpret{f(\YComb(f))}{\rho}\right) \bullet \GInterpret{x}\rho\right)\\
   \peq\ & (\ExpStep)\left(\GInterpret{f(\YComb(f))\,x}{\rho}\right)\\
   \end{align*}
\end{proof}

\subsection*{Section~\ref{sec:couplings}}

\begin{proof}[Proof of Lemma~\ref{lem:CRelLift:coinductive}]
  The proof is organised as follows: We first show how to apply~\cite[Theorem~4.3]{CubicalCloTT} to encode
  a coinductive type as the one described in Lemma~\ref{lem:CRelLift:coinductive}. Then we show that the
  encoding is the same as $\CRelLift{\Rel}$.

  Recall first the indexed version of~\cite[Theorem~4.3]{CubicalCloTT}: If \[F : (I \to \Univ) \to (I\to \Univ)\]
  is a functor in the naive sense, such that for any $X : \forall\kappa . (I \to \Univ)$, the canonical map
  \[
  F(\forall\kappa. X) \to \forall\kappa . F(\capp X)
  \]
  is an equivalence, then a final coalgebra for $F$ can be encoded by first solving $F(\Later X) \equi X$ and then defining
  the final coalgebra as $\forall\kappa. X$. Here $\forall \kappa . X$ means $\lambda (i:I). \forall\kappa . (\capp X \,i)$, and
  $\Later X$ means $\lambda (i : I). \Later(X\, i)$

  For this application, take $I = \LiftDist(A) \times \LiftDist(A)$, and define $F(P)(\mu, \nu)$ by the
  clauses of Lemma~\ref{lem:CRelLift:coinductive}, i.e., either (spelling out the implicit inclusions)
  \begin{enumerate}
\item   there exists $\mu' : \Dist{A}$ such that $\mu \peq \Dist(\inl)(\mu')$, 
 and there is a $\nu' : \Dist{B}$ such that $\nu \limleadsto \nu'$, and there is a $\rho: \CoupType\Rel{\mu'}{\nu'}$.
\item  there exists $\mu' : \Dist(\LiftDist{A})$ such that $\mu \peq \Dist(\inr)(\mu')$ and $P(\run(\mu), \nu)$.
\item there exist $\mu_1 : \Dist{A}, \mu_2 : \Dist(\LiftDist{A})$, and $p : \II$ such that
$\mu \peq \DistChoice{p}{\Dist(\inl)(\mu_1)}{\Dist(\inr)(\mu_2)}$, there exist $\nu_1, \nu_2 : \LiftDist{B}$ such that $\nu \limleadsto \CChoice{p}{\nu_1}{\nu_2}$, 
and $F(P)(\Dist(\inl)(\mu_1), \nu_1)$ and $F(P)(\Dist(\inr)(\mu_2), \nu_2)$.
\end{enumerate}
Note the use of $F(P)$ in condition (3), which is due to the two recursive calls just being shorthand for (1) and (2). 

%
    
Given $P : \forall\kappa. (\LiftDist(A) \times \LiftDist(A) \to \Univ)$ we show that $\forall\kappa . F(\capp P) \equi F(\forall\kappa . P)$,
the right hand side of this is the statement that either
\begin{enumerate}
\item   there exists $\mu' : \Dist{A}$ such that $\mu \peq \Dist(\inl)(\mu')$, 
 and there is a $\nu' : \Dist{B}$ such that $\nu \limleadsto \nu'$, and there is a $\rho: \CoupType\Rel{\mu'}{\nu'}$.
\item  there exists $\mu' : \Dist(\LiftDist{A})$ such that $\mu \peq \Dist(\inr)(\mu')$ and $\forall\kappa. \capp{P}(\run(\mu), \nu)$.
\item there exist $\mu_1 : \Dist{A}, \mu_2 : \Dist(\LiftDist{A})$, and $p : \II$ such that
$\mu \peq \DistChoice{p}{\Dist(\inl)(\mu_1)}{\Dist(\inr)(\mu_2)}$, and there exist $\nu_1, \nu_2 : \LiftDist{B}$ such that $\nu \limleadsto \CChoice{p}{\nu_1}{\nu_2}$, 
and $F(\forall\kappa. \capp{P})(\Dist(\inl)(\mu_1), \nu_1)$ and $F(\forall\kappa. \capp{P})(\Dist(\inr)(\mu_2), \nu_2)$.
\end{enumerate}
To prove this, recall the following facts from~\cite{CubicalCloTT}: $\forall\kappa. (-)$ commutes with sums, products, $\Sigma$ and
propositional truncation. In particular, it commutes with existential quantification over clock irrelevant types. Recall
also that our assumption of a clock constant makes all propositions clock irrelevant. In particular, in
$\forall\kappa . F(\capp P)$ we can commute $\forall\kappa$ first over the case split, then
remove it from the case of (1) because this is a proposition not referring to $P$, and so $\kappa$ will not appear in it.
In (2), we can commute $\forall\kappa$ over first the existential quantification over the clock irrelevant type $\Dist(\LiftDist{A})$,
then over the product and use that $\mu \peq \Dist(\inr)(\mu')$ is a proposition, arriving at condition (2) of $F(\forall\kappa . P)$.
The case of (3) is similar.

The conclusion from this is that the final coalgebra mentioned in Lemma~\ref{lem:CRelLift:coinductive} can be encoded as
$S(\mu, \nu) \defeq \forall\kappa . S^\kappa (\mu, \nu)$, where $S^\kappa$ is the guarded recursive relation satisfying
$S^\kappa(\mu, \nu)$ iff either
\begin{enumerate}
\item there exists $\mu' : \Dist{A}$ such that $\mu \peq \Dist(\inl)(\mu')$, 
 and there is a $\nu' : \Dist{B}$ such that $\nu \limleadsto \nu'$, and there is a $\rho: \CoupType\Rel{\mu'}{\nu'}$.
\item  there exists $\mu' : \Dist(\LiftDist{A})$ such that $\mu \peq \Dist(\inr)(\mu')$ and $\Later (S^\kappa(\run(\mu), \nu))$.
\item there exist $\mu_1 : \Dist{A}, \mu_2 : \Dist(\LiftDist{A})$, and $p : \II$ such that
$\mu \peq \DistChoice{p}{\Dist(\inl)(\mu_1)}{\Dist(\inr)(\mu_2)}$, and there exist $\nu_1, \nu_2 : \LiftDist{B}$ such that $\nu \limleadsto \CChoice{p}{\nu_1}{\nu_2}$, 
and $S^\kappa(\Dist(\inl)(\mu_1), \nu_1)$ and $S^\kappa(\Dist(\inr)(\mu_2), \nu_2)$.
\end{enumerate}

It remains to show that $S^\kappa(\mu, \nu)$ is equivalent to $\GRelLift{\Rel}(\capp\mu, \nu)$, which we prove by guarded recursion.
First recall that $(\capp\mu) \GRelLift{\Rel} \nu$ iff either
\begin{enumerate}
\item there exists $\mu' : \Dist{A}$ such that $\capp\mu \peq \Dist(\inl)(\mu')$, 
there is a $\nu' : \Dist{B}$ such that $\nu \limleadsto \nu'$, and there is a $\rho: \CoupType\Rel{\mu'}{\nu'}$. 
\item there exists $\mu' : \Dist(\Later\GLiftDist{A})$ such that $\capp\mu \peq \Dist(\inr)(\mu')$ and $\Later{(\alpha : \kappa)}\  (((\Gzeta (\capp\mu))[\alpha]) \GRelLift{\Rel} \nu)$.
\item there exist $\mu_1 : \Dist{A}, \mu_2 : \Dist(\Later{}\GLiftDist{A})$, $p \in \II$ such that $\capp\mu \peq \DistChoice{p}{\Dist(\inl)(\mu_1)}{\Dist(\inr)(\mu_2)}$, 
there exist $\nu_1, \nu_2 : \LiftDist{B}$ such that $\nu \limleadsto \CChoice{p}{\nu_1}{\nu_2}$, and $\Dist(\inl)(\mu_1) \GRelLift{\Rel} \nu_1$ and 
$\Dist(\inr)(\mu_2) \GRelLift{\Rel} \nu_2$ 
\end{enumerate}
%
We first show that $S^\kappa(\mu, \nu)$ implies $\GRelLift{\Rel}(\capp\mu, \nu)$. Cases (1) in the two definitions are equal. In the case of (2) of $S^\kappa(\mu, \nu)$, to prove (2) of
$\GRelLift{\Rel}(\capp\mu, \nu)$, consider $\Dist(\nextop\circ\mathsf{ev}_\kappa)(\mu'): \Dist(\Later\GLiftDist{A})$, where $\mathsf{ev}_\kappa : \LiftDist{A} \to \GLiftDist{A}$ 
is evaluation at $\kappa$.
Since $\mu \peq \Dist(\inr)(\mu')$, it is easy to see that $\capp\mu \peq \Dist(\inr)(\Dist(\nextop\circ\mathsf{ev}_\kappa)(\mu'))$, 
and clearly $\Gzeta(\capp\mu) = \nextop(\capp{\run(\mu)})$ so that $\Later (S^\kappa(\run(\mu), \nu))$ implies $\Later(\tickA : \kappa) (S^\kappa(\tapp{\Gzeta(\capp\mu)}, \nu))$,
which implies $\Later(\tickA : \kappa) (\tapp{\Gzeta(\capp\mu)}\GRelLift{\Rel} \nu)$ by guarded recursion as desired. The case of (3) is similar. 
%

In the other direction we just show that (2) of $(\capp\mu) \GRelLift{\Rel} \nu$ implies (2) of $S^\kappa(\mu, \nu)$. So suppose the former, and apply
the equivalence of Theorem~\ref{thm:dist:sum:equiv} to $\mu$. Since $\capp \mu$ is in the second of the three possible cases, so must $\mu$ be. In other words,
there must exist a $\mu'' : \Dist(\LiftDist{A})$ such that $\mu \peq \Dist(\inr)(\mu'')$ and $\Dist(\nextop\circ \mathsf{ev}_\kappa)(\mu'') = \mu'$.
Then, as before,  $\Gzeta(\capp\mu) = \nextop(\capp{\run(\mu)})$ and so the assumption $\Later(\tickA : \kappa) (\tapp{\Gzeta(\capp\mu)}\GRelLift{\Rel} \nu)$
is equivalent to $\Later(\tickA : \kappa) (\capp{\run(\mu)}\GRelLift{\Rel} \nu)$, which by guarded recursion implies $\Later (S^\kappa(\run(\mu), \nu))$
as desired. 
%
\end{proof}

\begin{proof}[Proof of Lemma~\ref{lem:eq1-probterm}]
  Recall that, by definition of $\leqlim$, we have to prove that for all $n$ and $\varepsilon > 0$ there exists $m$ such that $\probterm{n}{\mu} \leq \probterm{m}{\nu} + \varepsilon$.
  We first show that the lemma holds whenever $\mu : \Dist{1}$. Here, there exists $\nu' : \Dist{1}$ such that $\nu \limleadsto \nu'$ and there is a $\rho: \CoupType{\eqrel}{\mu}{\nu'}$. Then for any $p \in \II$ there exist $\nu'' : \LiftDist{1}$ such that $\nu \leadsto \CChoice{p}{\nu'}{\nu''}$. We now have:
        \begin{itemize}
	  \item $\probterm{n}{\mu} = \probterm{n}{\nu'} = 1$
	  \item By Lemma~\ref{lemma:limleadsto-eqlim}, $\probtermseq{\nu} \eqlim \probtermseq{\nu'}$, and there exists $m$ such that $\probterm{n}{\nu'} \leq \probterm{m}{\nu} + \varepsilon$.
        \end{itemize}

  Otherwise, we do induction on $n$. For $n = 0$, $\varepsilon > 0$, we have two cases for $\mu \CRelLift{\eqrel} \nu$:
  \begin{itemize}
  
    \item If $\mu : \Dist(\LiftDist{1})$, then $\probterm{0}{\mu} = 0$ and the statement is trivially true.
    \item If there exist $\mu_1 : \Dist{1}, \mu_2 : \Dist(\LiftDist{1}), p \in \II$ such that $\mu = \DistChoice{p}{\mu_1}{\mu_2}$, and there exist $\nu_1, \nu_2 : \LiftDist{1}$ such that $\nu \limleadsto \CChoice{p}{\nu_1}{\nu_2}$, and $\mu_1\CRelLift{\eqrel}\nu_1$, $\mu_2\CRelLift{\eqrel}\nu_2$:

        \begin{itemize}
          \item $\probterm{0}{\mu} = p$
	  \item By the previous case, $\probtermseq{\mu_1} \leq \probtermseq{\nu_1}$, and furthermore, $\probterm{0}{\mu_1} = 1$.
		  Therefore, there exists $m'$ such that $1 \leq \probterm{m'}{\nu_1} + \varepsilon/2$, and thus
			\begin{align*}
				p &\leq p\cdot \probterm{m'}{\nu_1} + p\cdot \varepsilon/2 \\
				  &\leq p\cdot \probterm{m'}{\nu_1} + p\cdot \varepsilon/2 + (1-p) \cdot \probterm{m'}{\nu_2} + (1-p)\cdot\varepsilon/2 \\
				  &\peq \probterm{m'}{\CChoice{p}{\nu_1}{\nu_2}} + \varepsilon/2
			\end{align*}
			so $p \leq \probterm{m'}{\CChoice{p}{\nu_1}{\nu_2}} + \varepsilon/2$.
			By Lemma~\ref{lemma:limleadsto-eqlim}, $\probtermseq{\CChoice{p}{\nu_1}{\nu_2}} \eqlim \probtermseq{\nu}$, so there exists $m$ such that
			$\probterm{m'}{\CChoice{p}{\nu_1}{\nu_2}} \leq \probterm{m}{\nu} + \varepsilon/2$,
			and then we can conclude $\probterm{0}{\mu} \leq \probterm{m}{\nu} + \varepsilon$.
	\end{itemize}
  \end{itemize}
  For $n = n' + 1$, $\varepsilon > 0$, we have again two cases for $\mu \CRelLift{\eqrel} \nu$.
  \begin{itemize}
    \item If $\mu : \Dist(\LiftDist{1})$ and $\run (\mu) \CRelLift{\eqrel} \nu$ then $\probterm{n}{\mu} = \probterm{n'}{\run(\mu)}$. By the induction hypothesis for $n'$, we know that then there is an $m$ such that $\probterm{n'}{\run(\mu)} \leq \probterm{m}{\nu} + \varepsilon$, and hence also $\probterm{n}{\mu} \leq \probterm{m}{\nu} + \varepsilon$.
    \item If there exist $\mu_1 : \Dist{1}, \mu_2 : \Dist(\LiftDist{1}), p \in \II$ such that $\mu = \DistChoice{p}{\mu_1}{\mu_2}$, and there exist $\nu_1, \nu_2 : \LiftDist{1}$ such that $\nu \limleadsto \CChoice{p}{\nu_1}{\nu_2}$, and $\mu_1\CRelLift{\eqrel}\nu_1$, $\mu_2\CRelLift{\eqrel}\nu_2$:
        Then:
        \begin{itemize}
		\item $\probterm{n}{\mu} \peq p\cdot \probterm{n'}{\mu_1} + (1-p)\cdot\probterm{n'}{\run(\mu_2)}$
		\item By definition of the lifting, $\run(\mu_2)\CRelLift{\eqrel}\nu_2$. By the inductive hypothesis, there exist $m_1, m_2$ such that
			$\probterm{n'}{\mu_1} = 1 \leq \probterm{m_1}{\nu_1} + \varepsilon/2$ and
			$\probterm{n'}{\run(\mu_2)} \leq \probterm{m_2}{\nu_2} + \varepsilon/2$.
		\item Let $m'$ be the largest of $m_1,m_2$. By Lemma~\ref{lemma:limleadsto-eqlim} there exists $m$ such that
			$\probterm{m'}{\DistChoice{p}{\nu_1}{\nu_2}} \leq \probterm{m}{\nu} + \varepsilon/2$.
	\end{itemize}
	Putting all of the above together, we now have:
	\begin{align*}
		\probterm{n}{\mu} &\peq p\cdot \probterm{n'}{\mu_1} + (1-p)\cdot\probterm{n'}{\run(\mu_2)} \\
				  &\leq p\cdot (\probterm{m_1}{\nu_1} + \varepsilon/2) + (1-p) \cdot (\probterm{m_2}{\nu_2} + \varepsilon/2) \\
				  &\leq p\cdot \probterm{m'}{\nu_1} + (1-p) \cdot \probterm{m'}{\nu_2} + \varepsilon/2 \\
				  &\leq p\cdot \probterm{m}{\nu} + \varepsilon/2 + \varepsilon/2 \\
				  &\peq p\cdot \probterm{m}{\nu} + \varepsilon
	\end{align*}
  which is what we needed to show. \qedhere
  \end{itemize}
\end{proof}

\begin{proof}[Proof of Lemma~\ref{lem:prereqs:leadsto-liftrel}]$\;$\\
 For left to right, suppose that $\mu \GRelLift{\Rel} \nu$. We need to show that also $\mu \GRelLift{\Rel} \nu'$. We consider the three possibilities of Definition~\ref{def:liftRel}.
 \begin{itemize}
   \item If $\mu : \Dist{A}$, then we need to show that there is a $\nu'' : \Dist{B}$ such that $\nu' \limleadsto \nu''$, and a $\rho: \CoupType\Rel{\mu}{\nu''}$. From our assumption $\mu \GRelLift{\Rel} \nu$ we get a $\nu'' : \Dist{B}$ such that $\nu \limleadsto \nu''$, and a $\rho: \CoupType\Rel{\mu}{\nu''}$.
       
       We just need to show that also $\nu' \limleadsto \nu''$. For any $p \in \II$, we know from $\nu \limleadsto \nu''$, that there is a $\nu''' : \LiftDist{B}$ such that $\nu \leadsto \CChoice{p}{\nu''}{\nu'''}$. 
 
       By Lemma~\ref{lemma:leadsto-run}.\ref{lemma:leadsto-run3}, there is an $n$ such that $\nu' \leadsto \run^n \nu$. By Lemma~\ref{lemma:leadsto-run}.\ref{lemma:leadsto-run2} we know:
        \begin{align*}
          \run^n \nu & \leadsto \run^n(\CChoice{p}{\nu''}{\nu'''}) \\
           & \peq \CChoice{p}{\run^n \nu''}{\run^n \nu'''} \\
           & \peq \CChoice{p}{\nu''}{\run^n \nu'''},
        \end{align*} 
        where the last equality holds because $\nu'' : \Dist{B}$.
        Hence by transitivity of $\leadsto$, also:
        \[\nu' \leadsto \CChoice{p}{\nu''}{\run^n \nu'''}, \]
        proving that indeed $\nu' \limleadsto \nu''$.
   \item If $\mu : \Dist{(\Later{}\GLiftDist{A})}$, then we need to show that $\Later{(\alpha : \kappa)}\  (((\Gzeta (\mu))[\alpha]) \GRelLift{\Rel} \nu')$. From our assumption $\mu \GRelLift{\Rel} \nu$ we know that $\Later{(\alpha : \kappa)}\  (((\Gzeta (\mu))[\alpha]) \GRelLift{\Rel} \nu)$, and so by guarded recursion, also $\Later{(\alpha : \kappa)}\  (((\Gzeta (\mu))[\alpha]) \GRelLift{\Rel} \nu')$. 
   \item If there exist $\mu_1 : \Dist{A}, \mu_2 : \Dist(\Later{}\GLiftDist{A})$, and $p \in \II$ such that $\mu = \DistChoice{p}{\mu_1}{\mu_2}$, then we need to show that there exist $\nu'_1, \nu'_2 : \LiftDist{B}$ such that $\nu' \limleadsto \CChoice{p}{\nu'_1}{\nu'_2}$, and $\mu_1 \GRelLift{\Rel} \nu'_1$ and $\mu_2 \GRelLift{\Rel} \nu'_2$.
       
       We know from $\mu \GRelLift{\Rel} \nu$ that there exist $\nu_1, \nu_2 : \LiftDist{B}$ such that $\nu \limleadsto \CChoice{p}{\nu_1}{\nu_2}$, and $\mu_1 \GRelLift{\Rel} \nu_1$ and $\mu_2 \GRelLift{\Rel} \nu_2$. 
       
       By Lemma~\ref{lemma:leadsto-run}.\ref{lemma:leadsto-run3}, there is an $n$ such that $\nu' \leadsto \run^n \nu$. Then by Lemma~\ref{lemma:leadsto-run}.\ref{lemma:leadsto-run2}:
        \begin{align*}
          \run^n \nu & \leadsto \run^n(\CChoice{p}{\nu_1}{\nu_2}) \\
           & \peq \CChoice{p}{\run^n \nu_1}{\run^n \nu_2}
        \end{align*}
        
        Then by Lemma~\ref{lemma:limleadsto-leadsto-interact}.\ref{limleadsto-leadsto-interact1} also $\run^n \nu \limleadsto \CChoice{p}{\run^n \nu_1}{\run^n \nu_2}$, and finally by Lemma~\ref{lemma:limleadsto-leadsto-interact}.\ref{limleadsto-leadsto-interact2} $\nu' \limleadsto \CChoice{p}{\run^n \nu_1}{\run^n \nu_2}$.
        
        What is left to show is that $\mu_1 \GRelLift{\Rel} \run^n \nu_1$ and $\mu_2 \GRelLift{\Rel} \run^n \nu_2$. Since by Lemma~\ref{lemma:leadsto-run}.\ref{lemma:leadsto-run1} $\nu_1 \leadsto \run^n \nu_1$ and $\nu_2 \leadsto \run^n \nu_2$, and $\mu_1 : \Dist{A}$ and $\mu_2 : \Dist(\Later{}\GLiftDist{A})$, these claims follow from the first two cases of this lemma. 
 \end{itemize}

 For right to left, suppose that $\mu \GRelLift{\Rel} \nu'$. We need to show that also $\mu \GRelLift{\Rel} \nu$. We consider the three possibilities of Definition~\ref{def:liftRel}.
 \begin{itemize}
   \item If $\mu : \Dist{A}$, then we need to show that there is a $\nu''$ such that $\nu \limleadsto \nu''$, and a $\rho: \CoupType\Rel{\mu}{\nu''}$. From our assumption $\mu \GRelLift{\Rel} \nu'$ we get a $\nu''$ such that $\nu' \limleadsto \nu''$, and a $\rho: \CoupType\Rel{\mu}{\nu''}$. By Lemma~\ref{lemma:limleadsto-leadsto-interact}, $\nu \leadsto \nu'$ implies that also $\nu \limleadsto \nu''$, and so indeed $\mu \GRelLift{\Rel} \nu$. 
   \item If $\mu : \Dist{(\Later{}\GLiftDist{A})}$, then we need to show that $\Later{(\alpha : \kappa)}\  (((\Gzeta (\mu))[\alpha]) \GRelLift{\Rel} \nu)$. From our assumption $\mu \GRelLift{\Rel} \nu'$ we know that $\Later{(\alpha : \kappa)}\  (((\Gzeta (\mu))[\alpha]) \GRelLift{\Rel} \nu')$, and so by guarded recursion, also $\Later{(\alpha : \kappa)}\  (((\Gzeta (\mu))[\alpha]) \GRelLift{\Rel} \nu)$.  
   \item If there exist $\mu_1 : \Dist{A}, \mu_2 : \Dist(\Later{}\GLiftDist{A})$, and $p \in \II$ such that $\mu = \DistChoice{p}{\mu_1}{\mu_2}$, then we need to show that there exist $\nu_1, \nu_2 : \LiftDist{B}$ such that $\nu \limleadsto \CChoice{p}{\nu_1}{\nu_2}$, and $\mu_1 \GRelLift{\Rel} \nu_1$ and $\mu_2 \GRelLift{\Rel} \nu_2$.
       
       From our assumption $\mu \GRelLift{\Rel} \nu'$ we know that there exist $\nu_1, \nu_2 : \LiftDist{B}$ such that $\nu' \limleadsto \CChoice{p}{\nu_1}{\nu_2}$, and $\mu_1 \GRelLift{\Rel} \nu_1$ and $\mu_2 \GRelLift{\Rel} \nu_2$. By Lemma~\ref{lemma:limleadsto-leadsto-interact}, $\nu \leadsto \nu'$ implies that also $\nu \limleadsto \CChoice{p}{\nu_1}{\nu_2}$, and so indeed $\mu \GRelLift{\Rel} \nu$. \qedhere
 \end{itemize}
\end{proof}

\begin{proof}[Proof of Lemma~\ref{lem:prereqs:limleadsto-liftrel}]$\;$\\
  Given $\mu : \GLiftDist{A}, \nu : \LiftDist{B}$, suppose that $\nu \limleadsto \nu'$ and $\mu \GRelLift{\Rel} \nu'$.
  We consider three cases, one for each part of Definition~\ref{def:liftRel}:
  \begin{itemize}
    \item If $\mu : \Dist{A}$, then we need to show that there is a $\nu''$ such that $\nu \limleadsto \nu''$, and a $\rho: \CoupType\Rel{\mu}{\nu''}$. From our assumption $\mu \GRelLift{\Rel} \nu'$ we get a $\nu''$ such that $\nu' \limleadsto \nu''$, and a $\rho: \CoupType\Rel{\mu}{\nu''}$. By transitivity of $\limleadsto$ (Lemma~\ref{lemma:limleadsto-properties}), also $\nu \limleadsto \nu''$, so this case is immediate from the assumptions.   
    \item If $\mu : \Dist{(\Later{}\GLiftDist{A})}$, then we need to show that $\Later{(\alpha : \kappa)}\  (((\Gzeta (\mu))[\alpha]) \GRelLift{\Rel} \nu)$. From our assumption $\mu \GRelLift{\Rel} \nu'$ we know that $\Later{(\alpha : \kappa)}\  (((\Gzeta (\mu))[\alpha]) \GRelLift{\Rel} \nu')$, and so by guarded recursion, also $\Later{(\alpha : \kappa)}\  (((\Gzeta (\mu))[\alpha]) \GRelLift{\Rel} \nu)$.
    \item If there exist $\mu_1 : \Dist{A}, \mu_2 : \Dist(\Later{}\GLiftDist{A})$, and $p \in \II$ such that $\mu = \DistChoice{p}{\mu_1}{\mu_2}$, then we need to show that there exist $\nu_1, \nu_2 : \LiftDist{B}$ such that $\nu \limleadsto \CChoice{p}{\nu_1}{\nu_2}$, and $\mu_1 \GRelLift{\Rel} \nu_1$ and $\mu_2 \GRelLift{\Rel} \nu_2$.
        
        From our assumption $\mu \GRelLift{\Rel} \nu'$ we know that there exist $\nu_1, \nu_2 : \LiftDist{B}$ such that $\nu' \limleadsto \CChoice{p}{\nu_1}{\nu_2}$, and $\mu_1 \GRelLift{\Rel} \nu_1$ and $\mu_2 \GRelLift{\Rel} \nu_2$. By transitivity of $\limleadsto$ (Lemma~\ref{lemma:limleadsto-properties}), also $\nu \limleadsto \nu''$, so this case is again immediate from the assumptions.   \qedhere
  \end{itemize} 
\end{proof}

\begin{proof}[Proof of Lemma~\ref{lem:prereqs:step-choice}]$\;$\\
For the first statement, notice that
\[
\GChoice{p}{(\GStep \mu_1)}{(\GStep \mu_2)} \peq \Dist(\inr)(\DistChoice{p}{\Dirac{\mu_1}}{\Dirac{\mu_2}}),\] 
and 
\begin{align*}
  \GStep & (\tabs{\tickA}{\kappa} \GChoice{p}{(\tapp[\tickA]{\mu_1})}{(\tapp[\tickA]{\mu_2})}) \\
   & \peq \Dist(\inr)(\Dirac{\tabs{\tickA}{\kappa} \GChoice{p}{(\tapp[\tickA]{\mu_1})}{(\tapp[\tickA]{\mu_2})}}),
\end{align*}
 where both right hands of these equations are of form $\Dist(\inr)(\mu')$ with $\mu' : \Dist(\Later\GLiftDist{A})$. Therefore, by Definition \ref{def:liftRel}(\ref{item:delay}) we have:
\begin{align*}
   & \GChoice{p}{(\GStep \mu_1)}{(\GStep \mu_2)} \GRelLift{\Rel} \nu \\
  \Leftrightarrow \; & \latbind\alpha\kappa{(((\Gzeta (\GChoice{p}{(\Dirac{\mu_1})}{(\Dirac{\mu_2})}))[\alpha]) \GRelLift{\Rel} \nu)} \\
  \peq \; & \latbind\alpha\kappa {((\GChoice{p}{(\tapp[\tickA]{\mu_1})}{(\tapp[\tickA]{\mu_2})}) \GRelLift{\Rel} \nu)} \\
  \peq \; & \latbind\alpha\kappa{(((\Gzeta (\Dirac{\tabs{\tickB}{\kappa} \GChoice{p}{(\tapp[\tickB]{\mu_1})}{(\tapp[\tickB]{\mu_2})}}  ))[\alpha]) \GRelLift{\Rel} \nu)} \\
  \Leftrightarrow \; & \GStep(\tabs{\tickB}{\kappa} \GChoice{p}{(\tapp[\tickB]{\mu_1})}{(\tapp[\tickB]{\mu_2})}) \GRelLift{\Rel} \nu.
\end{align*}

The second statement follows from Lemma~\ref{lem:prereqs:leadsto-liftrel}
since both
\begin{align*}
\CChoice p {(\CStep(\nu_1))}{(\CStep(\nu_2))} & \leadsto \CChoice p{\nu_1}{\nu_2} \\
\CStep(\CChoice p {\nu_1}{\nu_2}) & \leadsto \CChoice p{\nu_1}{\nu_2}. \qedhere
\end{align*}
\end{proof}

\begin{proof} [Proof of Lemma~\ref{lem:prereqs:choice-lemma}]$\;$\\
  We analyse the 3*3 different possibilities for $\mu_1 \GRelLift{\Rel} \nu_1$ and $\mu_2 \GRelLift{\Rel} \nu_2$.
  \begin{enumerate}
   \item If both $\mu_1, \mu_2 : \Dist{A}$, then from $\mu_1 \GRelLift{\Rel} \nu_1$ and $\mu_2 \GRelLift{\Rel} \nu_2$ we know that there are $\nu_1', \nu_2' : \Dist{B}$ such that $\nu_1 \limleadsto \nu_1'$ and $\nu_2 \limleadsto \nu_2'$, and that there are $\rho_1: \CoupType\Rel{\mu}{\nu_1'}$ and $\rho_2: \CoupType\Rel{\mu}{\nu_2'}$.
       
       By Lemma~\ref{lemma:limleadsto-properties} we have $\CChoice{p}{\nu_1}{\nu_2} \limleadsto \CChoice{p}{\nu_1'}{\nu_2'}$. Also: $\DistChoice{p}{\rho_1}{\rho_2} : \CoupType\Rel{\GChoice{p}{\mu_1}{\mu_2}}{\CChoice{p}{\nu_1'}{\nu_2'}}$, and so $(\GChoice{p}{\mu_1}{\mu_2}) \GRelLift{\Rel} (\CChoice{p}{\nu_1}{\nu_2})$ via Definition~\ref{def:liftRel}.\ref{item:val}.
   
   \item If $\mu_1 : \Dist{A}$ and $\mu_2 : \Dist{(\Later{}\GLiftDist{A})}$, then $(\GChoice{p}{\mu_1}{\mu_2}) \GRelLift{\Rel} (\CChoice{p}{\nu_1}{\nu_2})$ follows immediately from the given assumptions, via Definition~\ref{def:liftRel}.\ref{item:both}.
   
   \item If $\mu_2 : \Dist{A}$ and $\mu_1 : \Dist{(\Later{}\GLiftDist{A})}$, then again $(\GChoice{p}{\mu_1}{\mu_2}) \GRelLift{\Rel} (\CChoice{p}{\nu_1}{\nu_2})$ follows immediately from the given assumptions, via Definition~\ref{def:liftRel}.\ref{item:both}, and using that $\GChoice{p}{\mu_1}{\mu_2} \peq (\GChoice{1-p}{\mu_2}{\mu_1})$.
   
   \item If $\mu_1 : \Dist{A}$ and for $\mu_2$ there exist $\mu_{21} : \Dist{A}, \mu_{22} : \Dist(\Later{}\GLiftDist{A})$ such that $\mu_2 \peq \DistChoice{p'}{\mu_{21}}{\mu_{22}}$, then:
         \begin{align*}
           \GChoice{p}{\mu_1}{\mu_2} & \peq \GChoice{p}{\mu_1}{(\GChoice{p'}{\mu_{21}}{\mu_{22}})}  \\
            & \peq \GChoice{p + (1-p)p'}{(\GChoice{\frac{p}{p + (1-p)p'}}{\mu_1}{\mu_{21}})}{\mu_{22}},
         \end{align*}
         where $\GChoice{\frac{p}{p + (1-p)p'}}{\mu_1}{\mu_{21}} : \Dist{A}$ and $\mu_{22} : \Dist(\Later{}\GLiftDist{A})$. 
         
         So by Definition~\ref{def:liftRel}.\ref{item:both}, we need to show that there exist $\nu', \nu'' : \LiftDist{B}$ such that $\CChoice{p}{\nu_1}{\nu_2} \limleadsto \CChoice{p + (1-p)p'}{\nu'}{\nu''}$, and $(\GChoice{\frac{p}{p + (1-p)p'}}{\mu_1}{\mu_{21}}) \GRelLift{\Rel} \nu'$ and $\mu_{22} \GRelLift{\Rel} \nu''$.
       
         We know from $\mu_2 \GRelLift{\Rel} \nu_2$ that there exist $\nu_{21}, \nu_{22} : \LiftDist{B}$ such that $\nu_2 \limleadsto \CChoice{p'}{\nu_{21}}{\nu_{22}}$, and $\mu_{21} \GRelLift{\Rel} \nu_{21}$ and $\mu_{22} \GRelLift{\Rel} \nu_{22}$.
         We also know that $\mu_1 \GRelLift{\Rel} \nu_1$ by assumption.
         We then know, using the reflexivity and sum properties of Lemma~\ref{lemma:limleadsto-properties}, that:
         \begin{align*}
         \CChoice{p}{\nu_1}{\nu_2} & \limleadsto \CChoice{p}{\nu_1}{(\CChoice{p'}{\nu_{21}}{\nu_{22}})} \\
                                   & \peq \CChoice{p + (1-p)p'}{(\CChoice{\frac{p}{p + (1-p)p'}}{\nu_1}{\nu_{21}})}{\nu_{22}}.
         \end{align*}
         So we choose:
         \begin{align*}
           \nu' & \peq \CChoice{\frac{p}{p + (1-p)p'}}{\nu_1}{\nu_{21}} \\
           \nu'' & \peq \nu_{22}.
         \end{align*}
         Then from the first item of this lemma we get that $(\GChoice{\frac{p}{p + (1-p)p'}}{\mu_1}{\mu_{21}}) \GRelLift{\Rel} \nu'$, and finally it follows directly from our assumptions that $\mu_{22} \GRelLift{\Rel} \nu''$.
                
   \item If $\mu_2 : \Dist{A}$ and for $\mu_1$ there exist $\mu_{11} : \Dist{A}, \mu_{12} : \Dist(\Later{}\GLiftDist{A})$ such that $\mu_1 \peq \DistChoice{p'}{\mu_{11}}{\mu_{12}}$, then:
       \begin{align*}
           \GChoice{p}{\mu_1}{\mu_2} & \peq \GChoice{p}{(\GChoice{p'}{\mu_{11}}{\mu_{12}})}{\mu_2}  \\
            & \peq \GChoice{1-(1-p')p}{(\GChoice{\frac{pp'}{1 - (1-p')p}}{\mu_{11}}{\mu_2})}{\mu_{12}},
         \end{align*}
         where $\GChoice{\frac{pp'}{1 - (1-p')p}}{\mu_{11}}{\mu_2} : \Dist{A}$ and $\mu_{12} : \Dist(\Later{}\GLiftDist{A})$.
         
         We again need to satisfy Definition~\ref{def:liftRel}.\ref{item:both}, and via similar arguments as the previous case we can choose
         \begin{align*}
           \nu' & \peq \CChoice{\frac{pp'}{1 - (1-p')p}}{\nu_{11}}{\nu_2} \\
           \nu'' & \peq \nu_{12}
         \end{align*}
         to do so.
   
   \item If both $\mu_1, \mu_2 : \Dist{(\Later{}\GLiftDist{A})}$ then we know that 
   \[\Later{(\alpha : \kappa)}\  (((\Gzeta (\mu_1))[\alpha]) \GRelLift{\Rel} \nu_1)\]
    and 
    \[ \Later{(\alpha : \kappa)}\  (((\Gzeta (\mu_2))[\alpha]) \GRelLift{\Rel} \nu_2).\]
       Then by guarded recursion we know:
       \[
       \Later{(\alpha : \kappa)}\  (\GChoice{p}{((\Gzeta (\mu_1))[\alpha])}{((\Gzeta (\mu_2))[\alpha])} \GRelLift{\Rel} (\CChoice{p}{\nu_1}{\nu_2})), 
       \]
       which is the same as:
       \[
       \Later{(\alpha : \kappa)}\  (((\Gzeta (\GChoice{p}{\mu_1}{\mu_2}))[\alpha]) \GRelLift{\Rel} (\CChoice{p}{\nu_1}{\nu_2})),
       \]
       which shows that $(\GChoice{p}{\mu_1}{\mu_2}) \GRelLift{\Rel} (\CChoice{p}{\nu_1}{\nu_2})$ via Definition~\ref{def:liftRel}.\ref{item:delay}.
       
   \item If $\mu_2 : \Dist{(\Later{}\GLiftDist{A})}$ and for $\mu_1$ there exist $\mu_{11} : \Dist{A}, \mu_{12} : \Dist(\Later{}\GLiftDist{A})$ such that $\mu_1 \peq \DistChoice{p'}{\mu_{11}}{\mu_{12}}$, then:
       \begin{align*}
           \GChoice{p}{\mu_1}{\mu_2} & \peq \GChoice{p}{(\GChoice{p'}{\mu_{11}}{\mu_{12}})}{\mu_2}  \\
            & \peq \GChoice{pp'}{\mu_{11}}{(\GChoice{\frac{(1-p')p}{1 - pp'}}{\mu_{12}}{\mu_2})},
         \end{align*}
         where $\mu_{11} : \Dist{A}$ and $\GChoice{\frac{(1-p')p}{1 - pp'}}{\mu_{12}}{\mu_2} : \Dist(\Later{}\GLiftDist{A})$.
         
         We know from $\mu_1 \GRelLift{\Rel} \nu_1$ that there exist $\nu_{11}, \nu_{12} : \LiftDist{B}$ such that $\nu_1 \limleadsto \CChoice{p_1}{\nu_{11}}{\nu_{12}}$, and $\mu_{11} \GRelLift{\Rel} \nu_{11}$ and $\mu_{12} \GRelLift{\Rel} \nu_{12}$.
         
         If we pick
         \begin{align*}
           \nu' & \peq \nu_{11}\\
           \nu'' & \peq \CChoice{\frac{(1-p')p}{1 - pp'}}{\nu_{12}}{\nu_2},
         \end{align*}
         then by the reflexivity and sum properties of Lemma~\ref{lemma:limleadsto-properties}
         \begin{align*}
         \CChoice{p}{\nu_1}{\nu_2} & \limleadsto \CChoice{p}{(\CChoice{p'}{\nu_{11}}{\nu_{12}})}{\nu_2} \\
          & \peq \CChoice{pp'}{\nu_{11}}{(\CChoice{\frac{(1-p')p}{1 - pp'}}{\nu_{12}}{\nu_2})} \\
          & \peq \CChoice{pp'}{\nu'}{\nu''}
         \end{align*}
         
         In addition, $\mu_{11} \GRelLift{\Rel} \nu_{11}$ immediately gives $\mu_{11} \GRelLift{\Rel} \nu'$, and from the sixth item of this lemma we know that $\mu_{12} \GRelLift{\Rel} \nu_{12}$ and $\mu_2 \GRelLift{\Rel} \nu_2$ imply $(\GChoice{\frac{(1-p')p}{1 - pp'}}{\mu_{12}}{\mu_2}) \GRelLift{\Rel} (\CChoice{\frac{(1-p')p}{1 - pp'}}{\nu_{12}}{\nu_2}) \peq \nu''$.
         
         Hence we have $(\GChoice{p}{\mu_1}{\mu_2}) \GRelLift{\Rel} (\CChoice{p}{\nu_1}{\nu_2})$ via Definition~\ref{def:liftRel}.\ref{item:both}.
   
   \item If $\mu_1 : \Dist{(\Later{}\GLiftDist{A})}$ and for $\mu_2$ there exist $\mu_{21} : \Dist{A}, \mu_{22} : \Dist(\Later{}\GLiftDist{A})$ such that $\mu_2 \peq \DistChoice{p'}{\mu_{21}}{\mu_{22}}$, then:
         \begin{align*}
           \GChoice{p}{\mu_1}{\mu_2} & \peq \GChoice{p}{\mu_1}{(\GChoice{p'}{\mu_{21}}{\mu_{22}})} \\
            & \peq \GChoice{(1-p)p'}{\mu_{21}}{(\GChoice{\frac{p}{1-(1-p)p'}}{\mu_1}{\mu_{22}})},
         \end{align*}
         where $\mu_{21} : \Dist{A}$ and $\GChoice{\frac{p}{1-(1-p)p'}}{\mu_1}{\mu_{22}} : \Dist(\Later{}\GLiftDist{A})$.
         
         We again need to satisfy Definition~\ref{def:liftRel}.\ref{item:both}, and via similar arguments as the previous case we can choose
         \begin{align*}
           \nu' & \peq \nu_{21} \\
           \nu'' & \peq \CChoice{\frac{p}{1-(1-p)p'}}{\nu_1}{\nu_{12}} 
         \end{align*}
         to do so.
          
   \item Finally, if for $\mu_1$ there exist $\mu_{11} : \Dist{A}, \mu_{12} : \Dist(\Later{}\GLiftDist{A})$ such that $\mu_1 \peq \DistChoice{p_1}{\mu_{11}}{\mu_{12}}$, and for $\mu_2$ there exist $\mu_{21} : \Dist{A}, \mu_{22} : \Dist(\Later{}\GLiftDist{A})$ such that $\mu_2 \peq \DistChoice{p_2}{\mu_{21}}{\mu_{22}}$, then:
       \begin{align*}
           \GChoice{p}{\mu_1}{\mu_2} & \peq \GChoice{p}{(\GChoice{p_1}{\mu_{11}}{\mu_{12}})}{(\GChoice{p_2}{\mu_{21}}{\mu_{22}})}  \\
            & \peq \GChoice{pp_1+(1-p)p_2}{(\GChoice{\frac{pp_1}{pp_1+(1-p)p_2}}{\mu_{11}}{\mu_{21}})}{(\GChoice{\frac{(1-p_1)p}{1 - (pp_1+(1-p)p_2)}}{\mu_{12}}{\mu_{22}})},
       \end{align*}
       where $\GChoice{\frac{pp_1}{pp_1+(1-p)p_2}}{\mu_{11}}{\mu_{21}} : \Dist{A}$ and $\GChoice{\frac{(1-p_1)p}{1 - (pp_1+(1-p)p_2)}}{\mu_{12}}{\mu_{22}} : \Dist(\Later{}\GLiftDist{A})$.
       
        We know from $\mu_1 \GRelLift{\Rel} \nu_1$ that there exist $\nu_{11}, \nu_{12} : \LiftDist{B}$ such that $\nu_1 \limleadsto \CChoice{p_1}{\nu_{11}}{\nu_{12}}$, and $\mu_{11} \GRelLift{\Rel} \nu_{11}$ and $\mu_{12} \GRelLift{\Rel} \nu_{12}$,
         and from $\mu_2 \GRelLift{\Rel} \nu_2$ that there exist $\nu_{21}, \nu_{22} : \LiftDist{B}$ such that $\nu_2 \limleadsto \CChoice{p_2}{\nu_{21}}{\nu_{22}}$, and $\mu_{21} \GRelLift{\Rel} \nu_{21}$ and $\mu_{22} \GRelLift{\Rel} \nu_{22}$.
       
       Similar to previous arguments, we can take:
         \begin{align*}
           \nu' & \peq \CChoice{\frac{pp_1}{pp_1+(1-p)p_2}}{\nu_{11}}{\nu_{21}} \\
           \nu'' & \peq \CChoice{\frac{(1-p_1)p}{1 - (pp_1+(1-p)p_2)}}{\nu_{12}}{\nu_{22}}
         \end{align*}
       Since by the sum property of Lemma~\ref{lemma:limleadsto-properties}, and our assumptions: 
       \begin{align*}
         \CChoice{p}{\nu_1}{\nu_2} & \limleadsto \CChoice{p}{(\CChoice{p_1}{\nu_{11}}{\nu_{12}})}{(\CChoice{p_2}{\nu_{21}}{\nu_{22}})} \\
          & \peq \CChoice{pp_1+(1-p)p_2}{(\CChoice{\frac{pp_1}{pp_1+(1-p)p_2}}{\nu_{11}}{\nu_{21}})}{(\CChoice{\frac{(1-p_1)p}{1 - (pp_1+(1-p)p_2)}}{\nu_{12}}{\nu_{22}})}.
       \end{align*}
       Then from the first item of this lemma we know that $(\DistChoice{\frac{pp_1}{pp_1+(1-p)p_2}}{\mu_{11}}{\mu_{21}}) \GRelLift{\Rel} \nu'$ and from the sixth item of this lemma we know $(\DistChoice{\frac{(1-p_1)p}{1 - (pp_1+(1-p)p_2)}}{\mu_{12}}{\mu_{22}}) \GRelLift{\Rel} \nu''$. 
       
       This shows that $(\GChoice{p}{\mu_1}{\mu_2}) \GRelLift{\Rel} (\CChoice{p}{\nu_1}{\nu_2})$ via Definition~\ref{def:liftRel}.\ref{item:both}.\qedhere
   \end{enumerate}
\end{proof}

\begin{proof}[Proof of Lemma~\ref{lem:prereqs:bind-lemma}]$\;$\\
The first statement is by induction on the coupling $\rho: \CoupType\Rel{\mu}{\nu}$.

If $\rho \peq \Dirac{\SemPair{a}{b}}$ for some $a : A$ and $b : B$ such that $a \Rel b$, then $\uext f(\mu) \peq \uext f(\GDirac a) \peq  f(a)$
 and $\uext g(\nu') \peq \uext g(\CDirac b) \peq g(b)$ are related in $\GRelLift{\mathcal{S}}$ by assumption.

If $\rho \peq \DistChoice{p}{\rho_1}{\rho_2}$, let $\mu_i \defeq \DistMarginalFst(\rho_i)$ and $\nu_i \defeq \DistMarginalSnd(\rho_i)$ for $i=1,2$.
 By induction $\uext f(\mu_i) \GRelLift{\mathcal{S}} \uext g(\nu_i)$ and since
\begin{align*}
  \uext f(\mu) & \peq \uext f(\GChoice p{\mu_1}{\mu_2}) \peq \GChoice p{\uext f(\mu_1)}{\uext f(\mu_2)} \\
  \uext g(\nu) & \peq \uext g(\CChoice p{\nu_1}{\nu_2}) \peq \CChoice p{\uext g(\nu_1)}{\uext g(\nu_2)}
\end{align*}
 the case follows from Lemma~\ref{lem:prereqs:choice-lemma}.

For the second statement, suppose that $\mu\GRelLift{\Rel} \nu$. We consider each of the three cases.
  \begin{itemize}
    \item If $\mu : \Dist{A}$, then there exist $\nu_1 : \Dist{B}$ such that $\nu \limleadsto \nu_1$ and there exists a $\rho: \CoupType\Rel{\mu}{\nu_1}$.
          What we want is to show that $\uext f(\mu) \GRelLift{\mathcal{S}} \uext g(\nu)$.
 
          Notice that since the $\limleadsto$ relation is preserved by homomorphisms (Lemma~\ref{lemma:limleadsto-homopreserve}), $\nu \limleadsto \nu_1$ implies $\uext g(\nu) \limleadsto \uext g(\nu_1)$. By Lemma~\ref{lem:prereqs:limleadsto-liftrel} it is therefore enough to show that $\uext f(\mu) \GRelLift{\mathcal{S}} \uext g(\nu_1)$, which follows directly from the first statement of this lemma.

    \item If $\mu : \Dist(\Later{}\GLiftDist{A})$, then $\Later{(\alpha : \kappa)}\  (((\Gzeta (\mu))[\alpha]) \GRelLift{\Rel} \nu)$. By guarded recursion then also: $\Later{(\alpha : \kappa)}\  (\uext f(\Gzeta (\mu)[\alpha]) \GRelLift{\mathcal{S}} \uext g(\nu))$. Note that $\Later{(\alpha : \kappa)}\ (\uext f(\Gzeta (\mu)[\alpha]) = (\Gzeta (\uext f(\mu))[\alpha]))$. We show this by induction on $\mu$:
     
      If $\mu : \Dirac{\mu'}$, then:
      \begin{align*}
         \uext f(\Gzeta (\Dirac{\mu'})[\alpha]) 
        & \peq \uext f(\tapp[\tickA]{\mu'}) \\
        & \peq (\tabs{\tickB}{\kappa} \uext f(\tapp[\tickB]{\mu'})) [\alpha] \\
        & \peq (\Gzeta (\Dirac{\tabs{\tickB}{\kappa} (\uext f(\tapp[\tickB]{\mu'}))}))[\alpha] \\
        & \peq \Gzeta (\uext f(\Dirac{\mu'}))[\alpha].
      \end{align*}

      If $\mu = \GChoice{q}{\mu'}{\mu''}$, then:
      \begin{align*}
        \uext f(\Gzeta (\GChoice{q}{(\tapp[\tickA]{(\mu')})}{(\tapp[\tickA]{(\mu'')})})) 
       & \peq \uext f(\GChoice{q}{(\tapp[\tickA]{(\Gzeta \mu')})}{(\tapp[\tickA]{(\Gzeta \mu'')})})\\
       & \peq \GChoice{q}{(\tapp[\tickA]{\uext f(\Gzeta \mu')})}{(\tapp[\tickA]{\uext f(\Gzeta \mu'')})} \\
       & \peq \GChoice{q}{\tapp[\tickA]{(\Gzeta(\uext f(\mu')))}}{\tapp[\tickA]{(\Gzeta(\uext f(\mu'')))}} \\
       & \peq \tabs{\tickB}{\kappa} \GChoice{q}{(\tapp[\tickB]{(\Gzeta(\uext f(\mu')))})}{(\tapp[\tickB]{(\Gzeta(\uext f(\mu'')))})} [\alpha] \\
       & \peq (\Gzeta (\GChoice{q}{\uext f(\mu')}{\uext f(\mu'')}))[\alpha] \\
       & \peq (\Gzeta (\uext f(\GChoice{q}{\mu'}{\mu''})))[\alpha].
       \end{align*}

       Hence we may conclude that:
       \[
       \Later{(\alpha : \kappa)}\ ((\Gzeta (\uext f(\mu))[\alpha]) \GRelLift{\mathcal{S}} (\uext g(\nu))),
       \]
       which proves that $\uext f(\mu) \GRelLift{\mathcal{S}} \uext g(\nu)$ via the condition~\ref{item:delay} of Definition~\ref{def:liftRel}.
       
    \item If there exist $\mu_1 : \Dist{A}, \mu_2 : \Dist(\Later{}\GLiftDist{A})$, and $p : \II$, such that $\mu \peq \GChoice{p}{\mu_1}{\mu_2}$ then
          from $\mu\GRelLift{\Rel} \nu$ we know that there exist $\nu_1, \nu_2 : \LiftDist{B}$ such that $\nu \limleadsto \CChoice{p}{\nu_1}{\nu_2}$, $\mu_1\GRelLift{\Rel} \nu_1$, and $\mu_2 \GRelLift{\Rel} \nu_2$.
          
           Note that $\uext f(\mu) \peq \DistChoice{p}{\uext f(\mu_1)}{\uext f(\mu_2)}$. Since $\mu_1 : \Dist{A}$, the first case of this lemma proves that from $\mu_1\GRelLift{\Rel} \nu_1$ we may conclude $\uext f(\mu_1)\GRelLift{\Rel} \uext g(\nu_1)$. Similarly since $\mu_2 : \Dist(\Later{}\GLiftDist{A})$, the second case of this lemma proves that from $\mu_2\GRelLift{\Rel} \nu_2$ we may conclude $\uext f(\mu_2)\GRelLift{\Rel} \uext g(\nu_2)$.
          
          Then by Lemma~\ref{lem:prereqs:choice-lemma}: $\uext f(\mu) \peq \DistChoice{p}{\uext f(\mu_1)}{\uext f(\mu_2)} \GRelLift{\Rel} \CChoice{p}{\uext g(\nu_1)}{\uext g(\nu_2)} \peq \uext g(\CChoice{p}{\nu_1}{\nu_2})$.
          
          Finally, by Lemma~\ref{lem:prereqs:leadsto-liftrel} and the fact that $\nu \limleadsto \CChoice{p}{\nu_1}{\nu_2}$, we conclude $\uext f(\mu) \GRelLift{\Rel} \uext g(\nu)$.  
  \end{itemize}
\end{proof}

\subsection*{Section~\ref{sec:rel:syntax:semantics}}
In preparation for the proofs of Lemmas~\ref{lem:guarded-congruence} and~\ref{lem:guarded-fundamental} 
we prove several lemmas which show that $\GLogRel\sigma$ is compatible with the typing rules in Figure~\ref{fig:pfpc:typing}.

\begin{lemma}
  \label{lem:logrel:bind}
  Let $f : \GSVal\sigma \to \GSVal\tau$ and $g: \Val\sigma \to \Val\tau$ as well as $\mu : \GLiftDist{\GSVal\sigma}$ and $\nu : \LiftDist{(\Val\sigma)}$ such that $\mu \GLogRel\sigma \nu$.
  If for any $v:\sigma$ and $V:\Val\sigma$ we have that $(v \GValLogRel\sigma V) \to (f(v) \GValLogRel\tau g(V))$ it follows that
  \[
  (\mu \GBind \GDirac \circ f) \GLogRel\tau (\nu \Bind \CDirac \circ g).
  \]
  or equivalently
  \[
  \GLiftFun{f}(\mu) \GLogRel\tau \LiftFun{g}(\nu)
  \]
\end{lemma}
\begin{proof}
  This is a direct consequence of Lemma~\ref{lem:prereqs:bind-lemma}.
  Looking at the requirements, it suffices to prove that $v \GValLogRel\sigma V$ implies $\GDirac{(f(v))} \GLogRel\Nat \CDirac{(g(V))}$.
  Using Lemma~\ref{lem:prereqs:eta-lemma}, this follows from our assumption that $(v \GValLogRel\sigma V) \to (f(v) \GValLogRel\tau g(V))$.
\end{proof}

\begin{lemma}
  \label{lem:logrel:bind-two}
  Let $f : \GSVal{\sigma_1 \times \sigma_2} \to \GSVal{\sigma_3}$ and $g: \Val{\sigma_1}\times \Val{\sigma_2} \to \Val{\sigma_3}$ as well as $\mu \GLogRel{\sigma_1} d$ and $\nu \GLogRel{\sigma_2} e$.
  If we have that
  \[
  (v \GValLogRel{\sigma_1} V) \to (w \GValLogRel{\sigma_2} W) \to (f(v,w) \GValLogRel{\sigma_3} g(V,W))
  \]
  it follows that
  \begin{align*}
  (\mu \GBind \lambda v. \nu \GBind \lambda w. \GDirac(f(v,w)))  
  \GLogRel\tau &(d \Bind \lambda V. e \Bind \lambda W.\CDirac(g(V,W))).
  \end{align*}
\end{lemma}
\begin{proof}
  We apply Lemma~\ref{lem:prereqs:bind-lemma} with the functions $\lambda v. \nu \GBind \lambda w. \GDirac{(f(v,w))}$ and $\lambda V.e \Bind \lambda W.\CDirac(g(V,W))$.
  Since we have that $\mu \GLogRel{\sigma_1} d$, it suffices to show that
\begin{align*}
  (v \GValLogRel{\sigma_1}V) \to  (&(\nu \GBind \lambda w. \GDirac{(f(v,w))}) 
  \GLogRel{\sigma_3}(e \Bind \lambda W.\CDirac(g(V,W))))
\end{align*}

  Let now $v \GValLogRel{\sigma_1}V$, using Lemma~\ref{lem:logrel:bind} for the functions $\lambda w. \GDirac{(f(v,w))}$ and $\lambda W. \CDirac(g(V,W))$, it suffices to show that
  \[
  (w \GValLogRel{\sigma_2} W) \to  (f(v,w) \GValLogRel{\sigma_3} g(V,W))
  \]
  and the claim follows.
\end{proof}


\begin{corollary}
  Let $\SemOp{op} \in \{\SemSuc, \SemPred{}\}$ and $\mu \GLogRel\Nat \nu$. Then
  \[
  \GLiftFun{\SemOp{op}}(\mu) \GLogRel\Nat \LiftFun{\SemOp{op}}(\nu).
  \]
\end{corollary}

\begin{proof}
  This is a direct consequence of Lemma~\ref{lem:logrel:bind}. Note that if $n \peq m$, then also $\SemOp{op}(n) \peq \SemOp{op}(m)$.
\end{proof}
The arguments for $\SemInl{}$, $\SemInr{}$, $\SemFst{}$ and $\SemSnd{}$ proceed similarly.

\begin{corollary}
 Assume that $\mu_1 \GLogRel\Nat \nu_1$, $\mu_2 \GLogRel\sigma \nu_2$ and $\mu_3 \GLogRel\sigma \nu_3$. Then also
 \[
 \left(\mu_1 \GBind
 \begin{casesalt}
  0 &\mapsto \mu_2 \\
  n+1 &\mapsto \mu_3
 \end{casesalt}\right)
 \GLogRel\sigma
 \left(\nu_1 \Bind
 \begin{casesalt}
  \Numeral{0} & \mapsto \nu_2 \\
  \Numeral{n+1} &\mapsto \nu_3
 \end{casesalt}\right)
 \]
\end{corollary}
\begin{proof}
  We apply the bind lemma for
  $f = \begin{casesalt}
  0 &\mapsto \mu_2 \\
  n+1 &\mapsto \mu_3
 \end{casesalt}
 $
 and
 $g =
 \begin{casesalt}
  \Numeral{0} &\mapsto \nu_2 \\
  \Numeral{n+1} &\mapsto \nu_3
 \end{casesalt}
 $.
 It thus suffices to show that for any $u \GValLogRel\Nat U$ we have that
 $
 f(u) \GLogRel\sigma g(U)
 $.

 We proceed by case distinction on $u : \Nat$.

 \begin{case}{$u \peq 0$}\\
  Then necessarily also $U \peq\Numeral{0}$ and thus $f(u)\peq \mu_2$ and $g(U)\peq \nu_2$.
  But $\mu_2 \GLogRel\sigma \nu_2$ is one of our assumptions.
 \end{case}

 \begin{case}{$u \peq n+1$}\\
  This case is analogous to the previous case.
 \end{case}
\end{proof}

\begin{corollary}
  Let $\mu \GLogRel\sigma d$ and $\nu \GLogRel\tau e$, then
  \begin{align*}
  (\mu \GBind \lambda v. \nu \GBind \lambda w. \GDirac{((v,w))}) 
  & \GLogRel{\ProdTy\sigma\tau} (d \Bind \lambda V. e \Bind \lambda W. \CDirac(\Pair{V}{W})).
  \end{align*}
\end{corollary}
\begin{proof}
 Using \ref{lem:logrel:bind-two} it suffices to show that
  \[
  (v \GValLogRel{\sigma} V) \to (w \GValLogRel{\tau} W) \to ((v,w) \GValLogRel{\ProdTy\sigma\tau} (V,W))
  \]
  which is immediate by the definition of $(v,w) \GValLogRel{\ProdTy\sigma\tau} (V,W)$.
\end{proof}

\begin{lemma}
  \label{lem:app-compatibility}
  Let $\mu_1 \GLogRel{\FnTy\sigma\tau} \nu_2$ and $\mu_2 \GLogRel{\sigma} \nu_2$, then
  \[
  (\mu_1 \GBind v . \mu_2 \GBind w . v(w)) \GLogRel{\tau} (\nu_1 \Bind \lambda (\Lam{M}). \nu_2 \Bind \lambda W. \CStep{}(\Eval(M[W/x]))).
  \]
\end{lemma}
\begin{proof}
To show this, we use Lemma~\ref{lem:prereqs:bind-lemma} on the functions
\begin{enumerate}
  \item $f_1 \defeq \lambda v. \mu_2 \GBind \lambda w. v(w)$
  \item $g_1 \defeq\lambda (\Lam{M}).\nu_2\GBind \lambda W.\CStep{}(\Eval(M[W/x]))$
\end{enumerate}
   and it remains to show that assuming $v \GValLogRel{\FnTy\sigma\tau}\Lam{M}$ we get that
   \[
   \left(\mu_2 \GBind \lambda w. v(w)\right)\GLogRel\tau\left(\nu_2 \Bind \lambda W. \CStep{}(\Eval{(M[W/x])})\right)
   \]
   We proceed by using Lemma~\ref{lem:prereqs:bind-lemma} once more, this time for the functions
   \begin{enumerate}
    \item $f_2 \defeq \lambda w. v(w)$ and
    \item $g_2\defeq\lambda W.\CStep{}(\Eval(M[W/x]))$.
   \end{enumerate}
  Now it suffices to prove that assuming $(v \GValLogRel{\FnTy\sigma\tau} \Lam{M})$ and $(w \GValLogRel{\FnTy\sigma\tau} W)$ we get
\[
 v(w) \GLogRel\tau \CStep{}(\Eval{}(M[W/x]))
\]
Observe that by the definition of $v \GValLogRel{\FnTy\sigma\tau} \Lam{M}$ we immediately get that $v(w) \GLogRel\tau \Eval{}(M[W/x])$ and thus
the result follows by Lemma~\ref{lem:prereqs:leadsto-liftrel}.
\end{proof}

\begin{lemma}
  \label{lem:case-compatibility}
  Assume $\mu_1 \GLogRel{\CoprodTy{\sigma_1}{\sigma_2}} \nu_1$ as well as
  \begin{enumerate}
    \item $\forall v,V. v \GValLogRel{\sigma_1} V \to e_1(v) \GLogRel\tau \CStep{}(\Eval(M[V/x]))$
    \item $\forall v,V. v \GValLogRel{\sigma_2} V \to e_2(v) \GLogRel\tau \CStep{}(\Eval(N[V/x]))$
  \end{enumerate}
  then it follows that
  \begin{align*}
  \left(\mu_1 \GBind
  \begin{casesalt}
    \SemInl{v} \mapsto e_1(v)\\
    \SemInr{v} \mapsto e_2(v)
  \end{casesalt}\right) 
  \GLogRel\tau \left(\nu_1 \Bind
  \begin{casesalt}
    \SemInl{v} \mapsto \CStep{}(\Eval{(M[V/x])}) \\
    \SemInr{v} \mapsto \CStep{}(\Eval{(N[V/x])})
  \end{casesalt}
  \right)
  \end{align*}
\end{lemma}
\begin{proof}
  Again, we use Lemma~\ref{lem:prereqs:bind-lemma} and considering the assumptions of the lemma, it only remains to show that given $w \GValLogRel{\CoprodTy{\sigma_1}{\sigma_2}} W$ we also have
  \begin{align*}
   \left(\begin{casesalt}
    \SemInl{v} \mapsto e_1(v)\\
    \SemInr{v} \mapsto e_2(v)
  \end{casesalt}\right)(w) 
  \GLogRel\tau
  \left(\begin{casesalt}
    \SemInl{v} \mapsto \CStep{}(\Eval{(M[V/x])}) \\
    \SemInr{v} \mapsto \CStep{}(\Eval{(N[V/x])})
  \end{casesalt}
  \right)(W)
  \end{align*}
  This we can show by a case distinction on $w \GValLogRel{\CoprodTy{\sigma_1}{\sigma_2}} W$.
  In the case of $\SemInl{w} \GValLogRel{\CoprodTy{\sigma_1}{\sigma_2}} \Inl{W}$
%
  it remains to show that
    \[
      e_1(w) \GLogRel\tau \CStep{}(\Eval{(M[W/x])})
    \]
    which follows from the assumptions. The other case is similar.
%
%
%
\end{proof}

\begin{lemma}
  \label{lem:fold-compatibility}
  If $\mu \GLogRel{\tau[\RecTy\tau/X]}\nu$ then also
  \[
  \GLiftFun{\GNext}(\mu) \GLogRel{\RecTy\tau} \LiftFun{\Fold{}}(\nu) .
  \]
\end{lemma}
\begin{proof}
  We apply Lemma~\ref{lem:logrel:bind} to $f \defeq \GNext$ and $g \defeq \Fold{}$.
  It thus suffices to show that for $v \GValLogRel{\tau[\RecTy\tau/X]} V$ we have that $\GNext{v} \GValLogRel{\RecTy\tau} \Fold{V}$.
  This is however immediate, since
  \begin{align*}
  \GNext{v} \GValLogRel{\RecTy\tau} \Fold{V} &\defeq (\latbind{\alpha}{\kappa}\tapp{(\GNext{v})} \GValLogRel{\RecTy\tau} V) \\
  &\biimp \latbind{\alpha}{\kappa} (v \GValLogRel{\RecTy\tau} V)
  \end{align*}
\end{proof}

\begin{lemma}
  \label{lem:unfold-compatibility}
  If $\mu \GLogRel{\RecTy\tau} \nu$ then
  \begin{align*}
  (\mu \GBind \lambda v. \GStep(\lambda (\alpha:\kappa). \GDirac{(v[\alpha])})) 
  &\GLogRel{\tau[\RecTy\tau/X]} (\nu \GBind \lambda(\Fold{V}).\CStep{(\CDirac{(V)})})
  \end{align*}
\end{lemma}
\begin{proof}
  We apply Lemma~\ref{lem:prereqs:bind-lemma} with the functions $\lambda v. \GStep(\lambda (\alpha:\kappa). v[\alpha])$ and $\lambda (\Fold{V}). \CStep{}(\CDirac{}(V))$.
  It thereby suffices to show that if $v \GValLogRel{\RecTy\tau} (\Fold{V})$ then also
  \[
    \GStep(\lambda (\alpha:\kappa). \GDirac{(v[\alpha])}) \GLogRel{\tau[\RecTy\tau/X]} \CStep{}(\CDirac(V)).
  \]
  Note that $v \GValLogRel{\RecTy\tau} (\Fold{V})$ is equivalent to $\ClockedLater{}(v[\alpha]\GValLogRel{\tau[\RecTy\tau/X]} V)$, which in turn,
  by Lemma~\ref{lem:prereqs:leadsto-liftrel}, is equivalent to
  \[
  \GStep(\lambda (\alpha:\kappa). \GDirac{(v[\alpha])}) \GLogRel{\tau[\RecTy\tau/X]} \CStep{}(\CDirac(V)).
  \]
  This concludes the claim.
\end{proof}

\begin{lemma}
  \label{lem:choice-compatibility}
  If $\mu \GLogRel\sigma \nu$ and $\mu_1 \GLogRel\sigma \nu_1$ then also
  \[
  \GChoice{p}\mu{\mu_1} \GLogRel\sigma \CChoice{p}{\nu}{\nu_1}
  \]
\end{lemma}
\begin{proof}
  This is a direct consequence of Lemma~\ref{lem:prereqs:choice-lemma}.
\end{proof}

\subsubsection{The Guarded fundamental lemma and congruence lemma}
\label{app:subsec:fundamental-congruence}

\begin{proof}[Proof of Lemma~\ref{lem:guarded-fundamental}]

  The proof proceeds by induction on $M : \Tm[\Gamma]\sigma$ and relies almost entirely on the compatibility lemmas with the only exception of the $M\peq \Lam{M'}$ case.

  \begin{case}{$\Lam{M'}$}

    \noindent By induction hypothesis, we have that
    \[
    \forall \rho, \delta. (\rho \GValLogRel{\Gamma,\sigma} \delta) \to \GInterpret{M'}{\rho} \GLogRel\tau \Eval{(M'[\delta])}.
    \]
    Thus we get that
    \begin{align*}
    &
    \begin{aligned}
    \forall \rho,\delta, v,V. &(\rho \GValLogRel{\Gamma} \delta) \wedge (v \GValLogRel\sigma V) \\
    &\to \GInterpret{M'}{\rho,v} \GLogRel\tau \Eval{(M'[\delta,V/x])}
    \end{aligned}\\
    &\biimp \left(
    \begin{aligned}
    \forall \rho,\delta. &(\rho \GValLogRel\Gamma \delta) \to \forall v,V. (v \GValLogRel\sigma V) \\
    &\to (\GInterpret{M'}{\rho,v} \GLogRel\tau \Eval{((M'[\delta])[V/x])})
    \end{aligned}\right) \\
    &\biimp \left( \forall \rho,\delta. (\rho \GValLogRel\Gamma \delta) \to \GValInterpret{\Lam{M'}}{\rho} \GValLogRel{\FnTy\sigma\tau} \Lam{M'[\delta]} \right)\\
    &\biimp \left(
      \begin{aligned}
      \forall \rho,\delta. &(\rho \GValLogRel\Gamma \delta) \\
      &\to \GDirac{(\GValInterpret{\Lam{M'}}{\rho})} \GLogRel{\FnTy\sigma\tau} \CDirac(\Lam{M'[\delta]})
      \end{aligned} \right)\\
    &\biimp \left(\forall \rho,\delta. (\rho \GValLogRel\Gamma \delta) \to \GInterpret{\Lam{M'}}{\rho} \GLogRel{\FnTy\sigma\tau} \CDirac(\Lam{M'}[\delta])\right)
    \end{align*}
  \end{case}

  \begin{case}{$M\,N$}

  \noindent  Let $\rho \GValLogRel\Gamma \delta$, by induction hypothesis we have $\GInterpret{M}{\rho} \GLogRel{\FnTy\sigma\tau} \Eval{(M[\delta])}$ and $\GInterpret{N}\rho \GLogRel\sigma \Eval{(N[\delta])}$.
  Now Lemma~\ref{lem:app-compatibility} implies that $\GInterpret{MN}\rho \GLogRel\tau \Eval{((M\,N)[\delta])}$.
  \end{case}

  \begin{case}{$\Case{L}{x.M}{y.N}$}\\
     Let $\rho \GValLogRel\Gamma \delta$, by induction hypothesis we have
\begin{enumerate}
  \item $\GInterpret{L}{\rho} \GLogRel{\CoprodTy{\sigma_1}{\sigma_2}} \Eval{(L[\delta])}$
  \item $\forall v, V.(v \GValLogRel{\sigma_1}V) \to \GInterpret{M}{\rho.v} \GLogRel\tau \Eval{(M[\delta,V/x])}$
  \item $\forall w,W. (w \GValLogRel{\sigma_2} W) \to \GInterpret{N}{\rho.w}\GLogRel\tau \Eval{(N[\delta,V/y])}$.
\end{enumerate}
     By Lemma~\ref{lem:case-compatibility} we then get
  \begin{align*}
     \left(\GInterpret{L}{\rho} \GBind
     \begin{casesalt}
     \SemInl{v} \mapsto \GInterpret{M}{\rho.v}\\
     \SemInr{w} \mapsto  \GInterpret{N}{\rho.w}
     \end{casesalt}\right) 
     \GLogRel\tau
     \left(
      \Eval{(L[\delta])} \Bind
      \begin{casesalt}
        \Inl{V} \mapsto \CStep{(\Eval{(M[\delta][V/x])})}\\
        \Inr{V} \mapsto \CStep{(\Eval{(N[\delta][V/x])})}
      \end{casesalt}
      \right)
     \end{align*}
     Observe, that this is precisely
     \[
     \GInterpret{\Case{L}{x.M}{y.M}}\rho \GLogRel\tau \Eval{(\Case{L}{x.M}{y.M}[\delta])}
     \]
  \end{case}

  \begin{case}{$\Unfold{M}$}\\
    Let $\rho \GValLogRel\Gamma \delta$, by the induction hypothesis we have $\GInterpret{M}\rho \GLogRel{\RecTy\tau}\Eval{(M[\delta])}$ and now 
    Lemma~\ref{lem:unfold-compatibility} implies that
    $\GInterpret{\Unfold{M}}\rho \GLogRel{\tau[\RecTy\tau/X]} \Eval{(\Unfold{M}[\delta])}$
  \end{case}

  The remaining cases follow similarly.
\end{proof}

\begin{proof}[Proof of Lemma~\ref{lem:guarded-congruence}]
  The proof proceeds by induction on context derivations. All cases --- except for the $M\peq \Lam{M'}$ case --- are direct consequences of the compatibility lemmas.

  \begin{case}{$C \peq \Lam{C'} : \CtxTy\Delta{\tau_1 \to \tau_2}$} \\
  Assume $M \GOLogRel\sigma N$. We have that $C' : \CtxTy{\Delta,(x:\tau_1)}{\tau_2}$ and by the induction hypothesis it follows that
    \[
    C'[M] \GOLogRel{\tau_2} C'[N]
    \]
    It now follows that
    \begin{align*}
     C'[M] \GOLogRel{\tau_2} C'[N]  
     \defeq\ &\forall \rho,\delta. \rho \GValLogRel{\Delta,\tau_1} \delta \to \GInterpret{C'[M]}\rho \GLogRel{\tau_2} \Eval{((C'[N])[\delta])} \\
     \biimp\ &\left(
     \begin{aligned}
     \forall \rho,\delta,v,V. &\left(\rho\GValLogRel\Delta \delta \wedge v\GValLogRel{\tau_1} V\right) \\
     &\to \GInterpret{C'[M]}{\rho.v} \GLogRel{\tau_2} \Eval{((C'[N])[\delta,V/x])}
     \end{aligned}\right)\\
     \biimp\ &\left(
     \begin{aligned}
     \forall &\rho,\delta (\rho\GValLogRel\Delta \delta) \\
     &\to \left(
            \begin{aligned}
                 \forall &v,V. v\GValLogRel{\tau_1} V \\
                 &\to \GInterpret{C'[M]}{\rho.v} \GLogRel{\tau_2} \Eval{((C'[N])[\delta,V/x])}
            \end{aligned}
      \right)
     \end{aligned}\right)\\
     \biimp\ &\left(
     \begin{aligned}
     \forall &\rho,\delta (\rho\GValLogRel\Delta \delta) \\
     &\to \GValInterpret{\Lam{C'[M]}}{\rho} \GValLogRel{\FnTy{\tau_1}{\tau_2}} (\Lam{C'[N]})[\delta,V/x]
     \end{aligned}\right)\\
     \biimp\ &\left(
     \begin{aligned}
      \forall &\rho,\delta (\rho\GValLogRel\Delta \delta) \\
      &\to \GDirac{(\GValInterpret{\Lam{C'[M]}}{\rho})} \GLogRel{\FnTy{\tau_1}{\tau_2}} \CDirac{(\Lam{C'[N]})[\delta,V/x]}
     \end{aligned}\right)\\
     \biimp\ & \left(\forall \rho,\delta (\rho\GValLogRel\Delta \delta) \to \GInterpret{C[M]}{\rho} \GLogRel{\FnTy{\tau_1}{\tau_2}} (C[N])[\delta,V/x]\right) \\
     \biimp\ &C[M] \GOLogRel{\FnTy{\tau_1}{\tau_2}} C[N]
    \end{align*}
  \end{case}

  \begin{case}{$C \peq \Fold{C'}:\CtxTy\Delta{\RecTy\tau}$}\\
    Assume $M \GOLogRel\sigma N$. We have that $C' :\CtxTy\Delta{\tau[\RecTy\tau/X]}$ and thus by induction hypothesis we have
    $
    C'[M] \GOLogRel{\tau[\RecTy\tau/X]} C'[N].
    $
    By Lemma~\ref{lem:fold-compatibility} we have that
    \begin{align*}
     &C'[M] \GOLogRel{\tau[\RecTy\tau/X]} C'[N] \\
     \biimp\ & \left(\forall\rho,\delta. (\rho\GValLogRel\Delta \delta) \to \GInterpret{C'[M]}\rho \GLogRel{\tau[\RecTy\tau/X]} \Eval{((C'[N])[\delta])}\right)\\
     \to\ &\left(
      \begin{aligned}
      &\forall\rho,\delta. (\rho\GValLogRel\Delta \delta) \\
      &\to \left(\GLiftFun{\GNext}(\GInterpret{C'[M]}\rho)\right) \GLogRel{\RecTy\tau} \LiftFun{\Fold{}}\left(\Eval{}\left({C'[N]}[\delta]\right)\right)
      \end{aligned} \right)
      \\
     \biimp\ &\forall\rho,\delta. (\rho\GValLogRel\Delta \delta) \to \GInterpret{\Fold{C'[M]}}\rho \GLogRel{\RecTy\tau} \Eval{((\Fold{C'[N]})[\delta])}\\
     \biimp\ & C[M] \GOLogRel{\RecTy\tau} C[N]
    \end{align*}
  \end{case}

  \begin{case}{$C \peq \Case{L}{x.C'}{y.N}: \CtxTy\Delta\tau$} \\
    Assume that $M \GOLogRel\sigma M'$. We furthermore have that
    \begin{enumerate}
      \item $\Delta \vdash L : \CoprodTy{\tau_1}{\tau_2}$
      \item $C' : \CtxTy{\Delta,x:\tau_1}\tau$
      \item $\Delta,y:\tau_2 \vdash N : \tau$
    \end{enumerate}
    By the induction hypothesis and Lemma~\ref{lem:guarded-fundamental} it follows directly that
    $C'[M] \GOLogRel\tau C'[M']$, $L \GOLogRel{\CoprodTy{\tau_1}{\tau_2}} L$ and $N \GOLogRel\tau N$.
    Thus, for all $\rho$ and $\delta$ such that $\rho \GValLogRel\Delta \delta$, we get that
    \begin{align*}
      &(1)\quad\GInterpret{L}\rho \GLogRel{\CoprodTy{\tau_1}{\tau_2}} \Eval{(L[\delta])}\\
      &(2)\quad\begin{aligned}
      \forall v,V. &(v \GValLogRel{\tau_1} V) 
      \to \GInterpret{C'[M]}{\rho.v} \GLogRel\tau \Eval{((C'[M'])[\delta,V/x])}
      \end{aligned}\\
      &(3)\quad\forall v,V. (v \GValLogRel{\tau_2} V) \to \GInterpret{N}{\rho.v} \GLogRel\tau \Eval{(N[\delta,V/x])}
    \end{align*}
    Now, using Lemma~\ref{lem:case-compatibility} we conclude that
    \begin{align*}
      \left(\GInterpret{L}\rho \GBind \SemCase{}{\SemInl{v}}{\GInterpret{C'[M]}{\rho.v}}{\SemInr{v}}{\GInterpret{N}{\rho.v}} \right)
      &\GLogRel\tau \Eval{(\Case{L}{x.C'[M']}{y.N}[\delta])}
    \end{align*}
    and thus we get that
    \[
    \forall \rho,\delta. (\rho \GValLogRel\Delta \delta) \to \GInterpret{C[M]}\rho \GLogRel\tau \Eval{(C[M'])}
    \]
  \end{case}

  \begin{case}{$C \peq \Choice{p}{C'}{N}:\CtxTy\Delta\tau$}\\
    Assume $M \GOLogRel\sigma M'$. It follows from the assumptions that furthermore $\Delta \vdash N : \sigma$ and thus by 
    Lemma~\ref{lem:guarded-fundamental} we get $N \GOLogRel\sigma N$.
   Now the induction hypothesis implies $C'[M] \GOLogRel[\Delta]\sigma C'[M']$ and consequently $C[M] \GOLogRel[\Delta]\tau C[M']$ follows by Lemma~\ref{lem:choice-compatibility}.
  \end{case}

  The remaining cases are similar.
\end{proof}


\subsubsection{Proof of Theorem~\ref{thm:eq-pterm-denot-op}}

The proof of Theorem~\ref{thm:eq-pterm-denot-op} is only briefly sketched in the main text. We start by
defining the alternative denotational semantics $\SInterp{-}$ and proving it sound with respect to the operational semantics. 
First recall that $\SInterp{-}$ agrees with $\Interp{-}$ on types and most terms. The only place where they do not agree are 
\begin{align*}
    \GInterpret{MN}\rho &\defeq \GInterpret{M}\rho \GBind \lambda f. \GInterpret{N}\rho \GBind \lambda v. fv \\
    \SGInterpret{MN}\rho &\defeq  \SGInterpret{M}\rho \GBind \lambda f. \SGInterpret{N}\rho \GBind \lambda v.\GStep(\tabs\tickA\kappa fv) \\
     \GInterpret{\Case{L}{x.M}{y.N}}\rho &\defeq \GInterpret{L}\rho \GBind 
    \begin{cases}
      \SemInl{v}\mapsto {\GInterpret{M}{\rho.x\mapsto v}} \\
      \SemInr{v}\mapsto {\GInterpret{N}{\rho.y\mapsto v}}
    \end{cases} \\
     \SGInterpret{\Case{L}{x.M}{y.N}}\rho &\defeq \GInterpret{L}\rho \GBind 
    \begin{cases}
      \SemInl{v}\mapsto \GStep{}(\tabs\tickA\kappa{\SGInterpret{M}{\rho.x\mapsto v}})\\
      \SemInr{v}\mapsto \GStep{}(\tabs\tickA\kappa{\SGInterpret{N}{\rho.y\mapsto v}})
    \end{cases}
\end{align*}

We first show the following soundness theorem. 
 \begin{theorem}[Soundness]
   \label{thm:guarded-soundness}
   For any well typed closed expression $\cdot \vdash M : \sigma$ we have that
   \[
   \GLiftFun{\SGValInterpret{-}{}}(\GEval M) \peq \SGInterpret{M}{}
   \]
 \end{theorem}
\begin{proof}
  First, assume the guarded hypothesis
  \[
   \forall M : \Tm\sigma. (\ClockedLater (\GEval M \GBind \GDirac \circ \SGValInterpret{-}{}) \peq \SGInterpret{M}{}).
  \]
 We proceed with case analysis of term derivations.

\begin{case}{$M = V$}
  The value cases are all the same:
  \begin{align*}
    \GLiftFun{\SGValInterpret{-}{}}(\GEval V) &\peq \GLiftFun{\SGValInterpret{-}{}}(\GDirac(V))\\
    &\peq \GDirac(\SGValInterpret{V}{}) \\
    &\peq \SGInterpret{V}{}
  \end{align*}
\end{case}

\begin{case}{$M = \Suc{M'}$}
  \begin{align*}
  \GLiftFun{\SGValInterpret{-}{}}(\GEval (\Suc{M'})) &\peq\GLiftFun{\SGValInterpret{-}{}}\left(\GLiftFun{\SynSuc{}}(\GEval M')\right) \\
  &\peq \GLiftFun{\SGValInterpret{\SynSuc(-)}{}}(\GEval M') \\
  &\peq \GLiftFun{\SemSuc(\SGValInterpret{-}{})}(\GEval M') \\
  &\peq \GLiftFun{\SemSuc{}}\left(\GLiftFun{\SGValInterpret{-}{}}(\GEval M')\right)\\
  &\peq \GLiftFun{\SemSuc{}}\left(\SGInterpret{M'}{}\right)\\
  &\peq \SGInterpret{\Suc{M'}}{}
  \end{align*}
\end{case}
The case of $M = \Pred{M'}$ is similar
%

\begin{case}{$M = \Ifz{L}{M'}{N}$}
  \begin{align*}
    \GLiftFun{\SGValInterpret{-}{}}(\GEval (\Ifz{L}{M'}{N})) 
    \peq \ & \GLiftFun{\SGValInterpret{-}{}}\left(\GEval L \GBind
    \begin{casesalt}
      \Zero \mapsto \GEval(M') \\
      \Numeral{n+1} \mapsto \GEval(N)
    \end{casesalt} \right)\\
    \peq\ & \left(\GEval L \GBind
    \begin{casesalt}
      \Zero \mapsto \GLiftFun{\SGValInterpret{-}{}}(\GEval(M')) \\
      \Numeral{n+1}\mapsto \GLiftFun{\SGValInterpret{-}{}}(\GEval(N))
    \end{casesalt} \right)\\
  \peq\ &\GEval L \GBind
    \begin{casesalt}
      \Zero \mapsto \SGInterpret{M'}{}\\
      \Numeral{n+1} \mapsto\SGInterpret{N}{}
    \end{casesalt}
    \\
  \peq\ &\GEval L \GBind \lambda \Numeral{n}.\mathsf{match\ }\SGValInterpret{\Numeral{n}}{}\mathsf{\ with}
    \begin{casesalt}
      0 \mapsto \SGInterpret{M'}{}\\
      n+1 \mapsto \SGInterpret{N}{}
    \end{casesalt}
    \\
  \peq\ &(\GEval L \GBind \GDirac \circ \SGValInterpret{-}{}) \GBind
    \begin{casesalt}
      0 \mapsto \SGInterpret{M'}{}\\
      n+1 \mapsto \SGInterpret{N}{}
    \end{casesalt}
    \\
  \peq&\SGInterpret{\Ifz{L}{M'}{N}}{}
  \end{align*}
\end{case}

\begin{case}{$M=\Pair{M'}{N}$}
  \begin{align*}
  \GLiftFun{\SGValInterpret{-}{}}(\GEval{\Pair{M'}{N}})
  \peq\ &\GLiftFun{\SGValInterpret{-}{}}\left(\GEval{M'} \GBind \lambda V.(\GEval N \GBind \lambda W. \GDirac(\Pair{V}{W}))\right)\\
  \peq\ &\GEval{M'} \GBind \lambda V.\left(\GEval N \GBind \lambda W. \GLiftFun{\SGValInterpret{-}{}}\left(\GDirac(\Pair{V}{W})\right)\right)\\
  \peq\ &\GEval{M'} \GBind \lambda V.\GEval N \GBind \lambda W. \GDirac(\SGValInterpret{\Pair{V}{W}}{})  \\
  \peq\ &\GEval{M'} \GBind \lambda V.\GEval N \GBind \lambda W. \GDirac(\SemPair{\SGValInterpret{{V}}{}}{\SGValInterpret{{W}}{}})  \\
  \peq\ &
  \begin{aligned}
  (\GLiftFun{\SGValInterpret{-}{}}(\GEval{M'})) \GBind &\lambda v.(\GLiftFun{\SGValInterpret{-}{}}(\GEval N)) \\
  &\GBind \lambda w. \GDirac(\SemPair{v}{w})
  \end{aligned}\\
  \peq\ & \SGInterpret{M'}{} \GBind \lambda v.\SGInterpret{N}{} \GBind \lambda w. \GDirac(\SemPair{v}{w})  \\
  \peq\ &\SGInterpret{\Pair{M'}{N}}{}
  \end{align*}
\end{case}

\begin{case}{$M = \Fst{M'}$}
  \begin{align*}
  \GLiftFun{\SGValInterpret{-}{}}(\GEval{(\Fst{M'})})
  \peq\ & \GLiftFun{\SGValInterpret{-}{}}\left(\GLiftFun{\SynFst}{}(\GEval{M'})\right)\\
  \peq\ & \GLiftFun{\SGValInterpret{\SynFst{(-)}}{}}\left(\GEval{M'}\right)\\
  \peq\ & \GLiftFun{\SemFst{}}{}\left(\GLiftFun{\SGValInterpret{-}{}}\left(\GEval{M'}\right)\right)\\
  \peq\ & \GLiftFun{\SemFst{}}\left(\SGInterpret{M'}{}\right)\\
  \peq\ & \SGInterpret{\Fst{M'}}{}
  \end{align*}
\end{case}
The case of $M = \Snd{M'}$ is similar. 
%

\begin{case}{$M = MN$}

  \noindent We start by unfolding the definition of $(\GEval MN)\GBind \GDirac \circ \SGValInterpret{-}{}$ and use associativity as well as the definition of $\GBind$ to get
  \begin{align*}
    &
    \begin{aligned}
      \big(\GEval M &\GBind \lambda(\Lam{M'}).\GEval N \GBind \\
      &\lambda V. (\ExpStep(\GEval(M'[V/x])))\big) \GBind \GDirac \circ \SGValInterpret{-}{}
    \end{aligned}\\
    \peq\ & \left(
    \begin{aligned}
      \GEval M \GBind &\lambda(\Lam{M'}).\GEval N \\
      &\GBind \lambda V. (\ExpStep\big(\GEval(M'[V/x])\big) \GBind \GDirac \circ \SGValInterpret{-}{})
    \end{aligned}\right)\\
    \peq\ & \left(
    \begin{aligned}
     \GEval M \GBind &\lambda(\Lam{M'}).\GEval N \\
     &\GBind \lambda V. \ExpStep\left(\GEval(M'[V/x]) \GBind \GDirac \circ \SGValInterpret{-}{}\right)
    \end{aligned}\right)\\
    \peq\ & \left(
    \begin{aligned}
    \GEval M \GBind &\lambda(\Lam{M'}).\GEval N \\
    &\GBind \lambda V. (\GStep{}(\lambda\alpha. \SGInterpret{M'[V/x]}{\rho}))
    \end{aligned}\right)
  \end{align*}
  In the last step we applied the guarded hypothesis.
  Next, we use the substitution lemma for terms and recall the definition of semantic function type values, which leaves us with
  \begin{align*}
    \peq\ &\GEval M \GBind \lambda(\Lam{M'}).\GEval N \GBind \lambda W. (\ExpStep (\SGInterpret{\ M'}{\SGValInterpret{W}{}})) \\
    \peq\ &\left(
      \begin{aligned}
    \GEval M \GBind &\lambda(\Lam{M'}).\GEval N \\
    &\GBind \lambda W. (\ExpStep(\SGValInterpret{\Lam{M'}}{}(\SGValInterpret{W}{})))
      \end{aligned} \right)\\
    \peq\ &\GEval M \GBind \lambda V. \GEval N \GBind \lambda W. (\ExpStep (\SGValInterpret{V}{}(\SGValInterpret{W}{})))
  \end{align*}
  In the last equation, we simply omitted the superfluous case analysis of function values.
  We now use the functoriality of $\GLiftDist{}$ to get that the above equals
  \begin{align*}
    (\GEval M\GBind \SGValInterpret{-}{}) \GBind &\lambda v. (\GEval N \GBind \SGValInterpret{-}{}) 
    \GBind \lambda w. \ExpStep(vw)
  \end{align*}
  which by the induction hypothesis gives
   \begin{align*}
     \SGInterpret{M}{} \GBind \lambda v. \SGInterpret{N}\cdot \GBind \lambda w. \GStep(\lambda(\alpha:\kappa).v w) 
    &\peq \SGInterpret{M}{} \cdot \SGInterpret{N}{}
    \peq \SGInterpret{MN}{}
  \end{align*}
\end{case}

\begin{case}{$M=\Inl{M'}$}
  \begin{align*}
  \GLiftFun{\SGValInterpret{-}{}}\left(\GEval{(\Inl{M'})}\right) 
  \peq\ &\GLiftFun{\SGValInterpret{-}{}}\left(\GLiftFun{\Inl{}}\left(\GEval{M'}\right)\right) \\
  \peq\ &\GLiftFun{\SGValInterpret{\Inl{(-)}}{}}\left(\GEval{M'}\right) \\
  \peq\ &\GLiftFun{\SemInl{}}{}\left(\GLiftFun{\SGValInterpret{-}{}}{}\left(\GEval{M'}\right)\right) \\
  \peq\ &\GLiftFun{\SemInl{}}{}\left(\SGInterpret{M'}{}\right) \\
  \peq\ &\SGInterpret{\Inl{M'}}{}
  \end{align*}
\end{case}
The case of $M=\Inr{M'}$ is similar.
%

\begin{case}{$M = \Case{L}{x.M}{y.N}$}
%
  We proceed similarly to the function application case and first unfold the definition. 
   \begin{align*}
   & \GLiftFun{\SGValInterpret{-}{}}\left(\GEval{(\Case{L}{x.M}{y.N})}\right) \\
    \peq &\GLiftFun{\SGValInterpret{-}{}}\left( \GEval L \GBind
    \begin{casesalt}
      \Inl{V} \mapsto (\ExpStep\GEval(M[V/x])) \\
      \Inr{V} \mapsto (\ExpStep\GEval(N[V/x]))
    \end{casesalt} \!\!\!\!\! \right)\\
    \peq\ &\GEval L \GBind
    \begin{casesalt}
      \Inl{V} \mapsto\GLiftFun{\SGValInterpret{-}{}}\left(\ExpStep(\GEval(M[V/x]))\right) \\
      \Inr{V} \mapsto\GLiftFun{\SGValInterpret{-}{}}\left(\ExpStep(\GEval(N[V/x]))\right)
    \end{casesalt}  \\
    \peq\ &\GEval L \GBind
    \begin{casesalt}
      \Inl{V} \mapsto\ExpStep(\GLiftFun{\SGValInterpret{-}{}}\left(\GEval(M[V/x])\right)) \\
      \Inr{V} \mapsto\ExpStep(\GLiftFun{\SGValInterpret{-}{}}\left(\GEval(N[V/x])\right))
    \end{casesalt}  \\
    \peq\ &\GEval L \GBind
    \begin{casesalt}
      \Inl{V} \mapsto (\ExpStep\SGInterpret{M[V/x]}{}) \\
      \Inr{V} \mapsto (\ExpStep\SGInterpret{N[V/x]}{})
    \end{casesalt}  \\
    \peq\ &\GEval L \GBind
    \begin{casesalt}
      \Inl{V} \mapsto \ExpStep(\SGInterpret{M}{\SGValInterpret{V}{}}) \\
      \Inr{V} \mapsto \ExpStep(\SGInterpret{N}{\SGValInterpret{V}{}})
    \end{casesalt}  \\
    \peq\ &\GEval L \GBind \lambda V. \mathsf{match\ }\SGValInterpret{V}{}\mathsf{\ with}
    \begin{casesalt}
      \SemInl{v} \mapsto (\ExpStep\SGInterpret{M}{v}) \\
      \SemInr{v} \mapsto (\ExpStep\SGInterpret{N}{v})
    \end{casesalt} \\
    \peq\ &\left(\GLiftFun{\SGValInterpret{-}{}}(\GEval L)\right) \GBind
    \begin{casesalt}
      \SemInl{v} \mapsto \ExpStep(\SGInterpret{M}{v}) \\
      \SemInr{v} \mapsto \ExpStep(\SGInterpret{N}{v})
    \end{casesalt} \\
    \peq\ &\SGInterpret{L}{}\GBind
    \begin{casesalt}
      \SemInl{v} \mapsto \ExpStep(\SGInterpret{M}{v}) \\
      \SemInr{v} \mapsto \ExpStep(\SGInterpret{N}{v})
    \end{casesalt}\\
    \peq\ &\SGInterpret{\Case{L}{x.M}{y.N}}{}
  \end{align*}
\end{case}

\begin{case}{$M = \Fold{M'}$}
 \begin{align*}
    \GLiftFun{\SGValInterpret{-}{}}\left(\GEval (\Fold{M'})\right)  
   \peq\  &\GLiftFun{\SGValInterpret{-}{}}\left(\GLiftFun{(\Fold{})}\left(\GEval (M')\right)\right)  \\
   \peq\  &\GLiftFun{\SGValInterpret{\Fold{-}}{}}\left(\GEval (M')\right)  \\
   \peq\  &\GLiftFun{\GNext \circ \SGValInterpret{-}{}}\left(\GEval (M')\right)  \\
   \peq\  &\GLiftFun{\GNext} \left(\Dist{(\SGValInterpret{-}{})}\left(\GEval (M')\right)\right)  \\
   \peq\  &\GLiftFun{\GNext} \left(\SGInterpret{M'}{}\right)  \\
    \peq\ &\SGInterpret{\Fold{M'}}{}
 \end{align*}
\end{case}

\begin{case}{$M=\Unfold{M'}$}. 
  Note that for $V' : \Val{\tau[\RecTy\tau/X]}$ we have that
  \[
    \latbind{\alpha}{\kappa}\left(\tapp{\SGValInterpret{\Fold{V'}}{}} \peq\SGValInterpret{V'}{}\right)
    \]
  and thus we get
  \begin{align*}
    &\GLiftFun{\SGValInterpret{-}{}}\left(\GEval (\Unfold M')\right) \\
    \peq\ &\GLiftFun{\SGValInterpret{-}{}}\left(\GEval(M')\GBind \lambda (\Fold{V'}). \ExpStep (\GDirac(V'))\right) \\
    \peq\ &\GEval(M')\GBind \lambda (\Fold{V'}).\GLiftFun{\SGValInterpret{-}{}}(\ExpStep(\GDirac(V'))) \\
    \peq\ &\GEval(M')\GBind \lambda (\Fold{V'}).(\GStep (\tabs{\alpha}{\kappa}(\GDirac(\SGValInterpret{V'}{})))) \\
    \peq\ &\GEval(M')\GBind \lambda (\Fold{V'}).(\GStep (\tabs{\alpha}{\kappa}(\GDirac(\tapp{\SGValInterpret{\Fold{V'}}{}})))) \\
    \peq\ &\GEval(M')\GBind \lambda V.(\GStep (\tabs{\alpha}{\kappa}(\GDirac(\tapp{\SGValInterpret{V}{}})))) \\
    \peq\ &\GEval(M')\GBind \left(\lambda v.\GStep(\tabs{\alpha}{\kappa}\GDirac(\tapp{v}))\right) \circ \SGValInterpret{-}{} \\
    \peq\ &(\GLiftFun{\SGValInterpret{-}{}}(\GEval(M')))\GBind \lambda v.\GStep(\tabs{\alpha}{\kappa}\GDirac(\tapp{v}{})) \\
    \peq\ &\SGInterpret{\Unfold M'}{}
  \end{align*}
\end{case}

\begin{case}{$M=\Choice{p}{N_1}{N_2}$}
  \begin{align*}
    & \GLiftFun{\SGValInterpret{-}{}}\left(\GEval{}(\Choice{p}{N_1}{N_2})\right) \\
    \peq\ &\GLiftFun{\SGValInterpret{-}{}}\left(\GChoice{p}{\GEval{}(N_1)}{\GEval{}(N_2)}\right)  \\
    \peq\ &\GChoice{p}{\left(\GLiftFun{\SGValInterpret{-}{}}\left(\GEval{}(N_1)\right)\right)}{\left(\GLiftFun{\SGValInterpret{-}{}}\left(\GEval{}(N_2)\right)\right)}  \\
    \peq\ &\GChoice{p}{\SGInterpret{N_1}{}}{\SGInterpret{N_2}{}} \\
    \peq\ &\SGInterpret{\Choice{p}{N_1}{N_2}}{}
  \end{align*}
\end{case}
\end{proof}

As a consequence, we note the following easy corollary.

\begin{corollary}  \label{cor:soundness:PT:unit}
If $M$ is a closed term of unit type, then
 $\probtermseq {\Eval{\,M}} \peq \probtermseq {\SInterp{M}}$
\end{corollary}

To relate the two denotational semantics we 
construct a logical relation 
\begin{align*}
\GSValLogRel\sigma & :\SGSVal\sigma \to \SVal\sigma \to \Prop \\
\GLogRel\sigma & :\GLiftDist{}{\SGSVal\sigma} \to \LiftDist{}(\SVal\sigma) \to \Prop
\end{align*}
by guarded recursion and induction over $\sigma$ as follows
\begin{align*}
v \GSValLogRel\Unit w & \defeq  \True \\
n \GSValLogRel\Nat m & \defeq  (n\peq m) \\
v \GSValLogRel{\ProdTy \sigma\tau} w & \defeq  \left(\SemFst{}v \GSValLogRel{\sigma} \SemFst{} w\right) \land\left( \SemSnd{}v \GSValLogRel{\tau} \SemSnd{} w \right) \\
v \GSValLogRel{\CoprodTy \sigma\tau} w & \defeq  \left(\exists v', w' . v \peq \SemInl v' \land w \peq \SemInl w' \land v' \GSValLogRel{\sigma} w'\right) \lor  \\
 & \hspace{.5cm} \left(\exists v', w' . v \peq \SemInr v' \land w \peq \SemInr w' \land v' \GSValLogRel{\tau} w'\right) \\
f \GSValLogRel{\sigma\to\tau} g & \defeq \forall v,w . (v \GSValLogRel{\sigma} w) \to (f\,v)\GSLogRel{\tau}(\lambda\kappa. \capp g\,(\capp w)) \\
v \GSValLogRel{\RecTy\tau} w & \defeq \latbind\tickA\kappa (\tapp v \GSValLogRel{\tau\subst{\RecTy\tau}X}\force( w)) \\
x \GSLogRel{\sigma} y & \defeq x \GRelLift{\GSValLogRel{\sigma}} y
\end{align*}

Note that many of these cases use implicit isomorphism, for example
\begin{align*}
 \SVal{\Nat} & = \forall\kappa. \NN \equi \NN \\
 \SVal{\CoprodTy\sigma\tau} & = \forall\kappa . (\GSVal{\sigma} + \GSVal{\tau}) \equi (\forall\kappa . \GSVal{\sigma}) + (\forall\kappa . \GSVal{\tau}) = \SVal\sigma + \SVal\tau
\end{align*}
In the case of function types, $g : \SVal{\sigma\to\tau}$ and so $\lambda\kappa. \capp g\,(\capp w) : \forall\kappa. \GLiftDist(\GSVal\tau)$, which is not quite right, because,
to apply $\GSLogRel{\tau}$ we need something of type $\LiftDist(\SVal\tau)$. However, the next lemma states that these sets are isomorphic, and we will leave the 
isomorphism implicit in order not to clutter proofs.

%

\begin{lemma} \label{lem:forall:iso}
 Given a type $X$, in which $\kappa$ may appear, the two types $\LiftDist (\forall\kappa .X)$ and $\forall\kappa .\GLiftDist(X)$ are isomorphic. The isomorphism is natural in $X$. 
\end{lemma}

\begin{proof}
 Define the map $\Psi: \LiftDist (\forall\kappa .X) \to \forall\kappa .\GLiftDist(X)$ as $\forall\kappa.\GLiftDist(\mathsf{ev}_\kappa)$, so 
\begin{align*}
 \capp{\Psi(\CDirac{x})} & = \GDirac{(\capp x)} \\
 \capp{\Psi(\CStep{x})} & = \GStep{(\tabs\tickA\kappa\capp{\Psi(x)})} \\ 
 \capp{\Psi(\CChoice pxy)} & = \GChoice p{\capp{\Psi(x)}}{\capp{\Psi(y)}}
\end{align*}
 For the other direction, use the isomorphism
 \begin{equation} \label{eq:lem:forall:iso}
 \forall\kappa .\GLiftDist(X) \equi \Dist{}(\forall\kappa . X + \forall\kappa .\GLiftDist(X))
 \end{equation}
 to define $\Phi :  (\forall\kappa .\GLiftDist(X)) \to \LiftDist (\forall\kappa .X)$ and $\Phi^\kappa  : (\forall\kappa .\GLiftDist(X)) \to \GLiftDist (\forall\kappa .X)$ by
\begin{align*}
 \Phi(x) & \defeq \forall\kappa . \Phi^\kappa(x) \\
 \Phi^\kappa(\Dirac{\inl\,{x}}) & \defeq \GDirac(x) \\
 \Phi^\kappa(\Dirac{\inr\,{x}}) & \defeq \GStep(\tabs\tickA\kappa{\Phi^\kappa(x)}) \\
 \Phi^\kappa(\DistChoice pxy) & \defeq \GChoice p{\Phi^\kappa(x)}{\Phi^\kappa(y)}
\end{align*}
We show that $\capp{(\Phi\circ\Psi(x))} = \capp x$ by guarded recursion and $\Dist$-induction:
\begin{align*}
 \capp{(\Phi\circ\Psi(\CDirac(x)))} & = \Phi^\kappa(\Dirac{\inl{\,x}}) \\
 & = \GDirac{x} \\
 & = \capp{\CDirac x} \\
 \capp{(\Phi\circ\Psi(\CStep(x)))} & = \Phi^\kappa(\Psi(\CStep x)) \\
 & = \Phi^\kappa(\Dirac{\inr\,(\Psi(x))} \\
 &  =  \GStep(\tabs\tickA\kappa{\Phi^\kappa(\Psi(x))}) \\
 & = \GStep(\tabs\tickA\kappa{\capp x}) \\
 & = \capp{\CStep x}
\end{align*}
The case of $\CChoice pxy$ is also easy. Similarly, we show that $\capp{\Psi(\Phi(x))} = \capp x$ by guarded recursion:
\begin{align*}
 \capp{\Psi(\Phi(\Dirac{\inl\,{x}}))} & = \capp{\Psi(\CDirac(x))} 
 = \GDirac{(\capp x)} 
\end{align*}
 which up to the isomorphism (\ref{eq:lem:forall:iso}) is $\Dirac{\inl\,{x}}$. 
 \begin{align*}
 \capp{\Psi(\Phi(\Dirac{\inr\,{x}}))} & = \capp{\Psi(\lambda\kappa . \GStep(\tabs\tickA\kappa{\Phi^\kappa(x)}))} \\ 
 & = \GStep(\tabs\tickA\kappa{\capp{\Psi(\Phi(x))}}) \\
 & = \GStep(\tabs\tickA\kappa{\capp x})
\end{align*}
which up to the isomorphism (\ref{eq:lem:forall:iso}) is $\Dirac{\inr\,{x}}$.
\end{proof}

\begin{lemma} \label{lem:fund:log:rel:sem}
 If $M : \Tm[\Gamma]\sigma$ and $\rho\GSValLogRel\Gamma \delta$, then $\SGInterpret M\rho \GSLogRel{\sigma} \Interpret M\delta$. 
\end{lemma}
Here $\rho$ is as usual, $\delta$ is a mapping of variables $x: \tau$ in $\Gamma$ to $\delta (x) : \SVal\tau$ and $\Interpret M\delta \defeq \lambda\kappa. \GInterpret M{\capp \delta}$,
where $\capp \delta$ is pointwise application to $\kappa$. We have left the isomorphism of Lemma~\ref{lem:forall:iso} implicit in the expression $\Interpret M\delta$.


Before proving Lemma~\ref{lem:fund:log:rel:sem}, we prove a lemma that relates the $\Bind$ of $\LiftDist$ to the $\GBind$ for $\GLiftDist$ as it is used in the definition of $\Interp M_\delta$.
This will allow us to use the bind lemma in the proof.

\begin{lemma} \label{lem:two:binds}
 Let $g : \forall\kappa. \GSVal{\sigma} \to \GLiftDist(\GSVal\tau)$ and $x : \forall\kappa . \GLiftDist(\GSVal\sigma)$, and let $\Psi : \forall\kappa . \GLiftDist(\GSVal\sigma) \to \LiftDist(\SVal\sigma)$
 be the isomorphism of Lemma~\ref{lem:forall:iso}. Then
 \[
  (\Psi(x) \Bind z . \lambda\kappa . \capp g (\capp z)) = 
  \lambda\kappa . (\capp x) \GBind \capp g 
 \]
 both sides of type $\forall\kappa.  \GLiftDist(\GSVal\tau)$.
\end{lemma}

\begin{proof}
  First recall that $\Bind$ is defined in terms of $\GBind$. Precisely, given $x : \LiftDist X$, $f : X \to \LiftDist Y$,
\begin{align*}
 (x \Bind f) = \lambda\kappa . (\capp x) \GBind \lambda z. (\capp{f(z)})
\end{align*}
We prove that 
\[
\capp{(\Psi(x) \Bind z . \lambda\kappa . \capp g (\capp z))} = 
  (\capp x) \GBind \capp g
\]
by guarded recursion. Since $\Dist$ commutes with $\forall\kappa$,  $x$ is on one of the following three forms
\begin{align*}
 x & = \lambda\kappa . \Dirac{\inl(\capp {x'})} && \text{some } x' : \forall\kappa. \GSVal\sigma \\
 x & = \lambda \kappa . \Dirac{\inr(\GStep(\nextop(\capp{x'}))) }&& \text{some } x' : \forall\kappa . \GLiftDist(\GSVal\sigma) \\
 x & =  \lambda\kappa . \DistChoice p{\capp y}{\capp z} && \text{some } z,y :  \forall\kappa . \GLiftDist(\GSVal\sigma)
\end{align*}
In the first case both sides reduce to $\capp g (\capp{x'})$. In the second case we get
\begin{align*}
 \capp{(\Psi(x) \Bind z . \lambda\kappa . \capp g (\capp z))} 
 & = (\capp{\Psi(x)} \GBind z . \capp g (\capp z)) \\
 & = (\Dirac{\inr(\GStep(\nextop(\capp{(\Psi(x'))})))} \GBind z . \capp g (\capp z)) \\
 & = \Dirac{\inr(\GStep(\nextop(\capp{(\Psi(x')} \GBind z .\capp g (\capp z)))))} \\
 & = \Dirac{\inr(\GStep(\nextop(\capp{(\Psi(x') \Bind z . \lambda\kappa . \capp g (\capp z))})))} \\
 & = \Dirac{\inr(\GStep(\nextop(\capp {x'}) \GBind \capp g))} \\
 & = \Dirac{\inr(\GStep(\nextop(\capp {x'})))} \GBind \capp g \\
 & = \capp x \GBind \capp g 
\end{align*}
We omit the last easy case. 
\end{proof}

\begin{proof}[Proof of Lemma~\ref{lem:fund:log:rel:sem}]
  The proof is by induction on the typing derivation of $M$. 
  Most cases are trivial, we just do a few.
  
  Case $M = N\, P$. If $f \GSValLogRel{\sigma\to\tau} g$ and $v \GSValLogRel{\sigma} w$, then by definition $f\, v\GSLogRel{\tau} \lambda\kappa. \capp g\, (\capp w)$. So, then also 
  \[\GStep(\tabs\tickA\kappa(f\, v))\GSLogRel{\tau} \lambda\kappa. \capp g\, (\capp w)\] 
  Now, by Lemma~\ref{lem:two:binds},
\begin{align*}
 \Interp{N\, P}_\delta & = \lambda\kappa . \GInterpret{N\, P}{\capp \delta} \\
 & = \lambda\kappa . \GInterpret{N}{\capp\delta} \GBind  g. \GInterpret{P}{\capp\delta} \GBind \lambda w. gw \\
 & = \Interp{N}_\delta \Bind g. \Interp{P}_\delta \Bind w. \lambda\kappa. \capp g(\capp w)
\end{align*}
  The case now follows by the bind lemma and the induction hypothesis. 
  
  
  Case $M = \Fold N$. By induction $\SGInterpret N\rho \GSLogRel{\tau\subst{\RecTy\tau}X} \Interpret N\delta$. We must show that 
\begin{align*}
 \GLiftDist{(\GNext)}(\SGInterpret{M}\rho) \GRelLift{\GSValLogRel{\RecTy\tau}} \lambda\kappa . \GLiftDist{(\GNext)}(\GInterpret{M}{\capp\delta})
\end{align*}
Suppose first that $v\GSValLogRel{\tau\subst{\RecTy\tau}X} w$. Then 
\begin{align*}
  (\GNext v \GSValLogRel{{\RecTy\tau}} \lambda\kappa . (\GNext {\capp w})) &  \biimp \Later(v\GSValLogRel{\tau\subst{\RecTy\tau}X} w)
\end{align*}
because $\lambda\kappa . (\GNext {\capp w}) \leadsto w$, and so is true. Note that 
\begin{align*}
 \GLiftDist{(\GNext)}(\SGInterpret{M}\rho) & =  \SGInterpret{M}\rho \GBind v. \GDirac(\GNext v) \\
 \lambda\kappa . \GLiftDist{(\GNext)}(\GInterpret{M}{\capp\delta}) & =  \lambda\kappa . \GInterpret{M}{\capp\delta} \GBind w. \GDirac(\GNext w) \\
 & =  \Interpret{M}\delta \Bind w. \CDirac(\lambda\kappa. (\GNext (\capp w))) 
\end{align*}
using Lemma~\ref{lem:two:binds} in the last of these, so that the case now follows from an application of the bind lemma.

Case $M = \Unfold N$. By induction $\SGInterpret N\rho \GSLogRel{\RecTy\tau} \Interpret N\delta$. We must show that 
\begin{align*}
 \SGInterpret{M}\rho \GBind\lambda v. \GStep (\tabs\tickA\kappa \GDirac{(\tapp v)})
  \GRelLift{\GSValLogRel{\tau\subst{\RecTy\tau}X}} 
  \lambda\kappa . \GInterpret{M}{\capp\delta} \GBind\lambda w. \GStep (\tabs\tickA\kappa \GDirac{(\tapp w)})
\end{align*}
Note first that by Lemma~\ref{lem:two:binds}, 
\[
\lambda\kappa . \GInterpret{M}\delta \GBind\lambda w. \GStep (\tabs\tickA\kappa \GDirac{(\tapp w)})
= \Interp{M}_\delta \Bind \lambda w . \CStep(\CDirac w) 
\]
so by the bind lemma, the proof obligation reduces to showing 
\begin{equation} \label{eq:fund:log:rel:sem:unfold:goal}
\GStep (\tabs\tickA\kappa \GDirac{(\tapp v)}) \GRelLift{\GSValLogRel{\tau\subst{\RecTy\tau}X}}  \CStep(\CDirac w)
\end{equation}
under the assumption that $v \GSValLogRel{\RecTy\tau} w$. By definition, the assumption means 
\[
 \latbind\tickA\kappa(\tapp v \GSValLogRel{\tau\subst{\RecTy\tau}X} \force (w))
\]
and this is equivalent to (\ref{eq:fund:log:rel:sem:unfold:goal}) because $\CStep(\CDirac w) \leadsto \force(w)$.


In the case of $M = \Case{L}{x.N}{y.P}$, note that
\begin{align*}
 \Interpret{\Case{L}{x.N}{y.P}}\rho 
 &\defeq \lambda\kappa . \GInterpret{L}{\capp\rho} \GBind 
    \begin{cases}
      \SemInl{w}\mapsto {\GInterpret{N}{(\capp\rho).x\mapsto w}} \\
      \SemInr{w}\mapsto {\GInterpret{P}{(\capp\rho).y\mapsto w}}
    \end{cases} \\
 & = \Interpret{L}{\rho} \Bind 
    \begin{cases}
      \SemInl{w}\mapsto {\Interpret{N}{\rho.x\mapsto w}} \\
      \SemInr{w}\mapsto {\Interpret{P}{\rho.y\mapsto w}}
    \end{cases} \\
\end{align*}
so that by the induction hypothesis and the bind lemma, the proof obligation reduces to showing 
\[
 \left(\GDirac v\GBind 
    \begin{cases}
      \SemInl{v}\mapsto \GStep{}(\lambda\tickA.{\SGInterpret{N}{\rho.x\mapsto v}})\\
      \SemInr{v}\mapsto \GStep{}(\lambda\tickA.{\SGInterpret{P}{\rho.y\mapsto v}})
    \end{cases} \right)
    \GRelLift{\GSValLogRel{\tau}}
 \left(\CDirac w\Bind 
    \begin{cases}
      \SemInl{w}\mapsto {\Interpret{N}{\rho.x\mapsto w}} \\
      \SemInr{w}\mapsto {\Interpret{P}{\rho.y\mapsto w}}
    \end{cases} \right)
\]
assuming $v \GSValLogRel{\CoprodTy{\sigma_1}{\sigma_2}}w$. There are two cases for the latter, and we just show
the case of $v = \SemInl{v'}$ and $w = \SemInl{w'}$ with $v' \GSValLogRel{\sigma_1}w'$. In that case the obligation reduces to 
\[
\GStep{}(\lambda\tickA.{\SGInterpret{N}{\rho.x\mapsto v'}}) \GRelLift{\GSValLogRel{\tau}} {\Interpret{N}{\rho.x\mapsto w'}}
\]
which by definition is 
\[
\Later(\SGInterpret{N}{\rho.x\mapsto v'} \GRelLift{\GSValLogRel{\tau}} {\Interpret{N}{\rho.x\mapsto w'}})
\]
which is just $\Later$ applied to the induction hypothesis for $N$. 
\end{proof}

We can now finally prove Theorem~\ref{thm:eq-pterm-denot-op}. 

\begin{proof}[Proof of Theorem~\ref{thm:eq-pterm-denot-op}]
 One direction follows directly from Lemma~\ref{lem:guarded-fundamental}: Since $\forall \kappa. \GInterpret{M}{} \GLogRel\Unit \Eval{\,M}$, by 
 Lemma~\ref{lem:eq1-probterm} also $\probtermseq{\Interpret{M}{}} \leqlim \probtermseq{\Eval \,M}$. For the other direction, 
 $\probtermseq{\Eval \,M} \peq \probtermseq{\SInterp{M}{}}$ by Corollary~\ref{cor:soundness:PT:unit}, and since 
 $\forall\kappa.(\SGInterpret M{} \GSLogRel{\Unit} \Interp M{})$ also $\probtermseq{\SInterp M} \leqlim \probtermseq{\Interp M{}}$ by Lemma~\ref{lem:eq1-probterm},
 so $\probtermseq{\Eval \,M} \leqlim \probtermseq{\Interp M{}}$ as required.
\end{proof}


\subsection*{Section~\ref{sec:examples}}

\begin{proof}[Proof of Theorem~\ref{ex:thm:randw}]
 Let $V_{\RandW}$, $V_{\RandWTwo}$ and $\VEverySnd$ be the values that 
 $\RandW$, $\RandWTwo$ and $\EverySnd$ reduce to. It suffices to show that $\VEverySnd(\VRandW(\Numeral{2n})) \CtxEq \VRandWTwo(\Numeral{2n})$. 
 We first show $\GInterpret{\VEverySnd(\VRandW(\Numeral{2n}))}{} \GLogRel\LazyL \Eval(\VRandWTwo(\Numeral{2n}))$ by guarded recursion, focussing just on the case of
 $n>0$, which is the harder one. Now,
\begin{align*}
 \GInterpret{\VRandW(\Numeral{2n})}{} & = \ExpStep\GInterpret{G_{\RandW}\, \VRandW (\Numeral{2n})}{} \\
 & \peq \ExpStep(\GDirac(\nextop(\SemInr{\SemPair{\Numeral{2n}}{
 \lambda\_ .\GInterpret{{\Choice{\frac{1}{2}}{\VRandW(\Numeral{2n\!-\!1})}{\VRandW(\Numeral{2n\!+\!1)}}}}{}}}))) \\
 & \peq \ExpStep\left( \GDirac(\nextop(\SemInr{\SemPair{\Numeral{2n}}{\lambda\_ .\GChoice{\frac12}{\GInterpret{\VRandW(\Numeral{2n\!-\!1})}{}}{\GInterpret{\VRandW(\Numeral{2n\!+\!1})}{}} } }))  \right)
\end{align*}
So that 
\begin{align*}
 \GInterpret{\VEverySnd(V_{\RandW}(\Numeral{2n}))}{} &
 \peq  \ExpStep(\GInterpret{G_{\EverySnd}(\YComb{}{\,G_{\EverySnd}})(V_{\RandW}(\Numeral{2n}))}{}) \\
 & \peq  \ExpStep(\GInterpret{G_{\EverySnd}\,\VEverySnd(V_{\RandW}(\Numeral{2n}))}{}) \\
 & \peq (\ExpStep)^3( \GDirac(\nextop(\SemInr{\SemPair{\Numeral{2n}}{\lambda\_ . W_{\RandW}}})))
\end{align*}
where 
\[
W_{\RandW} =  
  \left( \GChoice{\frac12}{\GInterpret{\VEverySnd(\Tail(\VRandW(\Numeral{2n\!-\!1})))}{}}{\GInterpret{\VEverySnd(\Tail(\VRandW(\Numeral{2n\!+\!1})))}{}}  \right) 
\]
Similarly, 
\begin{align*}
 \Eval(\VRandWTwo(\Numeral{2n})) & \leadsto \Eval(\Cons{\Numeral{2n}}{\Lam[y]{\Ifz{\Numeral{2n}}{\Nil}{\Choice{\frac{1}{2}}{\VRandWTwo(\Numeral{2n})}{\Choice{\frac{1}{2}}{\VRandWTwo(\Numeral{2n-2})}{\VRandWTwo(\Numeral{2n+2})}}}}}) \\
 & \peq \Fold{(\Inr(\Numeral{2n}, \Lam[y]{\Ifz{\Numeral{2n}}{\Nil}{\Choice{\frac{1}{2}}{\VRandWTwo(\Numeral{2n})}{\Choice{\frac{1}{2}}{\VRandWTwo(\Numeral{2n\!-\!2})}{\VRandWTwo(\Numeral{2n\!+\!2})}}}}))}
\end{align*}
So showing $\GInterpret{\VEverySnd(\VRandW(\Numeral{2n}))}{} \GLogRel\LazyL \Eval(\VRandWTwo(\Numeral{2n}))$ easily reduces to showing 
\begin{equation} \label{eq:randomwalk:1:app}
 \Later(W_{\RandW} \GLogRel\LazyL \Eval(\Choice{\frac{1}{2}}{\VRandWTwo(\Numeral{2n})}{\Choice{\frac{1}{2}}{\VRandWTwo(\Numeral{2n\!-\!2})}{\VRandWTwo(\Numeral{2n\!+\!2})}} ))
\end{equation}
Since $2n-1>0$, a similar reduction to above shows that 
\begin{align*}
 \GInterpret{\VEverySnd(\Tail(\VRandW(\Numeral{2n\!-\!1})))}{} 
 & \peq (\ExpStep)^2\left(\GChoice{\frac12}{\GInterpret{\VEverySnd(\VRandW(\Numeral{2n\!-\!2}))}{}}{\GInterpret{\VEverySnd(\VRandW(\Numeral{2n}))}{}}\right)   \\
 \GInterpret{\VEverySnd(\Tail(\VRandW(\Numeral{2n\!+\!1})))}{} 
 & \peq (\ExpStep)^2\left(\GChoice{\frac12}{\GInterpret{\VEverySnd(\VRandW(\Numeral{2n}))}{}}{\GInterpret{\VEverySnd(\VRandW(\Numeral{2n\!-\!2}))}{}}  \right)
\end{align*}
so by Lemma~\ref{lem:prereqs:step-choice} (\ref{eq:randomwalk:1:app}) is equivalent to 
\[
  (\Later)^3
  \left( 
  \begin{aligned}
   &\GChoice{\frac12}{\GInterpret{\VEverySnd(\VRandW(\Numeral{2n}))}{}}{\left(\GChoice{\frac12}{\GInterpret{\VEverySnd(\VRandW(\Numeral{2n\!-\!2}))}{}}{\GInterpret{\VEverySnd(\VRandW(\Numeral{2n\!+\!2}))}{}}\right)} \\
  & \GLogRel\LazyL \CChoice{\frac{1}{2}}{\Eval(\VRandWTwo(\Numeral{2n}))}{\left(\CChoice{\frac{1}{2}}{\Eval(\VRandWTwo(\Numeral{2n\!-\!2}))}{\Eval(\VRandWTwo(\Numeral{2n\!+\!2}))}\right)} 
  \end{aligned}
  \right)
\]
which follows from the guarded induction hypothesis and Lemma~\ref{lem:prereqs:choice-lemma}. 

For the other direction, we will show $\GInterpret{\VRandWTwo(\Numeral{2n})}{} \GLogRel\LazyL \Eval({\VEverySnd(\VRandW(\Numeral{2n}))})$ by guarded recursion, focussing
again on the case of $n>0$. First unfold definitions on the left hand side:
\begin{align*}
 \GInterpret{\VRandWTwo(\Numeral{2n})}{} & = \ExpStep\GInterpret{G_{\RandWTwo}\, \VRandWTwo (\Numeral{2n})}{} \\
 & \peq \ExpStep(\GDirac(\nextop(\SemInr{\SemPair{\Numeral{2n}}{
 \lambda\_ .\GInterpret{{\Choice{\frac{1}{2}}{\VRandWTwo(\Numeral{2n})}{\Choice{\frac{1}{2}}{\VRandWTwo(\Numeral{2n\!-\!2)})}{\VRandWTwo(\Numeral{2n\!+\!2)}}}}}{}}}))) \\
 & \peq \ExpStep\left( \GDirac(\nextop(\SemInr{\SemPair{\Numeral{2n}}{\lambda\_ .\GChoice{\frac12}{\GInterpret{\VRandWTwo(\Numeral{2n})}{}}{\left( \GChoice{\frac12}{\GInterpret{\VRandWTwo(\Numeral{2n\!-\!2)})}{}}{\GInterpret{\VRandWTwo(\Numeral{2n\!+\!2)})}{}}  \right)} }}))  \right)
\end{align*}
and on the right:
\begin{align*}
 \Eval(\VRandW(\Numeral{2n})) & \leadsto \Eval(G_{\RandW}\, \VRandW (\Numeral{2n})) \\
 & \leadsto \Cons{\Numeral{2n}}{\Lam[y]{\Ifz{\Numeral{2n}}{\Nil}{\Choice{\frac{1}{2}}{\VRandW(\Numeral{2n}\!-\!1)}{\VRandW(\Numeral{2n}\!+\!1)}}}}
\end{align*}
so that 
\begin{align*}
 \Eval({\VEverySnd(\VRandW(\Numeral{2n}))}) 
 & \leadsto \Eval(G_{\EverySnd} \, \EverySnd \, (\VRandW(\Numeral{2n}))) \\
 & \leadsto \Eval(G_{\EverySnd} \, \VEverySnd \, (\VRandW(\Numeral{2n}))) \\
 & \leadsto \Cons{\Numeral{2n}}{\Lam[y]\VEverySnd(\Tail(\Lam[y]{\Ifz{\Numeral{2n}}{\Nil}{\Choice{\frac{1}{2}}{\VRandW(\Numeral{2n}\!-\!1)}{\VRandW(\Numeral{2n}\!+\!1)}}}\, \Star))}
\end{align*} 
and so the goal reduces to showing
\begin{equation}  \label{eq:randomwalk:2:app}
  \Later
  \left( 
  \begin{aligned}
   & \GChoice{\frac12}{\GInterpret{\VRandW(\Numeral{2n})}{}}{\left( \GChoice{\frac12}{\GInterpret{\VRandW(\Numeral{2(n\!-\!1)})}{}}{\GInterpret{\VRandW(\Numeral{2(n\!+\!1)})}{}}  \right)} 
   \GLogRel\LazyL \\
   & 
   \Eval(\VEverySnd(\Tail(\Lam[y]{\Ifz{\Numeral{2n}}{\Nil}{\Choice{\frac{1}{2}}{\VRandW(\Numeral{2n}\!-\!1)}{\VRandW(\Numeral{2n}\!+\!1)}}}\, \Star)))
  \end{aligned}
  \right)
\end{equation}
Now, 
\begin{align*}
 & \Eval(\Tail(\Lam[y]{\Ifz{\Numeral{2n}}{\Nil}{\Choice{\frac{1}{2}}{\VRandW(\Numeral{2n}\!-\!1)}{\VRandW(\Numeral{2n}\!+\!1)}}}\, \Star)) \\
 & \leadsto \Eval(\Tail(\Choice{\frac{1}{2}}{\VRandW(\Numeral{2n}\!-\!1)}{\VRandW(\Numeral{2n}\!+\!1)}))  \\
 & \leadsto \CChoice{\frac12}{\Eval(\Tail(\VRandW(\Numeral{2n}\!-\!1)))}{\Eval(\Tail(\VRandW(\Numeral{2n}\!+\!1)))} \\
 & \leadsto \CChoice{\frac12}{\Eval(\Choice{\frac{1}{2}}{\VRandW(\Numeral{2n\!-\!2})}{\VRandW(\Numeral{2n})})}{\Eval(\Choice{\frac{1}{2}}{\VRandW(\Numeral{2n})}{\VRandW(\Numeral{2n\!+\!2})})} \\
 & \peq \CChoice{\frac 12}{\Eval(\VRandW(\Numeral{2n}))}{\left( \CChoice{\frac 12}{\Eval(\VRandW(\Numeral{2n\!-\!2}))}{\Eval(\VRandW(\Numeral{2n\!+\!2}))}  \right)}
\end{align*}
so that 
\begin{align*}
 & \Eval(\VEverySnd(\Tail(\Lam[y]{\Ifz{\Numeral{2n}}{\Nil}{\Choice{\frac{1}{2}}{\VRandW(\Numeral{2n}\!-\!1)}{\VRandW(\Numeral{2n}\!+\!1)}}}\, \Star))) \\
 & \leadsto
 \CChoice{\frac 12}{\Eval(\VEverySnd(\VRandW(\Numeral{2n})))}{\left( \CChoice{\frac 12}{\Eval(\VEverySnd(\VRandW(\Numeral{2n\!-\!2})))}{\Eval(\VEverySnd(\VRandW(\Numeral{2n\!+\!2})))}  \right)}
\end{align*}
From this it follows that (\ref{eq:randomwalk:2:app}) can be proved using the guarded recursion assumption. 
\end{proof}

\end{document}